\documentclass[a4paper,USenglish,cleveref,nameinlink,numberwithinsect,thm-restate]{lipics-v2021}
\pdfoutput=1
\nolinenumbers
\usepackage{microtype}
\usepackage{framed}
\usepackage{xstring}
\makeatletter
\newcommand{\seq}[1]{
\saveexpandmode\expandarg
\StrDel{#1}{&}
\restoreexpandmode
} 
\makeatother

\makeatletter
\renewcommand\@oddfoot{
	\hfil
	\rlap{%
		\vtop{%
			\vskip10mm
			\colorbox[rgb]{0.99,0.78,0.07}
			{\@tempdima\evensidemargin
				\advance\@tempdima1in
				\advance\@tempdima\hoffset
				\hb@xt@\@tempdima{%
					\textcolor{darkgray}{\normalsize\sffamily
						\bfseries\quad\expandafter\textsolittle\expandafter{}}%
					\vphantom{$0abcd1$}\hss}}}}
}
\makeatother

\title{\texorpdfstring{From Chinese Postman to Salesman and Beyond II:\\ Inapproximability and Parameterized Complexity}{From Chinese Postman to Salesman and Beyond: Inapproximability and Parameterized Complexity}}
\titlerunning{From Chinese Postman to Salesman and Beyond II}

\author{Fabian Frei}{CISPA Helmholtz Center for Information Security, Saarbrücken}{fabian.frei@cispa.de}{https://orcid.org/0000-0002-1368-3205}{}
\author{Ahmed Ghazy}{CISPA Helmholtz Center for Information Security and Saarland University, Saarbrücken}{ahmed.ghazy@cispa.de}{https://orcid.org/0009-0009-7414-5871}{}
\author{Tim A.~Hartmann}{CISPA Helmholtz Center for Information Security, Saarbrücken}{tim.hartmann@cispa.de}{https://orcid.org/0000-0002-1028-6351}{}
\author{Florian Hörsch}{CISPA Helmholtz Center for Information Security, Saarbrücken}{florian.hoersch@cispa.de}{https://orcid.org/0000-0002-5410-613X}{}
\author{Dániel Marx}{CISPA Helmholtz Center for Information Security, Saarbrücken}{marx@cispa.de}{https://orcid.org/0000-0002-5686-8314}{}

\authorrunning{F.~Frei, A.~Ghazy, T.~A.~Hartmann, F.~Hörsch, D.~Marx}

\Copyright{Fabian Frei, Ahmed Ghazy, Tim A.~Hartmann, Florian Hörsch, Dániel Marx}

\ccsdesc[500]{Theory of computation~Parameterized complexity and exact algorithms}

\keywords{Traveling Salesman Problem on Cubic Bipartite Graphs, Continuous Graphs, Inapproximability, Parameterized Complexity}

\category{}

\relatedversion{}
\usepackage[utf8]{inputenc}
\usepackage[T1]{fontenc}
\usepackage{latexsym}

\usepackage{amssymb}
\usepackage{bibentry}
\usepackage{amsmath}
\usepackage{amsthm}
\usepackage{thmtools}
\usepackage{amsthm}
\usepackage{hyperref}
\usepackage[all]{hypcap}
\usepackage{mathtools}
\usepackage{xspace} 
\usepackage{xparse}
\usepackage{dsfont} 
\usepackage{tabularx} 
\usepackage{adjustbox}
\usepackage{changepage}
\usepackage{subcaption}
\usepackage{hyperref}
\usepackage[textsize=footnotesize,backgroundcolor=white]{todonotes}

\usepackage{algorithm}
\DeclareCaptionLabelFormat{nonumber}{%
  \kern0.05em{\color[rgb]{0.99,0.78,0.07}\rule{0.73em}{0.73em}}%
  \hspace*{0.67em}\bothIfFirst{#1}{~}} 

\let\oldtextsc\textsc
\renewcommand{\textsc}[1]{\mbox{\oldtextsc{#1}}}

\newtheorem*{theorem*}{Theorem}
\newtheorem{Case}{Case}
\crefname{claim}{Claim}{Claims}

\makeatletter
\newcommand{\xlabel}[2][]{\phantomsection\def\@currentlabelname{\ifthenelse{\equal{#1}{}}{#2}{#1}}\label{#2}}
\makeatother

\usepackage{tikz}
\usetikzlibrary{decorations.pathmorphing}
\usetikzlibrary{decorations.markings}
\usetikzlibrary{decorations,arrows.meta}
\pgfdeclaremetadecoration{many arrows}{initial}{ 
\state{initial}[width=0pt, next state=arrow] {
    \pgfmathdivide{100}{\pgfmetadecoratedpathlength}
    
    \pgfset{/pgf/decoration/segment length=4pt}
  }
  \state{arrow}[
switch if less than=\pgfmetadecorationsegmentlength to final, width=\pgfmetadecorationsegmentlength/3,
next state=end arrow]
  {
    \decoration{curveto}
    \beforedecoration
    {
      \pgfpathmoveto{\pgfpointmetadecoratedpathfirst}
} }
\state{end arrow}[width=\pgfmetadecorationsegmentlength/3, next state=move] {
    \decoration{curveto}
    \beforedecoration{\pgfpathmoveto{\pgfpointmetadecoratedpathfirst}}
    \afterdecoration
    {
      \pgfsetarrowsend{Latex[length=1pt,width=1pt]}
      \pgfusepath{stroke}
    }
}
\state{move}[width=\pgfmetadecorationsegmentlength/2, next state=arrow]{} \state{final}{}
}

\usetikzlibrary{shapes}
\usetikzlibrary{calc}
\usetikzlibrary{patterns}
\usetikzlibrary{fit}
\usepackage{pgfplots}

\newcommand{\p}{\textup{P}\xspace}
\newcommand{\np}{\textup{NP}\xspace}
\newcommand{\fpt}{\textup{FPT}\xspace}
\newcommand{\xp}{\textup{XP}\xspace}
\newcommand{\apx}{\textup{APX}\xspace}
\newcommand{\wone}{\textup{W[1]}\xspace}
\newcommand{\wtwo}{\textup{W[2]}\xspace}

\newcommand{\tail}{\textup{tail}\xspace}

\newcommand{\Oh}{\mathcal{O}}
\renewcommand{\Oh}{O}

\newcommand{\oh}{o}

\newcommand{\Ce}{\mathcal{C}\xspace}

\newcommand\ETH{\texorpdfstring{$\mathsf{ETH}$}{ETH}\xspace}

\newcommand{\TSP}{\textup{TSP}\xspace}

\newcommand{\abs}[1]{\lvert #1\rvert}
\renewcommand{\setminus}{-}

\newcommand{\ceil}[1]{\ensuremath{\lceil #1 \rceil}}
\newcommand{\floor}[1]{\ensuremath{\lfloor #1 \rfloor}}

\newcommand{\half}{\ensuremath{{\tfrac{1}{2}}}}

\def\eps{\ensuremath{\varepsilon}}
\renewcommand{\epsilon}{\varepsilon}
\def\phi{\ensuremath{\varphi}}

\newcommand{\tour}{\textsc{Tour}\xspace}

\renewcommand{\textsc}[1]{\textnormal{\scshape #1}}
\newcommand{\VC}{\textsc{Vertex}\-\textsc{Cover}\xspace}
\newcommand{\domset}{\textsc{Dominating}\-\textsc{Set}\xspace}

\newcommand{\covering}{\textsc{Covering}\xspace}

\newcommand{\CyclePartition}{\textsc{Cycle}\-\textsc{Partition}\xspace}
\newcommand{\CycleSubpartition}{\textsc{Cycle}\-\textsc{Subpartition}\xspace}

\newcommand{\CPP}{\textsc{Chinese}\-\textsc{Postman}\-\textsc{Problem}\xspace}
\newcommand{\CSPsubscript}{\textup{CSP}}
\newcommand{\binaryCSP}{\textsc{Binary}\-\textsc{CSP}\xspace}

\usepackage{tabularx, environ}

\makeatletter
\newcolumntype{\expand}{}
\long\@namedef{NC@rewrite@\string\expand}{\expandafter\NC@find}

\newcounter{myproblem}
\newcounter{temp}
\NewEnviron{myproblem}[1][]{%
\setcounter{temp}{\value{algorithm}}%
\setcounter{algorithm}{\value{myproblem}}%
\captionsetup[algorithm]{name=Optimization Problem, labelformat=nonumber, position=top, listformat=empty}%
  \noindent%
\begin{algorithm}[H]\caption{#1}%
\refstepcounter{myproblem}%
\vspace{1ex}%
\begin{tabularx}{\textwidth}{>{\bfseries}lXc}%
	\phantomsection%
    \BODY
  \end{tabularx}%
  \smallskip
\end{algorithm}
\setcounter{algorithm}{\value{temp}}
\LinkTargetOn
}

\crefname{algorithm}{Algorithm}{Algorithms}
\crefname{myproblem}{Problem}{Problems}
\Crefname{algorithm}{Algorithm}{Algorithms}
\Crefname{myproblem}{Problem}{Problems}

\makeatother

\NewEnviron{framedproblem}[1][]{%
  \noindent%
  
  \vspace{1ex}
  \begin{tabularx}{\textwidth}{|>{\bfseries}lX|c|}%
	\hline
    \BODY\\\hline
  \end{tabularx}%
  \smallskip
}
\makeatother

\newcommand{\dist}{\operatorname{{d}}}
\newcommand{\len}{\operatorname{\mathsf{\ell}}}

\newcommand{\opt}{\ensuremath{\mathsf{OPT}}}

\newcommand{\opttour}[1][\delta]{\ensuremath{\opt_{#1\textup{-tour}}}\xspace}

\newcommand{\deltatour}[1][\delta]{{\ensuremath{{#1\textup{-tour}}}}\xspace}
\newcommand{\deltatourprob}[1][\delta]{{\ensuremath{{#1\textup{-\textsc{Tour}}}}}\xspace}

\newcommand{\assig}[1]{\ensuremath{\mathcal{A}(#1)}}
\newcommand{\conn}[4]{\ensuremath{\ifthenelse{\equal{#4}{0}}{#1_{#2}^{#3}}{\bar{#1}_{#2}^{#3}}}}
\newcommand{\cand}[2]{\ensuremath{a_{#1}^{#2}}}
\newcommand{\midpath}[2]{\ensuremath{\ifthenelse{\equal{#2}{0}}{M_{#1}}{\bar{M}_{#1}}}}
\newcommand{\ball}[1]{\ensuremath{B_{G'}(#1,\delta)}}

\input{restatable-helpers_2.tex}

\hideLIPIcs

\begin{document}
\maketitle
\begin{abstract}

A well-studied continuous model of graphs, introduced by Dearing and Francis [Transportation Science, 1974],
considers each edge as a continuous unit-length interval of points.
In the problem \(\delta\)-\tour defined within this model, 
the objective to find a shortest tour that comes within a distance of \(\delta\) of every point on every edge. 
This problem was introduced in the predecessor to this article and shown to be essentially equivalent to the Chinese Postman problem for $\delta=0$, to the graphic Travel Salesman Problem (TSP) for $\delta=1/2$, and close to first Vertex Cover and then Dominating Set for even larger $\delta$. Moreover, approximation algorithms for multiple parameter ranges were provided.
In this article, we provide complementing inapproximability bounds and examine the  fixed-parameter tractability of
the problem. On the one hand, we show the following:
\begin{enumerate}
  \item[(1)] For every fixed \(0 < \delta < 3/2\), the problem \(\delta\)-\tour is APX-hard, while for every fixed \(\delta \geq 3/2\), the problem has no polynomial-time \(o(\log n)\)-approximation unless $\p=\np$.
\end{enumerate}
Our techniques also yield the new result that TSP remains \apx-hard on cubic (and even cubic bipartite) graphs. 
\begin{enumerate}
  \item[(2)] 
For every fixed \(0 < \delta < 3/2\), the problem  \(\delta\)-\tour is fixed-parameter tractable (FPT) when parameterized by the length of a shortest tour, while it is W[2]-hard for every fixed \(\delta \geq 3/2\) and para-\np-hard for $\delta$ being part of the input.
  \end{enumerate}
  On the other hand, if $\delta$ is considered to be part of the input, then an interesting nontrivial phenomenon occurs when $\delta$ is a constant fraction of the number of vertices:
  \begin{enumerate}
  \item[(3)]
If $\delta$ is part of the input, then the problem can be solved in time \(f(k)n^{\Oh(k)}\), where \(k = \lceil n/\delta \rceil\); however, assuming the Exponential-Time Hypothesis (ETH), there is no algorithm that solves the problem and runs in time \(f(k)n^{o(k/\log k)}\).
\end{enumerate}
\bigskip
\end{abstract}
\newpage

\section{Introduction}

We continue our previously initiated study~\cite{FreiGHHM24} of the problem \deltatourprob. 
It is defined in the well-studied continuous graph model introduced by Dearing and Francis~\cite{Dearing1974}, where each edge is seen as a continuous unit interval of points with its vertices as endpoints. 

For any given graph~$G$, this yields a compact metric space $(P(G),\dist)$ with a point set $P(G)$ and a distance function $\dist\colon P(G)^2\to \mathbb{R}_{\ge0}$.
Several natural continuous generalizations of standard problems on graphs
have been studied; we refer the reader to our previous article~\cite{FreiGHHM24}
for an overview of key results in the continuous model.

Within our model, the studied tours are \emph{continuous}, that is, not just closed walks on vertices
but closed walks in the graph's continuous point space.
Such a tour may travel from one vertex to the next as usual but can also make U-turns at arbitrary points inside the edges.
Roughly speaking, a continuous tour can be described by a closed sequence of points in the graph, where every pair of consecutive points represents a segment of the tour.
The formal definition is found in \cref{section:prel}.

A \emph{\deltatour} $T$ then is a tour
	that is $\delta$-covering,
	which means that every point in the graph is within distance at most $\delta$ from at least one point contained in the tour $T$.
The task in our problem \deltatourprob is to find a shortest \deltatour.
We observe that computing a shortest $0$-tour is equivalent to computing a shortest Chinese Postman tour (a closed walk going through every edge),
	which is known to be polynomial-time solvable~\cite[Chapter~29]{schrijver-book}.
	Moreover, one can observe that if every vertex of the input graph has degree at least two,
	then there is a shortest $1/2$-tour that stops at every vertex and, conversely,
	any tour stopping at every vertex is a $1/2$-tour.
	Thus, finding a shortest $1/2$-tour is essentially equivalent to
	solving a TSP problem on a graph, with some additional careful handling of degree-1 vertices.

The problem can be shown to be \np-hard for all $\delta > 0$. Our previous article therefore provided polynomial-time approximation algorithms.

\begin{theorem}[\textbf{Approximation Algorithms~\cite{FreiGHHM24}}]
	For every fixed $\delta \in (0,3/2)$, the problem \deltatourprob admits a constant-factor approximation algorithm.
	Moreover, for every fixed $\delta \geq 3/2$, the problem admits an $\Oh(\log n)$-approximation algorithm.
\end{theorem}

We extend these previous results by showing approximation lower bounds.

Further, we study the parameterized complexity of \deltatourprob.
As we show, for the length of the tour as the parameter, the problem is \fpt for every $\delta < 3/2$,
	and \wtwo-hard for every $\delta \geq 3/2$.
Notably, there is a jump in the complexity at the threshold of $3/2$,
	similarly as we observe for approximability.
This jump in the complexity also occurs
	for the similar task to $\delta$-cover all points on all edges
	not with a tour but with a minimum number of selected points in $P(G)$.
There, also $3/2$ marks the threshold where the complexity jumps from \fpt to \wtwo-hard~\cite{HartmannLW22}. 
Similarly, at $3/2$, the approximability of finding a $\delta$-cover jumps from
	allowing a polynomial-time constant-factor approximation to Log-\apx-hard~\cite{HartmannJanssen2024}.

Additionally, we point out an interesting phenomenon where the problem can be solved
much faster as $\delta$ gets very large relative to the number of vertices $n$.
More precisely, we also study the parameterization by $n/\delta$.

\subparagraph{Our Results.}

Complementing our previous positive approximation results~\cite{FreiGHHM24},
	we provide corresponding lower bounds.
Specifically, we first prove that, for each $\delta < 3/2$, computing a shortest \deltatour is \apx-hard.

\begin{theorem}[\textbf{\apx-Hardness}]\label{thm:approxmain}
For every fixed $\delta \in (0,3/2)$, the problem \deltatourprob is \apx-hard.
\end{theorem}          
        
Particularly, as an intermediate step for \apx-hardness for small values of $\delta$,
	we prove \apx-hardness of a versatile cycle covering problem.
Doing so, we also obtain \apx-hardness for \emph{TSP on cubic bipartite graphs}.
	To the best of our knowledge, even for TSP only restricted to cubic graphs,  
	there is only one
	\apx-hardness result that unfortunately happens to be flawed~\cite[Thm.~5.4]{KarpinskiS15}. In particular,
	the proposed tour reconfiguration appears to split the original TSP tour into two disjoint ones, rendering the argument invalid. 
	The issue seems to affect results in a series of other articles~\cite{EngebretsenK01,EngebretsenK06,KarpinskiS12,KarpinskiS13,KarpinskiS13website,Karpinski15FCT}.
	Fortunately, our separate approach closes the gap and yields the hardness on an even more restricted class.

\begin{theorem}\label{cor:tsphardnessmain}
	\TSP is \apx-hard even on cubic bipartite graphs.
\end{theorem}

Once $\delta$ reaches $3/2$, the problem \deltatourprob suddenly changes character: it becomes similar to \domset,
	where only a logarithmic-factor approximation is possible, unless $\p = \np$; see \cref{thm:approx:lb:log:n}.

\begin{theorem}[\textbf{Log-APX-Hardness}]\label{thm:logmain}
	For every fixed $\delta \geq 3/2$, the problem \deltatourprob has no
	$o(\log n)$-approximation algorithm unless $\p=\np$.
\end{theorem}
We then study the problem parameterized by the solution size,
which is the length of the \deltatour.
As mentioned above, some problems exhibit a jump in parameterized complexity
at a certain threshold value of the covering range $\delta$.
Notably, the $\delta$-\covering problem, introduced by Shier~\cite{Shier1977},
where the task is to find a minimum set of points that $\delta$-covers
the entire graph, becomes similar to \domset and is \wtwo-hard when $\delta \geq 3/2$.
Therefore, it is not very surprising that 
	$\delta = 3/2$ marks the threshold for the parameterized complexity of \deltatour as well; see \cref{section:overview:parameterized}. 
\begin{theorem}[\textbf{Natural Parameterization}]\label{thm:parammain}
	Computing a shortest length \(\delta\)-tour, parameterized by the length of the tour, is \fpt for every fixed \(0 < \delta < 3/2\), and \wtwo-hard for every fixed \(\delta \geq 3/2\).
\end{theorem}
In this context, it is natural to ask whether in the cases in which \wtwo-hardness has been obtained, we can at least design algorithms correctly solving the problem when further weakening the requirements on the running time. Here, it turns out to be crucial whether $\delta$ is fixed or part of the input. 
\begin{theorem}[\textbf{Classic Complexity for Fixed Tour Length}]\label{thm:parammain2}
	Computing a shortest length \(\delta\)-tour, parameterized by the length of the tour, is \xp for every fixed $\delta\geq 0$, and \np-hard if $\delta$ is part of the input, even if $\delta \leq 2$.
\end{theorem}
It is much more surprising what happens if $\delta$ is really large,
	say, comparable to the number of vertices.
We consider the problem of computing a shortest \deltatour when $\delta$ is part of the input.
For large $\delta$,
	the problem becomes similar to covering the whole graph with  $k=\ceil{\frac{n}{\delta}}$ balls of radius $\delta$, suggesting the problem to be solvable in large part by guessing $k$ centers in $n^{O(k)}$ time.
Indeed, we give an algorithm for computing a shortest \deltatour in this running time.

\begin{restatable}[XP Algorithm for Parameter \boldmath$n/\delta$]{linkedtheorem}{ThmParamKUbLargeDelta}
\label{ThmParamKUbLargeDelta}
\label{thm:param:k:ub:largedelta}
There is an algorithm, which, given a connected $n$-vertex graph~$G$, 
	computes a shortest \deltatour of $G$ in $f(k) \cdot n^{\Oh(k)}$ time
	where $k=\lceil n/\delta\rceil$. 
\end{restatable}

On the other hand, we show the exponent in the running time above to be essentially
optimal, ruling out an FPT-time algorithm for this choice of parameter.

\begin{restatable}[Hardness for Parameter \boldmath$n/\delta$]{linkedtheorem}{ThmParamKLbLargeDelta}
\label{ThmParamKLbLargeDelta}
\label{thm:param:k:lb:largedelta}
There are constants $\alpha>0$ and $k_0$ such that,
	unless \ETH fails,
	for every $k \geq k_0$,
	there is no algorithm that, given an $n$-vertex graph $G$ and a constant $K$,
	decides whether $G$ admits a \deltatour of length at most $K$
	in $\Oh(n^{\alpha k/\log{k}})$ time where $k = \ceil{n/\delta}$.
Moreover, the problem is \wone-hard parameterized by~$k$.
\end{restatable}

\Cref{section:prel} begins with formal notions including a thorough definition of a \deltatour.
Then \cref{section:overview} gives an extended overview of our results.
In \cref{section:approximation-lb}, we prove our inapproximability
results. Finally, the parameterized complexity of the problem is studied in \cref{section:param:complexity}.

\section{Formal Definitions and Preliminaries}
\label{section:prel}

\subparagraph{General Definitions.}
For a positive integer $n$, we denote the set $\{1, \dots, n\}$ by $[n]$.
All graphs in this article are undirected, unweighted and do not contain parallel edges or loops. Let $G$ be a graph.
For a subset of vertices $V' \subseteq V(G)$, we denote by $G[V']$ the subgraph induced by $V'$.
The neighborhood of a vertex $u$ is $N_G(u) \coloneqq \{ v \in V(G) \mid uv \in E(G)\}$. 
We write $uv$ for an edge $\{u,v\} \in E(G)$.
We denote by $\log$ the binary logarithm. 

\subparagraph{Problem Related Definitions.}
For a graph~$G$, we define a metric space whose 
	point set $P(G)$ contains, somewhat informally speaking, all points on the continuum of each edge, which has unit length. 
We use the word \emph{vertex} for the elements in $V(G)$, 
	whereas we use the word \emph{point} to denote elements in $P(G)$. Note however, that each vertex of $G$ is also a point of $G$.

More formally, the set $P(G)$ consists of the set of points $p(u,v,\lambda)$ for every edge $uv \in E(G)$ and every $\lambda \in [0,1]$
	where $p(u,v,\lambda)=p(v,u,1-\lambda)$; $p(u,v,0)$ coincides with $u$
	and $p(u,v,1)$ coincides with $v$.
The \emph{edge segment} $P(p,q)$ of $p$ and $q$ then is the subset of points
	$\{ p(u,v,\mu) \mid \min\{\lambda_p,\lambda_q\} \leq \mu \leq \max\{\lambda_p,\lambda_q\} \}$.
A $\seq{p&q}$-\emph{walk} $T$ between points $p_0 \coloneqq  p$ and $p_z \coloneqq  q$
	is a finite sequence of points $\seq{p_0&p_1&\dots&p_z}$
	where every two consecutive points are distinct and lie on the same edge,
	that is, formally, for every $i\in[z]$,
	there are an edge $u_i v_i \in E(G)$ and $\lambda_i,\mu_i \in [0,1]$
	such that $p_{i-1}=p(u_i,v_i,\lambda_i)$ and $p_{i}=p(u_i,v_i,\mu_i)$.
When $p$ and $q$ are not specified, we may simply write \emph{walk} instead of $\seq{p&q}$-walk.

For two points $p,q \in P(G)$, their {\it distance} $\dist(p,q)$ is defined as follows. If $p$ and $q$ are on the same edge $uv$,
say $p=p(u,v,\lambda_p)$ and $q=p(u,v,\lambda_q)$, we define  $\dist(p,q)=\abs{\lambda_q-\lambda_p}$. In order to extend this definition to arbitrary point pairs, we first need the definition of the length of a walk. 
Namely, the \emph{length} $\len(T)$ of a walk $T=\seq{p_0&p_1&\dots&p_z}$ is $\sum_{i\in[z]}\dist(p_{i-1},p_i)$.
A $\seq{p&q}$-walk with minimum length among all $\seq{p&q}$-walks is called \emph{shortest}. We can now extend the notion of distance to arbitrary point pairs.
The \emph{distance between two points} $p,q \in P(G)$, denoted $\dist(p,q)$,
	is either the length of a shortest $\seq{p&q}$-walk or, if no such walk exists, $\infty$ . It is easy to see that both definitions coincide for points located on the same edge.
Further, let $\dist(p,Q) = \inf\{\dist(p,q) \mid q \in Q \}$ for any point $p \in P(G)$ and point set $Q \subseteq P(G)$.
For some $p\in P(G)$ and a walk $T$, we write $\dist(p,T)$ as abbreviation for $\dist(p,P(T))$.
We later show that $\dist(p, T)$ is in fact a minimum taken over the set of stopping points of $T$ (see \Cref{nearstop}).

A \emph{tour} $T$ is a $\seq{p_0&p_z}$-walk with $p_0=p_z$.
For a real $\delta>0$, a $\delta$-\emph{tour} is a tour
	where $\dist(p,T) \leq \delta$ for every point $p \in P(G)$.
We study the following minimization problem.

\begin{myproblem}[\deltatourprob, where $\delta \geq 0$]
\label{prob:deltatour}%
Instance&A connected simple graph $G$.\\
Solution&Any \deltatour $T$.\\
Goal& Minimize the length $\len(T)$.
\end{myproblem}

Further, we use the following notions for a tour $T=\seq{p_0&p_1&\dots&p_z}$.
A \emph{tour segment} of $T$ is a walk
	given by a contiguous subsequence of $\seq{p_0&p_1&\dots&p_z}$.
The tour $T$ \emph{stops} at a point $p \in P(G)$ if $p \in \{p_0,p_1,\dots,p_z\}$ 
and \emph{traverses} an edge $uv$ if $uv$ or $vu$ is a tour segment of $T$.
The \emph{discrete length} of a tour is $z$, that is, the length 
of the finite sequence of points representing it. We denote the discrete length
of a tour $T$ by $\alpha(T)$.

A point $p\in P(G)$ is \emph{integral} if it coincides with a vertex.
Similarly, $p=p(u,v,\lambda)$ is \emph{half-integral} if $\lambda \in \{0,\half,1\}$. A tour is {\it integral} ({\it half-integral}) if all its stopping points are integral (half-integral).

The \emph{truncation} of a tour $T$, denoted as $\floor{T}$,
	is the integral tour where for every edge $uv \in E(G),\lambda<1$,
	every tour segment~$\seq{u&p(u, v, \lambda)&u}$ in $T$ is replaced by $u$.
We note that $P(\floor{T}) \subseteq P(T)$.

\subparagraph{Structural Results.}
As continuous TSP was studied in the predecessor of this article~\cite{FreiGHHM24} for the first time, 
it was necessary to lay plenty of groundwork. In what follows, we summarize the main structural results we make use of here.

\newcommand{\textDefNiceTour}{%
a tour $T=\seq{p_0&p_1&\dots&p_z}$ in a connected graph $G$ with $z \geq 3$ is \emph{nice} if
\begin{itemize}
\item $\{p_{i-1},p_i\}\cap V(G)\neq \emptyset$ for every $i \in [z]$
(i.e., $T$ has no two consecutive stopping points inside an edge),

\item
for every $i \in [z]$ with $p_i \notin V(G)$, we have $p_{i-1} = p_{(i+1)\bmod z}$\\
(i.e., whenever $T$ stops inside an edge, the previous and next stopping point are the same)
\item for every edge $uv \in E(G)$, there is at most one index $i \in [z]$ with
	$p_i \in \{ p(u,v,\lambda) \mid \lambda \in (0,1)\}$\\
(i.e., $T$ stops at most once inside any given edge),
\item if $T$ traverses an edge $uv \in E(G)$, there is no index $i \in [z]$ with $p_i \in \{ p(u,v,\lambda) \mid \lambda \in (0,1)\}$\\
(i.e., $T$ does not stop inside traversed edges), and 
\item every edge $uv \in E(G)$ is traversed at most twice by $T$.
\end{itemize}
}

A \emph{nice tour} is defined (see~\cite{FreiGHHM24}) as a tour
that interacts with the interior of edges in one of four desirable ways.
More precisely, \textDefNiceTour
See \cref{fig:fourbehaviors} for an illustration.

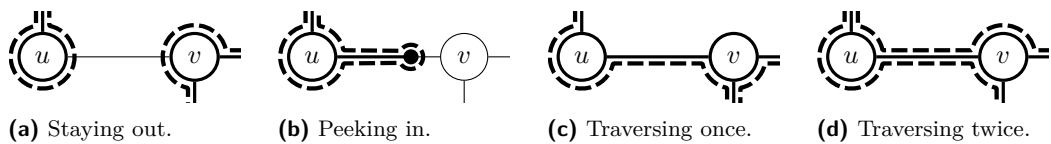
\begin{figure}[H]
\centering
\begin{subfigure}{.24\textwidth}
\begin{tikzpicture}[scale=2,rotate=0,
		vertex/.style={draw,circle,inner sep=.0pt,minimum size=0.62cm},]
    \clip(-.22,-.3) rectangle (1.3,.3);
	\coordinate (cx) at (0,1);
	\coordinate (cy) at ($(1,0)+(0:1)$);
	\coordinate (cz) at ($(1,0)+(-90:1)$);
	\coordinate (cu) at (0,0);
	\coordinate (cv) at (1,0);
	\node[vertex,very thick] (u) at (cu) {$u$};
	\node[vertex,very thick] (v) at (cv) {$v$};
	\node[vertex] (x) at (cx) {$x$};
	\node[vertex] (y) at (cy) {$y$};
	\node[vertex] (z) at (cz) {$z$};
\newcommand{\vertexangle}{12}
\newcommand{\vertexdistance}{.21}
\newcommand{\midpointangle}{30}
\newcommand{\midpointdistance}{.08}
	\draw[very thick] (x)--(u);
	\draw[very thick] (y)--(v)--(z);
	\draw (u)--(v)--(y)--(z)--(v);
	\draw [ultra thick,dash pattern={on 6pt off 2pt}, rounded corners=.5] 
	([shift={(-90+\vertexangle:\vertexdistance)}]cx) arc (-90+\vertexangle:360-90-\vertexangle:\vertexdistance)--([shift={(90+\vertexangle:\vertexdistance)}]cu) arc (90+\vertexangle:450-\vertexangle:\vertexdistance)--cycle;
	\draw [ultra thick,dash pattern={on 6pt off 2pt}, rounded corners=.5] 
	([shift={(-180+\vertexangle:\vertexdistance)}]cy) arc (-180+\vertexangle:180-\vertexangle:\vertexdistance)--([shift={(\vertexangle:\vertexdistance)}]cv) arc (\vertexangle:270-\vertexangle:\vertexdistance)--([shift={(90+\vertexangle:\vertexdistance)}]cz) arc (90+\vertexangle:90-\vertexangle:\vertexdistance);
	\end{tikzpicture}
\caption{Staying out.}
\end{subfigure}
\hfill
\begin{subfigure}{.24\textwidth}
\begin{tikzpicture}[scale=2,rotate=0,
		vertex/.style={draw,circle,inner sep=.0pt,minimum size=0.62cm},]
    \clip(-.22,-.3) rectangle (1.3,.3);
	\coordinate (cx) at (0,1);
	\coordinate (cy) at ($(1,0)+(0:1)$);
	\coordinate (cz) at ($(1,0)+(-90:1)$);
	\coordinate (cu) at (0,0);
	\coordinate (cv) at (1,0);
	\coordinate (cm) at (.65,0);
	\node[vertex,very thick] (u) at (cu) {$u$};
	\node[vertex] (v) at (cv) {$v$};
	\node[vertex] (x) at (cx) {$x$};
	\node[vertex] (y) at (cy) {$y$};
	\node[vertex] (z) at (cz) {$z$};
	\node[minimum size=0.2cm,inner sep=0,outer sep=-1mm,circle,fill=black]  (m) at (cm) {};
\newcommand{\vertexangle}{12}
\newcommand{\vertexdistance}{.21}
\newcommand{\midpointangle}{30}
\newcommand{\midpointdistance}{.08}
	\draw[very thick] (x)--(u)--(m);
	\draw (u)--(v)--(y)--(z)--(v);
	\draw [ultra thick,dash pattern={on 6pt off 2pt}, rounded corners=.5] 
	([shift={(-90+\vertexangle:\vertexdistance)}]cx) arc (-90+\vertexangle:360-90-\vertexangle:\vertexdistance)--([shift={(90+\vertexangle:\vertexdistance)}]cu) arc (90+\vertexangle:360-\vertexangle:\vertexdistance)--([shift={(-180+\midpointangle:\midpointdistance)}]cm) arc (-180+\midpointangle:180-\midpointangle:\midpointdistance)--([shift={(\vertexangle:\vertexdistance)}]cu) arc (\vertexangle:90-\vertexangle:\vertexdistance)--cycle;
	\end{tikzpicture}
\caption{Peeking in.}
\end{subfigure}
\hfill
\begin{subfigure}{.24\textwidth}
\begin{tikzpicture}[scale=2,rotate=0,
		vertex/.style={draw,circle,inner sep=.0pt,minimum size=0.62cm},]
    \clip(-.22,-.3) rectangle (1.3,.3);
	\coordinate (cx) at (0,1);
	\coordinate (cy) at ($(1,0)+(0:1)$);
	\coordinate (cz) at ($(1,0)+(-90:1)$);
	\coordinate (cu) at (0,0);
	\coordinate (cv) at (1,0);
	\node[vertex,very thick] (u) at (cu) {$u$};
	\node[vertex,very thick] (v) at (cv) {$v$};
	\node[vertex] (x) at (cx) {$x$};
	\node[vertex] (y) at (cy) {$y$};
	\node[vertex] (z) at (cz) {$z$};
\newcommand{\vertexangle}{12}
\newcommand{\vertexdistance}{.21}
\newcommand{\midpointangle}{30}
\newcommand{\midpointdistance}{.08}
	\draw[very thick] (x)--(u)--(v);
	\draw[very thick] (y)--(v)--(z);
	\draw (u)--(v)--(y)--(z)--(v);
	\draw [ultra thick,dash pattern={on 6pt off 2pt}, rounded corners=.5] 
	([shift={(-90+\vertexangle:\vertexdistance)}]cx) arc (-90+\vertexangle:360-90-\vertexangle:\vertexdistance)--([shift={(90+\vertexangle:\vertexdistance)}]cu) arc (90+\vertexangle:360-\vertexangle:\vertexdistance)--([shift={(180+\vertexangle:\vertexdistance)}]cv) arc (180+\vertexangle:270-\vertexangle:\vertexdistance)--([shift={(90+\vertexangle:\vertexdistance)}]cz) arc (90+\vertexangle:90-\vertexangle:\vertexdistance);
	\draw [ultra thick,dash pattern={on 6pt off 2pt}, rounded corners=.5] 
	([shift={(90-\vertexangle:\vertexdistance)}]cz)--([shift={(-90+\vertexangle:\vertexdistance)}]cv) arc (-90+\vertexangle:-\vertexangle:\vertexdistance)--([shift={(-180+\vertexangle:\vertexdistance)}]cy);
	\end{tikzpicture}
\caption{Traversing once.}
\end{subfigure}
\hfill
\begin{subfigure}{.24\textwidth}
\begin{tikzpicture}[scale=2,rotate=0,
		vertex/.style={draw,circle,inner sep=.0pt,minimum size=0.62cm},]
    \clip(-.22,-.3) rectangle (1.3,.3);
	\coordinate (cx) at (0,1);
	\coordinate (cy) at ($(1,0)+(0:1)$);
	\coordinate (cz) at ($(1,0)+(-90:1)$);
	\coordinate (cu) at (0,0);
	\coordinate (cv) at (1,0);
	\node[vertex,very thick] (u) at (cu) {$u$};
	\node[vertex,very thick] (v) at (cv) {$v$};
	\node[vertex] (x) at (cx) {$x$};
	\node[vertex] (y) at (cy) {$y$};
	\node[vertex] (z) at (cz) {$z$};
\newcommand{\vertexangle}{12}
\newcommand{\vertexdistance}{.21}
\newcommand{\midpointangle}{30}
\newcommand{\midpointdistance}{.08}
	\draw[very thick] (x)--(u)--(v);
	\draw[very thick] (y)--(v)--(z);
	\draw (u)--(v)--(y)--(z)--(v);
	\draw [ultra thick,dash pattern={on 6pt off 2pt}, rounded corners=.5] 
	([shift={(-90+\vertexangle:\vertexdistance)}]cx) arc (-90+\vertexangle:360-90-\vertexangle:\vertexdistance)--([shift={(90+\vertexangle:\vertexdistance)}]cu) arc (90+\vertexangle:360-\vertexangle:\vertexdistance)--([shift={(180+\vertexangle:\vertexdistance)}]cv) arc (180+\vertexangle:270-\vertexangle:\vertexdistance)--([shift={(90+\vertexangle:\vertexdistance)}]cz) arc (90+\vertexangle:90-\vertexangle:\vertexdistance);
	\draw [ultra thick,dash pattern={on 6pt off 2pt}, rounded corners=.5] 
	([shift={(-180+\vertexangle:\vertexdistance)}]cy) arc (-180+\vertexangle:180-\vertexangle:\vertexdistance)--([shift={(\vertexangle:\vertexdistance)}]cv) arc (\vertexangle:180-\vertexangle:\vertexdistance)--([shift={(\vertexangle:\vertexdistance)}]cu) arc (\vertexangle:90-\vertexangle:\vertexdistance)--([shift={(-90+\vertexangle:\vertexdistance)}]cx) arc (-90+\vertexangle:360-90-\vertexangle:\vertexdistance);
	\end{tikzpicture}
\caption{Traversing twice.}
\end{subfigure}
\caption{The four ways a nice \deltatour defined by at least $3$ points can interact with an edge $uv$.}
\label{fig:fourbehaviors}
\end{figure}

The following result regarding nice tours often allows us to restrict our search space.

\begin{lemma}[Nice Tours~\cite{FreiGHHM24}]\label{lem:tournice}
Let $G$ be a connected graph. Further, let a tour~$T$ in $G$ be given.
Then, in polynomial time, we can compute a tour~$T'$ in $G$ with $\len(T')\leq \len(T)$, such that,
$T'$ is either nice or has at most two stopping points, and if $T$ is a
\deltatour for some $\delta \geq 0$, then so is $T'$.
\end{lemma}

In a similar vein, the authors in~\cite{FreiGHHM24} show that we may restrict ourselves
even further to \deltatour{}s whose stopping points come from a small discrete set.
This becomes a very powerful tool given the continuous nature of the problem.

\begin{lemma}[Discretization Lemma~\cite{FreiGHHM24}]\label{lemma:discretization}
For every $\delta \geq 0$ and every connected graph $G$, there is a shortest \deltatour
	that is either nice or contains at most two stopping points and such that each stopping point of the tour can be described as $p(u,v,\lambda)$ with $\lambda \in S_\delta$
	where
$ S_\delta = \big \{0, \delta-\floor{\delta}, \half+\delta - \floor{\half+\delta}, 2\delta -\floor{2\delta} \big\}.$
\end{lemma}
For some given $\delta$, we use $S_\delta$ for the set defined in \cref{lemma:discretization}. We further use $s_\delta = \min \{|s_1-s_2|\mid \{s_i,1-s_i\}\cap S_\delta \neq \emptyset \text{ for $i \in [2]$ }, s_1 \notin \{s_2,1-s_2\}\}$.
We use the following simple property which is proved in \cite{FreiGHHM24}.
\begin{proposition}\label{lentour}
Let $T$ be a tour in a graph $G$ all of whose stopping points are of the form $p(u,v,\lambda)$ for some $uv\in E(G)$ and $\lambda\in S_\delta$. Then $\alpha(T)\leq \lceil\frac{\len(T)}{s_\delta}\rceil$.
\end{proposition}
For smaller values of $\delta$, the following stronger form of \cref{lemma:discretization} holds.

\begin{lemma}[Stronger Discretization Lemma~\cite{FreiGHHM24}]\label{discreteklein}
Let $\delta \in [0,1/2]$ and let $G$ be a connected graph. Then there is a shortest \deltatour~$T$ that either consists of at most two stopping points or is nice and such that for every $x_1x_2 \in E(G)$, one of the following holds:
\begin{enumerate}[(a)]
\item $V(G)=\{x_1,x_2\}$ and \[
    T= 
\begin{cases}
    p(x_1,x_2,\delta)p(x_1,x_2,1-\delta)p(x_1,x_2,\delta),& \text{if } \delta<\frac{1}{2}\\
    p(x_1,x_2,\frac{1}{2}),              & \text{if $\delta=\frac{1}{2}$,}
\end{cases}
\]
\item $\deg(x_i)=1$ and $\deg(x_{3-i})\geq 2$ for some $i \in \{1,2\}$, and $T$ contains the tour segment~$\seq{x_{3-i}&p(x_{3-i},x_i,1-\delta)&x_{3-i}}$ and does not stop at any other points on $x_1x_2$,
\item $\min\{\deg(x_1),\deg(x_2)\}\geq 2$ and $T$ traverses $x_1x_2$,
\item $\min\{\deg(x_1),\deg(x_2)\}\geq 2$ and $T$ contains the two tour segments $\seq{x_{3-i}}$ and ($\delta$-dependent)  
$\begin{cases}{\seq{x_{i}&p(x_{i},x_{3-i},1-2\delta)&x_{i}}}, &\text{if $\delta<\frac{1}{2}$}\\{\seq{x_i}}, &\text{if $\delta=\frac{1}{2}$}\end{cases}$, and $T$ does not stop at any other points on $x_1x_2$.
\end{enumerate}
\end{lemma}

In the following, we say that a tour {\it $\delta$-covers} an edge if it covers all the points on the edge. We speak about {\it covering} when $\delta$ is clear from the context.

We make use of the following characterization results, allowing us to check whether a given tour covers an edge. The following three lemmas correspond to three cases reflecting how many of the endvertices of an edge the tour stops at.
\begin{lemma}\label{nostopchar}
Let $G$ be a connected graph, $\delta\geq 0$ a real, $T$ a tour in $G$ and $x_1x_2 \in E(G)$ such that $T$ stops at none of $x_1$ and $x_2$. Then $T$ covers $x_1x_2$ if and only if one of the following holds:
\begin{enumerate}[(i)]
\item for $i \in [2]$, there are stopping points $p_i=p(u_i,v_i,\lambda_i)$ of $T$ with $\lambda_i \in [0,1)$ 
such that $\lambda_1+\lambda_2\geq \dist_G(x_1,v_1)+\dist_G(x_2,v_2)+3-2 \delta$ holds,
\item $T$ stops at points $p(x_1,x_2,\lambda_1)$ and $p(x_1,x_2,\lambda_2)$ for some $\lambda_1 \in (0,\delta]$ and some $\lambda_2 \in [1-\delta,1)$.
\end{enumerate}
\end{lemma}

\begin{lemma}[\cite{FreiGHHM24}]\label{char2stops}
Let $G$ be a connected graph, $\delta\geq 0$ a real, $T$ a nice tour in $G$ and $x_1x_2 \in E(G)$ such that $T$ stops at both $x_1$ and $x_2$. Then $T$ covers $x_1x_2$ if and only if one of the following holds:
\begin{enumerate}[(i)]
\item $T$ traverses $x_1x_2$,
\item $\delta \geq \frac{1}{2}$,
\item $T$ contains the segment~$\seq{x_i&p(x_i,x_{3-i},\lambda)&x_i}$ for some $i \in \{1,2\}$ and some $\lambda 
\geq 1-2 \delta$.
\end{enumerate}
\end{lemma}
\begin{lemma}\label{char1stop}
Let $G$ be a connected graph, $\delta\geq 0$ a real, $T$ a tour in $G$ and $x_1x_2 \in E(G)$ such that $T$ stops at $x_1$ but not at $x_2$. Then $T$ covers $x_1x_2$ if and only if one of the following holds:
\begin{enumerate}[(i)]
\item There are stopping points $p_1=p(x_1,x_2,\lambda_1)$ and $p_2=(v,x_2,\lambda_2)$ of $T$ with $\lambda_1,\lambda_2 \in [0,1)$ and $v \in N_G(x_2)$ such that $\lambda_1+\lambda_2\geq 2-2 \delta$ holds.
\item $T$ stops at the point $p(x_1,x_2,\lambda)$ for some $\lambda\in [1-\delta,1]$.
\end{enumerate}
\end{lemma}
We now state two further simple results.
The first is an observation giving an upper bound on the maximum length of $\delta$-tours if $\delta$ is not too small.
\begin{proposition}[\cite{FreiGHHM24}]\label{trgvftzuh}
Let $G$ be a connected graph and $\delta \geq 1/2$ a constant. Then $G$ admits a \deltatour of length at most $2n-2$ and this tour can be found in polynomial time. 
\end{proposition}

The second one asserts that the distance from a point outside a \deltatour to the points passed by the tour
is the distance to one of its stopping points.

\begin{proposition}[\cite{FreiGHHM24}]\label{nearstop}
Let $G$ be a connected graph, $T$ a \deltatour in $G$ and $p \in P(G)$ a point which is not passed by $T$.
Then $\dist_G(p,T)=\dist_G(p,q)$ for some point $q \in P(G)$ at which $T$ stops.
\end{proposition}

We finally conclude with two useful algorithmic corollaries.
Those essentially follow from \cref{nearstop}, the above structural results on discretization (\cref{lemma:discretization})
and the existence of nice tours (\cref{lem:tournice}).
The details are provided in~\cite{FreiGHHM24}.

\begin{corollary}[\cite{FreiGHHM24}]\label{check}
Given a graph~$G$, a constant $\delta$ and a tour~$T$ in $G$, we can decide in polynomial time whether $T$ is a \deltatour.

\end{corollary} 

\begin{corollary}[\cite{FreiGHHM24}]\label{decide}
Given a graph~$G$ and a constant $\delta>0$, there is an algorithm that computes a shortest \deltatour in $G$ and runs in $f(n)$.
\end{corollary}

\section{Overview of Results}
\label{section:overview}

We present our hardness of approximation results in \cref{section:overview:approximation:lowerbounds};
for positive approximation results we refer to our previous work~\cite{FreiGHHM24}.
Finally we turn to the parameterized complexity results in \cref{section:overview:parameterized}.

\subsection{Inapproximability Results}
\label{section:overview:approximation:lowerbounds}
Recall that in \cite{FreiGHHM24}, we showed that \deltatourprob admits a polynomial-time constant-factor approximation algorithm for every $\delta<\frac{3}{2}$.
We now rule out the existence of a PTAS by showing APX-hardness for every $\delta>0$. 
The main challenge is proving the result for the range $\delta\in (0,1/2]$, which we achieve by reducing from \VC with a construction creating cubic bipartite graphs. 
From the case $\delta=1/2$ we then obtain the new result that TSP remains \apx-hard even on this very restricted graph class. 
+We then extend the hardness result for \deltatourprob to any $\delta>0$ by a simple subdivision argument. 
Finally, we describe a stronger inapproximability result for $\delta \geq 3/2$.
\subparagraph*{APX-Hardness for Covering Range \boldmath$\delta\in(0,1/2]$.}
\label{section:inapprox-small-delta}
As the first step towards the APX-hardness of \deltatourprob in the range $\delta\in(0,1/2]$, 
we introduce a new family of 
optimization problems called  $(\alpha,\beta,\gamma,\kappa)$-\CycleSubpartition, 
that is also interesting on its own. 
			
\begin{myproblem}[$(\alpha,\beta,\gamma,\kappa)$-\CycleSubpartition, where $\alpha,\beta,\gamma,\kappa\in\mathbb{R}$ and $\alpha,\beta,\gamma>0$]
Instance&A simple graph $G$.\\
Solution&Any set $\mathcal{C}$ of pairwise vertex-disjoint cycles in $G$.\\
Goal& Minimize $\alpha\lvert\mathcal{C}\rvert+\beta\lvert V(G)\setminus \bigcup_{C \in \mathcal{C}} V(C) \rvert+\gamma\lvert V(G)\rvert+\kappa$.
\end{myproblem}

This bi-objective problem asks us, roughly speaking, for any given graph, to cover 
as many vertices as possible with a family of 
as few vertex disjoint cycles as possible. 
The precise balance between the two opposed optimization goals 
is tuned by the problem parameters. 
In particular, $\alpha$ specifies the cost for each cycle 
in the solution and $\beta$ for each vertex left uncovered.  
Disallowing uncovered vertices 
(or making them prohibitively expensive), 
yields the classical APX-hard minimization problem \CyclePartition~\cite[Thm~3.1,~Prob.~(iv)]{SahniG76}. 

The two remaining parameters $\gamma$ and $\kappa$
may appear artificial since their only immediate effect is to 
make any solution for a given graph~$G$ more expensive 
by the same cost $\gamma \lvert V(G)\rvert+\kappa$. 
They will prove meaningful, however, for our main goal of this section. 
Namely, we first establish APX-hardness for $(\alpha,\beta,\gamma,\kappa)$-\CycleSubpartition for cubic graphs
in the entire parameter range of $\alpha,\beta,\gamma,\kappa\in\mathbb{R}, \alpha,\beta,\gamma>0$ 
and then show that on cubic graphs, for every $\delta\in(0,1/2]$, we have that \deltatourprob coincides with this problem for an appropriate choice of the parameters $\alpha,\beta,\gamma$, and $\kappa$.
The proof uses a reduction from \VC on cubic graphs, which is known to be APX-hard~\cite[Thm.~5.4]{KarpinskiS15}. 

For some fixed $\alpha,\beta,\gamma,\kappa\in\mathbb{R}$ with $\alpha,\beta,\gamma>0$, given an instance $G$ of cubic \VC, we create a cubic graph~$H$ which we view as an instance of  $(\alpha,\beta,\gamma,\kappa)$-\CycleSubpartition. It is not difficult to obtain a packing of cycles in $H$ with the appropriate properties from a vertex cover in $G$. The other direction, that is, obtaining a vertex cover in $G$ from a cycle packing in $H$ is significantly more delicate. A collection of careful reconfiguration arguments is needed to transform an arbitrary cycle cover in $H$ into one that is of a certain particular shape and not more expensive. Having a cycle cover of this shape at hand, a corresponding vertex cover in $G$ can easily be found.
As mentioned above, we now easily obtain APX-hardness of \deltatourprob for any $\delta \in (0,1/2]$ and even 
the previously unknown APX-hardness for cubic bipartite graphic~TSP. 
\begin{restatable}{linkedtheorem}{ThmApproxLbZeroHalf}
\label{ThmApproxLbZeroHalf}
\label{thm:approx:lb:zero:half}
	On cubic bipartite graphs, \deltatourprob is \apx-hard for $\delta\in(0,1/2]$. 
\end{restatable}

\begin{restatable}{linkedcorollary}{APXhardnessTSPcorollary}
\label{APXhardnessTSPcorollary}
\label{cor:approx:lb:tsp}\TSP is \apx-hard even on cubic bipartite graphs.
\end{restatable}

With Theorem~\ref{thm:approx:lb:zero:half} at hand, we further easily obtain a hardness result for larger values of $\delta$. Namely, observe that for any nonnegative integer $k$, any constant $\delta$ and any connected graph~$G$, there is a direct correspondence between the $\delta$-tours in $G$ and the $k\delta$-tours in the graph obtained from $G$ by subdividing every edge $k-1$ times, yielding the following result. 

\begin{restatable}{linkedtheorem}{APXhardFromHalfToThreeHalves}
\label{APXhardFromHalfToThreeHalves}
\label{cor:approx:lb:half:threehalves}
The problem \deltatourprob is \apx-hard for any real $\delta>0$.
\end{restatable}

\subparagraph*{Stronger Inapproximability for Covering Range \boldmath$\delta\ge 3/2$.}

For $\delta \geq 3/2$,
	we give a lower bound of roughly $\log n$ on the approximation ratio,
	asymptotically matching our upper bound. Like for all our inapproximability results, we prove hardness for the decision version of the problem. 

\begin{restatable}{linkedtheorem}{ThmApproxLbLogN}
\label{ThmApproxLbLogN}
\label{thm:approx:lb:log:n}
	Unless $\p=\np$, for every $\delta \geq 3/2$, there is an absolute constant $\alpha_\delta$ such that there is no polynomial-time algorithm that, given a connected graph~$G$ and a constant $K$, returns `yes' if $G$ admits a \deltatour of length at most $K$ and `no' if $G$ does not admit a \deltatour of length at most $\alpha_\delta \log(|V(G)|)K$. 
\end{restatable}

We start from the inapproximability of \domset
	on split graphs,
	implicitly given by Dinur and Steurer~\cite[Corollary 1.5]{DinurS14}.
Given a split graph~$G$ satisfying some nontriviality condition, we can construct a graph~$G'$
	such that the minimum size of a dominating set of $G$
	and the length of a shortest \deltatour in $G'$ are closely related.

\subsection{Parameterized Complexity}
\label{section:overview:parameterized}
We examine the problem's parameterized complexity for two parameters:
tour length and~$n/\delta$.

\subparagraph*{Parameterization by Tour Length.}
We now give an overview of the proof of Theorems \ref{thm:parammain} and \ref{thm:parammain2}.
For $\delta \geq 3/2$, \wtwo-hardness follows by a reduction from \domset on split graphs,
	namely the same as used to show inapproximability for $\delta \geq 3/2$.
We complement this result by giving an \fpt-algorithm when $\delta < 3/2$.
In fact, we give an algorithm that allows $\delta$ to be part of the input
	and that is fixed-parameter tractable for $\delta$ and the maximum allowed tour length $K$ combined.

\begin{restatable}{linkedtheorem}{FPTSol}
\label{FPTSol}
\label{fpt:sol}
There is an algorithm that, given a graph~$G$ and reals $\delta \in (0,3/2)$ and $K \geq0$,
	decides whether $G$ has a \deltatour of length at most $K$
	in $f(K,\delta) \cdot n^{\Oh(1)}$ time, for some computable function $f$.
\end{restatable}

Our algorithm is based on a kernelization:
we either correctly conclude that $G$ has no \deltatour of length at most $K$,
	or output an equivalent instance of size at most $f(K,\delta)$
	for a computable function $f$.
The key insight is a bound on the vertex cover size of $f(K,\delta)$
	for a computable function,
	assuming there is a \deltatour of length at most $K$.
Hence we may compute an approximation $C$ of a minimum vertex cover,
	and reject the instance if $C$ is too large.
We partition the vertices in $V(G)\setminus C$ by their neighborhood in the vertex cover $C$.
Now if a set $S$ of the partition has size larger than $f(K,\delta)$
	for a certain computable function $f$,
	then deleting a vertex of $S$ yields an equivalent instance.
	In light of the \wtwo-hardness result mentioned above, it is natural to study the classical complexity of the problem of deciding whether a \deltatour of fixed length exists for large $\delta$. We first consider the case that $\delta$ is fixed. In that case, a polynomial-time algorithm follows rather easily. More concretely, thanks to \cref{lemma:discretization}, it is possible to enumerate a number of sequences of bounded size of which is a tour of the desired form if such a tour exists. It hence suffices to check whether one of the constructed sequences is such a tour.
	
On the other hand, it turns out that the situation drastically changes when $\delta$ is part of the input. We show that the problem becomes \np-hard, somewhat surprisingly even if severe restrictions on the instance are imposed. First, we show the hardness not only for the tour length being a large constant, but even for an arbitrarily small tour length. Moreover, we show that for hardness to occur, we do not actually require $\delta$ to be large. In fact, it is sufficient if $\delta$ is close enough to $2$. Formally, we prove the following result.

\begin{restatable}{linkedtheorem}{nphdinp}\label{nphdinp}
For every $\eps>0$, there is no polynomial-time algorithm that decides, given a graph $G$, a constant $\delta$ with $0<\delta<2$, and a constant $K$ with $K<\eps$, whether $G$ admits a \deltatour of length at most~$K$, unless $\p=\np$.
\end{restatable}

Theorem \ref{nphdinp} is proved by a reduction from Dominating Set in split graphs. Given a split graph with a partition $(C,I)$ of its vertex set such that $G[C]$ is a clique and $G[I]$ is edgeless, we add a new vertex $v_0$ and link it to all vertices in $C$. We further add a gadget to $v_0$ forcing the tour to stop at $v_0$. Further, we choose $K$ sufficiently small such that any tour of length $K$ can only leave $v_0$ to go to a very close stopping point and then return. Roughly speaking, the number of these segments that are allowed within the tour length corresponds to the size of the desired dominating set. It turns out that the constructed tour is a \deltatour if and only if the set of vertices in $C$ into whose direction the tour travels is a dominating set of the input graph. 

\subparagraph*{XP Algorithm for Large Covering Range.}
We here give an overview of the proof of \cref{thm:param:k:ub:largedelta}. The crucial idea is that, if $T$ is a \deltatour, then, while the length of $T$ can be linear in $n$, there is a set of points stopped at by this tour that covers all the points in $P(G)$ and whose size is bounded by a function depending only on $k$ where $k=\frac{n}{\delta}$. Intuitively speaking, the remainder of the tour is needed to connect the points in this set, but not to actually cover points in $P(G)$. Therefore, these segments connecting the points in the set can be chosen to be as short as possible. These observations can be subsumed in the following lemma, crucial to the proof of Theorem~\ref{thm:param:k:ub:largedelta}. 

\begin{restatable}{linkedlemma}{MainLemXP}
\label{MainLemXP}
\label{lem:main:lem:xp}
	Let $G$ be a connected graph of order $n$, $\delta$ a positive real, $k=\lceil\frac{n}{\delta}\rceil$ and suppose that $n \geq 12 k$. Further, let $T$ be a shortest \deltatour in $G$.
	Then there is a set $Z \subseteq P(G)$ of points stopped at by $T$
	with $\abs{Z} \leq 12k$ such
	that for every point $p \in P(G)$, we have $\dist_G(p,Z)\leq \delta$.
\end{restatable}

With Lemma~\ref{lem:main:lem:xp} at hand, Theorem~\ref{thm:param:k:ub:largedelta} follows easily. Assuming that the minimum size requirement in Lemma~\ref{lem:main:lem:xp} is met, due to the discretization lemma (\cref{lemma:discretization}), there is a shortest tour only stopping at points from a set of size $\Oh(n^2)$. We then enumerate all subsets of size at most $12k$ and for each of these sets, compute a shortest tour passing through its elements and check whether it is a \deltatour.  It follows from Lemma~\ref{lem:main:lem:xp} that the shortest tour found during this procedure is a shortest \deltatour.
In the proof of Lemma~\ref{lem:main:lem:xp}, we define $Z$ as the union of two sets $Z_1$ and $Z_2$. The set $Z_1$ is an inclusion-wise minimal set of points stopped at by $T$ that covers all points in $P(G)$ whose distance to $T$ is at least $\frac{n}{2k}$. For each $z \in Z_1$, by definition, there is such a point $p_z$ for which $\dist_G(p_z,z')>\delta$ holds for all $z' \in Z_1-z$. Now for every $z \in Z_1$, we choose a shortest walk from $p_z$ to $z$. It turns out that these walks do pairwise not share vertices of $V(G)$ and each of them contains $\Oh(\frac{n}{k})$ vertices of $V(G)$. It follows that $\abs{Z_1}=\Oh(k)$. We now walk along $\floor{T}$ and, roughly speaking, create $Z_2$ by adding every $\frac{n}{3k}$-th vertex stopped at by $\floor{T}$. It turns out that $Z_2$ covers all points in $P(G)$ whose distance to $T$ is at most $\frac{n}{2k}$. Further, as the length of $\floor{T}$ is bounded by $2n$, we have $\abs{Z_2}=\Oh(k)$.

\subparagraph*{W[1]-Hardness for Large Covering Range.}

Given the \xp-time algorithm running in time
$f(k) \cdot n^{\Oh(k)}$ designed for the regime $\delta = \Omega(n)$, it is natural to ask
whether there is an FPT-time algorithm for the same parameter $k \coloneqq \ceil{\frac{n}{\delta}}$.
The answer is no. We show \wone-hardness and the running time of our algorithm for \cref{thm:param:k:ub:largedelta}
to be close to optimal under the Exponential Time Hypothesis (\ETH).
The hardness is based on the fact that
\binaryCSP on cubic constraint graphs
	cannot be solved in time $f(k) \cdot n^{o(k/\log{k})}$ under
	\ETH; see~\cite{KarthikMPS23, Marx07}. 
(The exact formulation we use is stronger and gives a lower bound for every fixed $k$.) 
An instance of \binaryCSP is a graph~$G$ with $k$ edges,
where the vertices represent variables taking values from
a domain $\Sigma = [n]$, and every edge is associated with a
constraint relation $C_{i, j} \subseteq \Sigma \times \Sigma$ over the two variables $i$ and $j$.
The instance is \textit{satisfiable}
if there is an assignment to the variables $\mathcal{A} \colon  V(G) \to \Sigma$ such
that $(\assig{i}, \assig{j}) \in C_{i,j}$ for every
constraint relation $C_{i,j}$.

In our reduction, we construct $k$ gadgets corresponding to the constraints, with each gadget having some number of portals. Each gadget has multiple possible states corresponding to the
satisfying assignments of its two constrained variables. If two constraints share a variable, then the corresponding gadgets are connected by paths between appropriate portals. These connections ensure that the selected states of the two gadgets agree on the value of the variable. It is now easy to find a tour in the auxiliary graph given a satisfying assignment for the formula. On the other hand, the construction is designed so that all tours in the auxiliary graph are in a certain shape, allowing to obtain an assignment from them.

\section{Inapproximability Results}
\label{section:approximation-lb}
We now present the complementing lower bounds complementing the results in \cite{FreiGHHM24}. 
Specifically, we show \apx-hardness for 
any fixed $\delta\in(0,3/2)$, which rules out a PTAS, 
and show that there is no $o(\log n)$-approximation for any fixed $\delta\ge 3/2$ unless \p=\np.

\subsection[\texorpdfstring{APX-Hardness for Covering Range $\delta\le 1/2$}{APX-Hardness for Covering Range delta <= 1/2}]{APX-Hardness for Covering Range \boldmath$\delta\le 1/2$}
In this section, we prove \apx-hardness of \deltatourprob for every $\delta \in (0, 1/2)$ and subsequently show that our results imply \apx-hardness of TSP. 
Both hardness results hold even when the input is restricted to the class of cubic bipartite graphs; see~\cref{thm:approx:lb:zero:half,cor:approx:lb:tsp}.

The proofs of \cref{thm:approx:lb:zero:half} and \cref{cor:approx:lb:tsp} are based on a connection between tours and certain packings of cycles in cubic graphs.
In order to exploit this connection, we show that optimizing such cycle packings is \apx-hard.  We begin by introducing this packing problem.
Specifically, for any fixed reals $\alpha,\beta,\gamma,\kappa$ with $\alpha,\beta,\gamma>0$, we consider the following problem.

\begin{myproblem}[$(\alpha,\beta,\gamma,\kappa)$-\CycleSubpartition, where $\alpha,\beta,\gamma,\kappa\in\mathbb{R}$, $\alpha,\beta,\gamma>0$]
\label{prob:cubic-cycle-cover}%
Instance&A simple graph $G$.\\
Solution&Any set $\mathcal{C}$ of pairwise vertex-disjoint cycles in $G$.\\
Goal& Minimize $\alpha\lvert\mathcal{C}\rvert+\beta\lvert V(G)\setminus\bigcup_{\substack{C\in\mathcal{C}}} V(C)\rvert+\gamma\lvert V(G)\rvert+\kappa$.
\end{myproblem}

The main technical contribution in the proof of \cref{thm:approx:lb:zero:half} and \cref{cor:approx:lb:tsp} is showing the \apx-hardness of $(\alpha,\beta,\gamma,\kappa)$-\CycleSubpartition for all suitable constants. We do this in Section \ref{cycle1}. After that,  \cref{thm:approx:lb:zero:half} and \cref{cor:approx:lb:tsp} follow rather easily, thanks to the above-mentioned connection. We give these proofs in \cref{cycle2} and \cref{cycle3}, respectively.
\subsubsection{Hardness for Cycle Subpartitions}\label{cycle1}
Given a connected graph $G$ and some reals $\alpha,\beta,\gamma,\kappa$ with $\alpha,\beta,\gamma>0$, we call a collection $\mathcal{C}$ of vertex-disjoint cycles all of which are subgraphs of $G$ a {\it cycle subpartition} of $G$ and we call 
\[w(\mathcal{C})\coloneqq \alpha\lvert 
			\mathcal{C}\rvert + \beta \lvert 
			V(G)\setminus \bigcup_{\substack{C\in\mathcal{C}}} V(C) \rvert + \gamma\lvert 
			V(G)\rvert + 
			\kappa$ the weight of $\mathcal{C}.\]
If $\gamma$ and $\kappa$ are not specified, they default to $\gamma=\kappa=0$. We further write $V(\mathcal{C})$ for $\bigcup_{C \in \mathcal{C}}V(C)$.
			
Formally, we show the following result:
\begin{theorem}\label{cyclehard}
For all reals $\alpha,\beta,\gamma,\kappa$ with $\alpha,\beta,\gamma>0$, there is an $\eps >0$ such that, unless $\p=\np$, there is no polynomial-time algorithm that, given a connected cubic bipartite graph $G$ and a positive real $K$, returns `yes' if $G$ admits a cycle subpartition of weight $K$ and `no' if $G$ does not admit a cycle subpartition of weight $(1+\eps)K$.
\end{theorem}

Our proof of \cref{cyclehard} is based on a reduction from vertex cover on cubic graphs. Given a graph $G$, a vertex cover of $G$ is a set $X \subseteq V(G)$ such that every edge in $E(G)$ is incident to at least one vertex of $X$. The problem of finding a vertex cover of minimum size in a given graph has attracted significant attention.  The following well-known result on the \apx-hardness of this problem can be found in~\cite[Thm.~3.1 and Figure~4]{AlimontiK97}.

\begin{proposition}\label{vchard}
There is an $\eps>0$ such that, unless $\p=\np$, there is no polynomial-time algorithm that, given a connected simple cubic graph $G$ and a positive real $K$, returns `yes' if $G$ admits a vertex cover of size $K$ and `no' if $G$ does not admit a vertex cover of size $(1+\eps)K$.
\end{proposition}

In order to use \cref{vchard}, we need the following result which is the most difficult part of the proof.

\begin{lemma}\label{mainconstr}
Let $\alpha,\beta>0$ be constants. Then, given a connected simple cubic graph $G$, in polynomial time, we can construct a cubic bipartite graph $H$ that satisfies $|V(H)|\leq (54 \max\{12, \ceil{\frac{2\alpha}{\beta}}\} + 18) |V(G)|$ such that for every positive integer $K$,
we have that $H$ admits a cycle subpartition of weight at most $\alpha (|V(G)| + 2K)$ if and only if $G$ admits a vertex cover of size at most $K$.
\end{lemma}

Before we give the main proof of Lemma \ref{mainconstr}, we need to describe an auxiliary gadget that will be useful throughout the reduction. Namely, for two reals $\alpha,\beta>0$ and some positive integer $k$, an {\it $(\alpha,\beta,k)$-chain gadget connecting $\ell$ and $r$} is a graph $G$ together with two special vertices $\ell,r \in V(G)$ satisfying the following properties:
\begin{enumerate}[(i)]
\item $d_G(\ell)=d_G(r)=1$ and $d_G(v)=3$ for all $v \in V(G)-\{\ell,r\}$,
\item $|V(G)|=6k+2$,
\item there is a unique bipartition of $G$ with $\ell$ and $r$ being in different classes of this bipartition,
\item $G$ contains a Hamiltonian $\ell r$-path, and 
\item for every cycle subpartition $\mathcal{C}$ of $G-\{\ell,r\}$, we have $w(\mathcal{C})\geq \min\{\alpha,6 \beta\} k$.
\end{enumerate} 
We say that $V(G)-\{\ell,r\}$ is the set of {\it interior} vertices of the chain gadget.
\begin{proposition}\label{fzguhij}
For any reals $\alpha,\beta>0$ and any positive integer $k$, a $(\alpha,\beta,k)$-chain gadget $G$ exists and can be computed in polynomial time. 
\end{proposition}
\begin{proof}
We construct the gadget as follows. We let $V(G)$ contain $\ell,r$ and a vertex $v_i^{j}$ for every $i \in [k]$ and every $j \in [6]$. For $i \in [k]$, let $V_i=\{v_i^1,\ldots,v_i^6\}$. Next, for every $i \in [k]$, we let $E(G)$ contain the edges $v_i^1v_i^2,v_i^1v_i^4,v_i^2v_i^3,v_i^2v_i^5,v_i^3v_i^4,v_i^3v_i^6,v_i^4v_i^5$, and $v_i^5v_i^6$. Further, for every $i \in [k-1]$, we let $E(G)$ contain the edge $v_i^6v_{i+1}^1$.
Finally, we let $E(G)$ contain the edges $\ell v_1^1$ and $rv_k^6$. This finishes the construction of $G$, which is clearly possible in polynomial time. 
For an illustration, see \cref{fig:vc-covercycle-chaingadget}.

We now check verify that all listed properties are satisfied. 
It follows directly by construction that $(G,\ell,r)$ satisfies $(i)$ and $(ii)$.

Let $V_{\text{even}}=\{v_i^{j}:i \in [k], \text{$j \in [6]$ even }\}\cup \ell$ and $V_{\text{odd}}=V(G)-V_{\text{even}}$. Then $(V_{\text{even}},V_{\text{odd}})$ is a bipartition of $G$ with the desired properties, so $(iii)$ follows.

Let $P$ be the path defined by $V(P)=V(G)$ and where $E(P)$ contains the edges $v_i^1v_i^2,v_i^2v_i^3,v_i^3v_i^4,v_i^4v_i^5$, and $v_i^5v_i^6$ for $i \in [k]$, the edge $v_i^6v_{i+1}^1$ for all $i \in [k-1]$ and the edges $\ell v_1^1$ and $rv_k^6$. It is easy to see that $P$ is a Hamiltonian $\ell r$-path of $G$. This proves $(iv)$.

Now let $\mathcal{C}$ be a cycle subpartition of $G-\{\ell,r\}$. First suppose that there is a smallest index $i \in [k-1]$ such that $v_i^6v_{i+1}^1 \in E(C)$ for some $C \in \mathcal{C}$. Then, by definition, we obtain that $v_i^6v_{i+1}^1$ is the only edge of $E(C)$ linking $V_i$ and $V(G)-V_i$. This contradicts $C$ being a cycle. We hence obtain that for every $i \in [k]$, there is a cycle subpartition $\mathcal{C}_i$ of $G[V_i]$ inherited from $\mathcal{C}$. Consider some $i \in [k]$. If $\mathcal{C}_i \neq \emptyset$, we have $w(\mathcal{C}_i)\geq \alpha |\mathcal{C}_i|\geq \min\{\alpha,6 \beta\}$. 
Otherwise,
we have $w(\mathcal{C})\geq \beta|V_i|\geq \min\{\alpha,6 \beta\}$. It follows that $w(\mathcal{C})=\sum_{i \in [k]}w(\mathcal{C}_i)\geq \min\{\alpha,6 \beta\} k$. Hence $(v)$ holds.
\end{proof}

\begin{figure}
\begin{center}
	\centering
	\tikzstyle{wavey}=[decorate,decoration={coil, aspect=-.5, post length=1mm, segment length=1mm, pre length=2mm},
shorten <= -.8pt,shorten >= -0.8pt]

\begin{tikzpicture}[
    vertex/.style={
        draw,
        circle, 
        inner sep=.0pt, 
        minimum size=0.62cm},
    every edge/.append style={}
]
\newcommand{\hexside}{1.8}
	\node[vertex] (l) at (-6,0) {$\ell$};
	\node[vertex] (r) at (-4.75,0) {$r$};
	\draw (l)[wavey]--(r);

	\node[vertex] (left) at (-3.5,0) {$\ell$};
	\node[vertex] (a1) at (-2.5,0) {};
	\node[vertex] (a2) at (-2,.866) {};
	\node[vertex] (a3) at (-2,-.866) {};
	\node[vertex] (a4) at (-1,.866) {};
	\node[vertex] (a5) at (-1,-.866) {};
	\node[vertex] (a6) at (-0.5,0) {};
	\node[vertex] (b1) at (.5,0) {};
	\node[vertex] (b2) at (1,.866) {};
	\node[vertex] (b3) at (1,-.866) {};
	\node[vertex] (b4) at (2.0,.866) {};
	\node[vertex] (b5) at (2.0,-.866) {};
	\node[vertex] (b6) at (2.5,0) {};
	\node[vertex] (c1) at (4,0) {};
	\node[vertex] (c2) at (4.5,.866) {};
	\node[vertex] (c3) at (4.5,-.866) {};
	\node[vertex] (c4) at (5.5,.866) {};
	\node[vertex] (c5) at (5.5,-.866) {};
	\node[vertex] (c6) at (6,0) {};
	\node[vertex] (right) at (7,0) {$r$};

	\node at (a1) {$v_1^1$};
	\node at (a2)  {$v_1^2$};
	\node at (a3)  {$v_1^4$};
	\node at (a4)  {$v_1^3$};
	\node at (a5)  {$v_1^5$};
	\node at (a6)  {$v_1^6$};

	\node at (b1)  {$v_2^1$};
	\node at (b2)  {$v_2^2$};
	\node at (b3)  {$v_2^4$};
	\node at (b4)  {$v_2^3$};
	\node at (b5)  {$v_2^5$};
	\node at (b6)  {$v_2^6$};

	\node at (c1)  {$v_k^1$};
	\node at (c2)  {$v_k^2$};
	\node at (c3)  {$v_k^4$};
	\node at (c4)  {$v_k^3$};
	\node at (c5)  {$v_k^5$};
	\node at (c6)  {$v_k^6$};

	\draw (a1)--(a2)--(a4)--(a6)--(a5)--(a3)--(a1);
	\draw (a6)--(b1);
	\draw (b1)--(b2)--(b4)--(b6)--(b5)--(b3)--(b1);
	\draw [thick, dotted] (b6)--(c1);
	\draw (c1)--(c2)--(c4)--(c6)--(c5)--(c3)--(c1);
	\draw (left)--(a1);
	\draw (a2)--(a5);
	\draw (a3)--(a4);
	\draw (b2)--(b5);
	\draw (b3)--(b4);
	\draw (c2)--(c5);
	\draw (c3)--(c4);
	\draw (c6)--(right);

	\draw[dashed, rounded corners]
	(-6.5, -.5) rectangle (-4.25,.5) {};
	\draw[dashed, rounded corners]
	(-4, -1.25) rectangle (7.45,1.25) {};
\end{tikzpicture}
\end{center}
\captionsetup{aboveskip=-20pt,belowskip=20pt}
\caption{An $(\alpha, \beta, k)$-chain gadget $G$.}
\label{fig:vc-covercycle-chaingadget}
\end{figure}
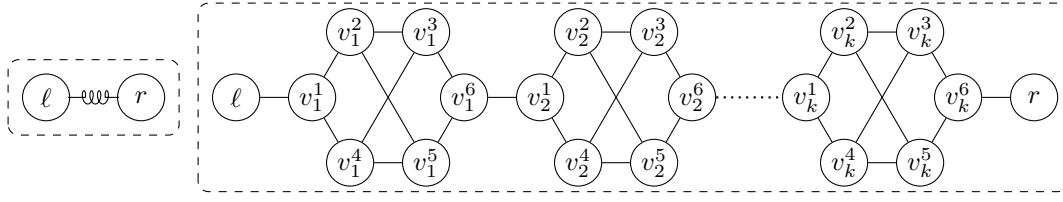

We are now ready to give the main proof of \cref{mainconstr}.

\begin{proof}[Proof of~\cref{mainconstr}]
Let $\alpha,\beta>0$ be fixed and let $G$ be a simple connected cubic graph.
We set $k=\max\{12,\lceil\frac{2 \alpha}{\beta}\rceil\}+1$. We now construct a graph $H$. During this construction, when we say that we connect two vertices $u,v$ with a chain gadget, we mean that we add an $(\alpha,\beta,k)$-chain gadget $(G,\ell,r)$ to the current graph and then identify $u$ with $\ell$ and $v$ with $r$.
First, for every $u \in V(G)$, we let $V(H)$ contain a set $X_u=\{u_1^x,u_2^x,u_1^y,u_2^y,u_1^z,u_2^z\}$ of 6 vertices where $(x,y,z)$ is an arbitrary ordering of $N_G(u)$. Next, we add the edges $u_2^xu_1^y,u_2^yu_1^z$, and $u_2^zu_1^x$ to $E(G)$ and we connect each of the following pairs of vertices with a chain gadget: $(u_1^x,u_2^x),(u_1^y,u_2^y),(u_1^z,u_2^z)$.
Next, for every edge $uv \in E(G)$, we add a set $X_{uv}=\{u_3^v,u_4^v,u_5^v,u_6^v,v_3^u,v_4^u,v_5^u,v_6^u\}$ of 8 vertices to $V(H)$, we add the edges $u_1^vu_6^v,u_2^vu_3^v,u_3^vv_4^{u},u_4^vu_5^v,u_4^vv_3^{u},u_5^vv_6^{u},u_6^vv_5^{u}, v_1^{u}v_6^{u}, v_2^{u}v_3^{u}$, and $v_4^{u}v_5^{u}$, and we connect the following pairs of vertices by chain gadgets: $(u_3^v,u_4^v),(u_5^v,u_6^v),(v_3^{u},v_4^{u}),(v_5^{u},v_6^{u})$. In the following, when speaking of a chain gadget, we mean one of those explicitly added in the construction.
This finishes the description of $H$. By \cref{fzguhij} and construction, $H$ can be computed in polynomial time. Further, for a vertex $u \in V(G)$, we use $X'_u$ for $X_u \cup \bigcup_{v \in N_G(u)}X_{uv}$.
For an illustration, see \cref{fig:vc-covercycle-transformation}.

\begin{figure}
	\centering
	\tikzstyle{wavey}=[decorate,decoration={coil, aspect=-.5, post length=1mm, segment length=1mm, pre length=2mm},
shorten <= -.8pt,shorten >= -0.8pt]

\begin{tikzpicture}[rotate=20,
    vertex/.style={
        draw,
        circle, 
        inner sep=.0pt, 
        minimum size=0.62cm},
    every edge/.append style={}]
    \def\radius{5}
    \def\gap{1}

    \pgfmathsetmacro\x{\radius * cos(0)}
    \pgfmathsetmacro\y{\radius * sin(0)}
    \node[vertex] (v1) at (\x,\y) {$u_1^x$};
    
    \pgfmathsetmacro\xvFourteen{\radius * cos(260)}
    \pgfmathsetmacro\yvFourteen{\radius * sin(260)}
    \node[vertex] (v14) at (\xvFourteen, \yvFourteen) {$u_2^x$};
    
    \node[vertex] (v15) at ({\xvFourteen + (1/7) * (\x - \xvFourteen)}, {\yvFourteen + (1/7) * (\y - \yvFourteen)}) {$u_3^x$};
    \node[vertex] (v16) at ({\xvFourteen + (3/7) * (\x - \xvFourteen)}, {\yvFourteen + (3/7) * (\y - \yvFourteen)}) {$u_4^x$};
    \node[vertex] (v17) at ({\xvFourteen + (4/7) * (\x - \xvFourteen)}, {\yvFourteen + (4/7) * (\y - \yvFourteen)}) {$u_5^x$};
    \node[vertex] (v18) at ({\xvFourteen + (6/7) * (\x - \xvFourteen)}, {\yvFourteen + (6/7) * (\y - \yvFourteen)}) {$u_6^x$};

    \pgfmathsetmacro\dx{\x - \xvFourteen}
    \pgfmathsetmacro\dy{\y - \yvFourteen}
    \pgfmathsetmacro\len{sqrt(\dx*\dx + \dy*\dy)}
    \pgfmathsetmacro\nx{\dy / \len} 
    \pgfmathsetmacro\ny{-\dx / \len}

    \node[vertex] (mv1) at ({\x + 2*\gap*\nx}, {\y + 2*\gap*\ny}) {$x_1^u$};

    \node[vertex] (mv14) at ({\xvFourteen + 2*\gap*\nx}, {\yvFourteen + 2*\gap*\ny}) {$x_2^u$};
    \node[vertex] (mv15) at ({(\xvFourteen + (1/7) * (\x - \xvFourteen)) + 2*\gap*\nx}, 
                              {(\yvFourteen + (1/7) * (\y - \yvFourteen)) + 2*\gap*\ny}) {$x_3^u$};
    \node[vertex] (mv16) at ({(\xvFourteen + (3/7) * (\x - \xvFourteen)) + 2*\gap*\nx}, 
                              {(\yvFourteen + (3/7) * (\y - \yvFourteen)) + 2*\gap*\ny}) {$x_4^u$};
    \node[vertex] (mv17) at ({(\xvFourteen + (4/7) * (\x - \xvFourteen)) + 2*\gap*\nx}, 
                              {(\yvFourteen + (4/7) * (\y - \yvFourteen)) + 2*\gap*\ny}) {$x_5^u$};
    \node[vertex] (mv18) at ({(\xvFourteen + (6/7) * (\x - \xvFourteen)) + 2*\gap*\nx}, 
                              {(\yvFourteen + (6/7) * (\y - \yvFourteen)) + 2*\gap*\ny}) {$x_6^u$};

    \draw (v14) -- (v15);
    \draw[wavey] (v15) -- (v16); 
    \draw (v16) -- (v17);
    \draw[wavey] (v17) -- (v18); 
    \draw (v18) -- (v1);

    \draw[wavey] (mv1) to[bend left=30] (mv14);
    \draw (mv14) -- (mv15);
    \draw[wavey] (mv15) -- (mv16); 
    \draw (mv16) -- (mv17);
    \draw[wavey] (mv17) -- (mv18); 
    \draw (mv18) -- (mv1);

    \pgfmathsetmacro\xvTwo{\radius * cos(20)}
    \pgfmathsetmacro\yvTwo{\radius * sin(20)}
    \node[vertex] (v2) at (\xvTwo, \yvTwo) {$u_2^y$};
    
    \pgfmathsetmacro\xvSeven{\radius * cos(120)}
    \pgfmathsetmacro\yvSeven{\radius * sin(120)}
    \node[vertex] (v7) at (\xvSeven, \yvSeven) {$u_1^y$};
    
    \node[vertex] (v3) at ({\xvTwo + (1/7) * (\xvSeven - \xvTwo)}, {\yvTwo + (1/7) * (\yvSeven - \yvTwo)}) {$u_3^y$};
    \node[vertex] (v4) at ({\xvTwo + (3/7) * (\xvSeven - \xvTwo)}, {\yvTwo + (3/7) * (\yvSeven - \yvTwo)}) {$u_4^y$};
    \node[vertex] (v5) at ({\xvTwo + (4/7) * (\xvSeven - \xvTwo)}, {\yvTwo + (4/7) * (\yvSeven - \yvTwo)}) {$u_5^y$};
    \node[vertex] (v6) at ({\xvTwo + (6/7) * (\xvSeven - \xvTwo)}, {\yvTwo + (6/7) * (\yvSeven - \yvTwo)}) {$u_6^y$};

    \pgfmathsetmacro\dx{\xvSeven - \xvTwo}
    \pgfmathsetmacro\dy{\yvSeven - \yvTwo}
    \pgfmathsetmacro\len{sqrt(\dx*\dx + \dy*\dy)}
    \pgfmathsetmacro\nx{\dy / \len} 
    \pgfmathsetmacro\ny{-\dx / \len}

    \node[vertex] (mv2) at ({\xvTwo + 2*\gap*\nx}, {\yvTwo + 2*\gap*\ny}) {$y_2^u$};
    \node[vertex] (mv7) at ({\xvSeven + 2*\gap*\nx}, {\yvSeven + 2*\gap*\ny}) {$y_1^u$};
    \node[vertex] (mv3) at ({(\xvTwo + (1/7) * (\xvSeven - \xvTwo)) + 2*\gap*\nx}, 
                                {(\yvTwo + (1/7) * (\yvSeven - \yvTwo)) + 2*\gap*\ny}) {$y_3^u$};
    \node[vertex] (mv4) at ({(\xvTwo + (3/7) * (\xvSeven - \xvTwo)) + 2*\gap*\nx}, 
                                {(\yvTwo + (3/7) * (\yvSeven - \yvTwo)) + 2*\gap*\ny}) {$y_4^u$};
    \node[vertex] (mv5) at ({(\xvTwo + (4/7) * (\xvSeven - \xvTwo)) + 2*\gap*\nx}, 
                                {(\yvTwo + (4/7) * (\yvSeven - \yvTwo)) + 2*\gap*\ny}) {$y_5^u$};
    \node[vertex] (mv6) at ({(\xvTwo + (6/7) * (\xvSeven - \xvTwo)) + 2*\gap*\nx}, 
                                {(\yvTwo + (6/7) * (\yvSeven - \yvTwo)) + 2*\gap*\ny}) {$y_6^u$};

    \draw[wavey] (mv2) to[bend right=30] (mv7);
    \draw (mv2) -- (mv3); 
    \draw[wavey] (mv3) -- (mv4); 
    \draw (mv4) -- (mv5);
    \draw[wavey] (mv5) -- (mv6); 
    \draw (mv6) -- (mv7);

    \pgfmathsetmacro\xvEight{\radius * cos(140)}
    \pgfmathsetmacro\yvEight{\radius * sin(140)}
    \node[vertex] (v8) at (\xvEight, \yvEight) {$u_2^z$};

    \pgfmathsetmacro\xvThirteen{\radius * cos(240)}
    \pgfmathsetmacro\yvThirteen{\radius * sin(240)}
    \node[vertex] (v13) at (\xvThirteen, \yvThirteen) {$u_1^z$};

    \node[vertex] (v9) at ({\xvEight + (1/7) * (\xvThirteen - \xvEight)}, {\yvEight + (1/7) * (\yvThirteen - \yvEight)}) {$u_3^z$};
    \node[vertex] (v10) at ({\xvEight + (3/7) * (\xvThirteen - \xvEight)}, {\yvEight + (3/7) * (\yvThirteen - \yvEight)}) {$u_4^z$};
    \node[vertex] (v11) at ({\xvEight + (4/7) * (\xvThirteen - \xvEight)}, {\yvEight + (4/7) * (\yvThirteen - \yvEight)}) {$u_5^z$};
    \node[vertex] (v12) at ({\xvEight + (6/7) * (\xvThirteen - \xvEight)}, {\yvEight + (6/7) * (\yvThirteen - \yvEight)}) {$u_6^z$};

    \pgfmathsetmacro\dx{\xvThirteen - \xvEight}
    \pgfmathsetmacro\dy{\yvThirteen - \yvEight}
    \pgfmathsetmacro\len{sqrt(\dx*\dx + \dy*\dy)}
    \pgfmathsetmacro\nx{\dy / \len} 
    \pgfmathsetmacro\ny{-\dx / \len}

    \node[vertex] (mv8) at ({\xvEight + 2*\gap*\nx}, {\yvEight + 2*\gap*\ny}) {$z_2^u$};
    \node[vertex] (mv13) at ({\xvThirteen + 2*\gap*\nx}, {\yvThirteen + 2*\gap*\ny}) {$z_1^u$};
    \node[vertex] (mv9) at ({(\xvEight + (1/7) * (\xvThirteen - \xvEight)) + 2*\gap*\nx}, 
                             {(\yvEight + (1/7) * (\yvThirteen - \yvEight)) + 2*\gap*\ny}) {$z_3^u$};
    \node[vertex] (mv10) at ({(\xvEight + (3/7) * (\xvThirteen - \xvEight)) + 2*\gap*\nx}, 
                              {(\yvEight + (3/7) * (\yvThirteen - \yvEight)) + 2*\gap*\ny}) {$z_4^u$};
    \node[vertex] (mv11) at ({(\xvEight + (4/7) * (\xvThirteen - \xvEight)) + 2*\gap*\nx}, 
                              {(\yvEight + (4/7) * (\yvThirteen - \yvEight)) + 2*\gap*\ny}) {$z_5^u$};
    \node[vertex] (mv12) at ({(\xvEight + (6/7) * (\xvThirteen - \xvEight)) + 2*\gap*\nx}, 
                              {(\yvEight + (6/7) * (\yvThirteen - \yvEight)) + 2*\gap*\ny}) {$z_6^u$};

    \draw[wavey] (mv8) to[bend right=30] (mv13);
    \draw (mv8) -- (mv9);
    \draw[wavey] (mv9) -- (mv10); 
    \draw (mv10) -- (mv11);
    \draw[wavey] (mv11) -- (mv12); 
    \draw (mv12) -- (mv13);

    \draw (v1) -- (v2);
    \draw (v2) -- (v3); 
    \draw[wavey] (v3) -- (v4); 
    \draw (v4) -- (v5);
    \draw[wavey] (v5) -- (v6); 
    \draw (v6) -- (v7);
    \draw (v7) -- (v8); 
    \draw (v8) -- (v9);
    \draw[wavey] (v9) -- (v10); 
    \draw (v10) -- (v11);
    \draw[wavey] (v11) -- (v12); 
    \draw (v12) -- (v13);
    \draw (v13) -- (v14); 
    \draw[wavey] (v1) to[bend right=30] (v14);
    \draw[wavey] (v2) to[bend left=30] (v7);
    \draw[wavey] (v8) to[bend left=30] (v13);

    \draw (mv3) -- (v4);
    \draw (mv4) -- (v3);
    \draw (mv5) -- (v6);
    \draw (mv6) -- (v5);

    \draw (mv9) -- (v10);
    \draw (mv10) -- (v9);
    \draw (mv11) -- (v12);
    \draw (mv12) -- (v11);

    \draw (mv15) -- (v16);
    \draw (mv16) -- (v15);
    \draw (mv17) -- (v18);
    \draw (mv18) -- (v17);

\end{tikzpicture}
	\caption{The construction of the graph $H$. The figure illustrates the
	subgraph generated for a vertex $u \in V(G)$ incident to edges $ux$, $uy$, and $uz$.
	The remaining vertices in $X_x$, $X_y,$ and $X_z$ are omitted for clarity.}
\label{fig:vc-covercycle-transformation}
\end{figure}
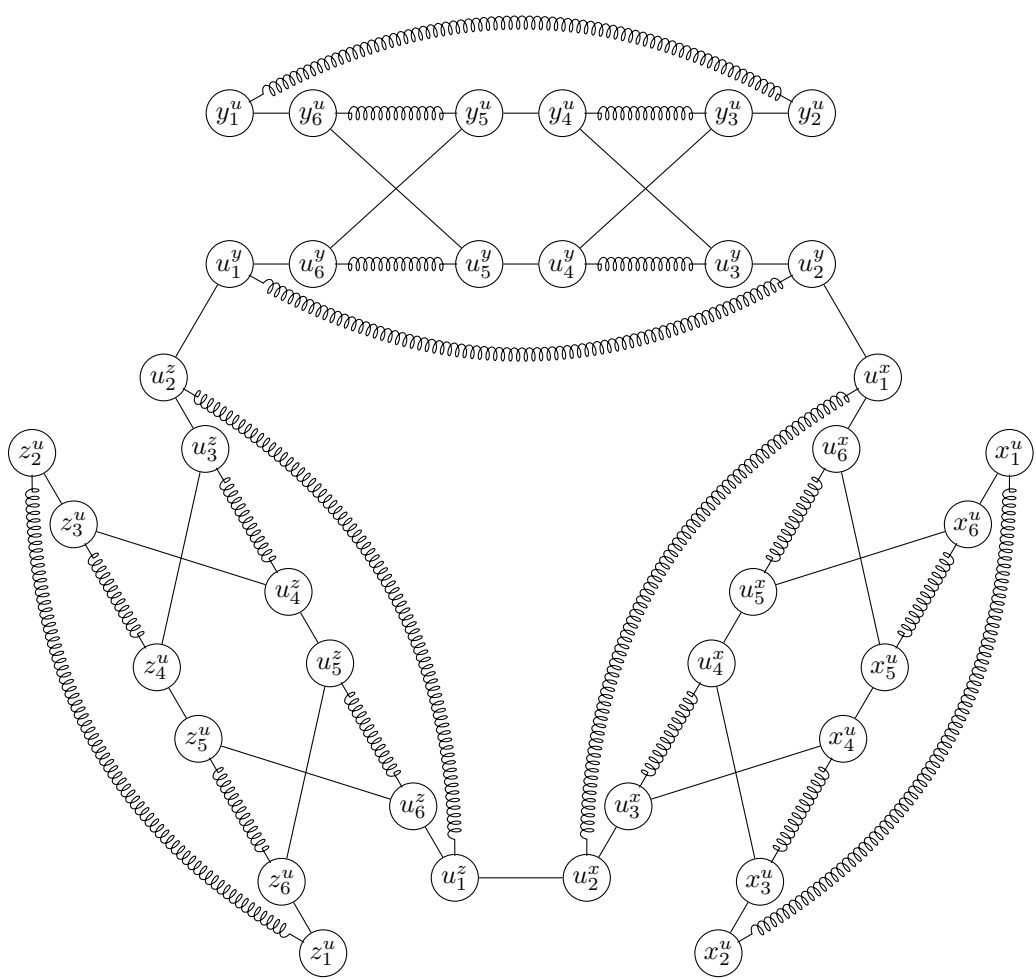

It follows directly by construction and property $(i)$ of chain gadgets that $H$ is cubic. Further, observe that the partition of $V(H)$ into $\{u_i^v: uv \in E(G), \text{$i$ odd }\}$ and $\{u_i^v: uv \in E(G), \text{$i$ even }\}$ is a bipartition of the graph obtained from $H$ by replacing all chain gadgets by edges. By property $(iii)$ of chain gadgets, this bipartition can be extended into a bipartition of $H$. It follows that $H$ is bipartite.

Observe that $V(H)$ consists of the vertex sets of exactly $9$ chain gadgets for every $v \in V(G)$; thus,
property $(ii)$ of chain gadgets yields $|V(H)| = 9(6k + 2) |V(G)| = \left(54 \max\{12, \ceil{\frac{2\alpha}{\beta}}\} + 18\right) |V(G)|$.

It remains to prove that for every positive integer $K$, we have that $H$ admits a cycle subpartition of weight at most $\alpha (|V(G)|+2K)$ if and only if $G$ admits a vertex cover of size at most $K$. It turns out that one of the two directions, namely constructing a cycle subpartition from a vertex cover is rather straight forward while it is significantly more difficult to construct a vertex cover from a cycle subpartition. 

We next introduce a little more notation. Namely, for some $x,y \in V(H)$ which are connected by a chain gadget and a cycle $C$ which is a subgraph of $H$, we say that $C$ contains $\overline{xy}$ if $C$ contains a Hamiltonian path of the chain gadget connecting $x$ and $y$.

We now define a certain collection of cycles in $H$ of four different types. It will turn out that cycle subpartitions of $H$ of minimum weight only contain cycles of one of these kinds. First, with the description of these cycles at hand, it will be easy to conclude the existence of a cycle subpartition of small weight of $H$ from the existence of a small vertex cover of $G$. Afterwards, given a cycle subpartition of $H$ of minimum weight, we use this minimality to show that all cycles of the cycle subpartition are of one of these types and that for every vertex in $V(G)$, some of these cycles intersect the corresponding gadget in a certain way. Once this is established, the remainder of the proof is not difficult.

For some $u \in V(G)$ with $N_G(u)=\{x,y,z\}$ and $u_2^xu_1^y \in E(H)$, a {\it 1-cycle for $u$} consists of $\overline{u_1^xu_2^x},u_2^xu_1^y,\overline{u_1^yu_2^y},u_2^yu_1^z,\overline{u_1^zu_2^z},$ and $u_2^zu_1^x$. Next for an ordered pair $(u,v)$ of vertices in $V(G)$ with $uv \in E(G)$, a {\it 2-cycle for $(u,v)$} consists of $u_1^vu_6^v,\overline{u_6^vu_5^v},u_5^vu_4^v,\overline{u_4^vu_3^v},u_3^vu_2^v$, and $\overline{u_2^vu_1^v}$ and a {\it 3-cycle} for $(u,v)$ consists of  $u_1^vu_6^v,\overline{u_6^vu_5^v},u_5^vv_6^u,\overline{v_6^uv_5^u},v_5^uv_4^u, \overline{v_4^uv_3^u}, v_3^{u}u_4^v,\overline{u_4^vu_3^v}, u_3^vu_2^v$, and $\overline{u_2^vu_1^v}$. Finally, for some $e=uv \in E(G)$, a {\it 4-cycle} for $e$ consists either of $\overline{u_3^vu_4^v},u_4^vv_3^{u},\overline{v_3^{u}v_4^{u}}$, and $v_4^{u}u_3^v$ or of $\overline{u_5^vu_6^v},u_6^vv_5^{u},\overline{v_5^{u}v_6^{u}}$, and $v_6^{u}u_5^v$.
An illustration of these cycles can be found in~\cref{fig:vc-covercycle-cycle-types}. Observe that a 1-cycle for a vertex $u$ and a 2-cycle or a 3-cycle for an ordered pair $(u,v)$ are uniquely defined except that there may be several Hamiltonian paths inside chain gadgets which will be of no relevance. For every edge $e \in E(G)$, there are two 4-cycles for $e$ which are distinct even when not taking into account different Hamiltonian paths of chain gadgets.

\begin{figure}[htbp]
    \centering
    \begin{subfigure}{0.45\linewidth}
		\tikzstyle{wavey}=[decorate,decoration={coil, aspect=-.5, post length=1mm, segment length=1mm, pre length=2mm},
shorten <= -.8pt,shorten >= -0.8pt]

\begin{tikzpicture}[rotate=20,scale=0.5,
    vertex/.style={
        draw,
        circle, 
        inner sep=.0pt, 
        minimum size=0.1cm},
    every edge/.append style={}]
    \def\radius{5}
    \def\gap{1}

    \pgfmathsetmacro\x{\radius * cos(0)}
    \pgfmathsetmacro\y{\radius * sin(0)}
    \node[vertex] (v1) at (\x,\y) {};
    
    \pgfmathsetmacro\xvFourteen{\radius * cos(260)}
    \pgfmathsetmacro\yvFourteen{\radius * sin(260)}
    \node[vertex] (v14) at (\xvFourteen, \yvFourteen) {};
    
    \node[vertex] (v15) at ({\xvFourteen + (1/7) * (\x - \xvFourteen)}, {\yvFourteen + (1/7) * (\y - \yvFourteen)}) {};
    \node[vertex] (v16) at ({\xvFourteen + (3/7) * (\x - \xvFourteen)}, {\yvFourteen + (3/7) * (\y - \yvFourteen)}) {};
    \node[vertex] (v17) at ({\xvFourteen + (4/7) * (\x - \xvFourteen)}, {\yvFourteen + (4/7) * (\y - \yvFourteen)}) {};
    \node[vertex] (v18) at ({\xvFourteen + (6/7) * (\x - \xvFourteen)}, {\yvFourteen + (6/7) * (\y - \yvFourteen)}) {};

    \pgfmathsetmacro\dx{\x - \xvFourteen}
    \pgfmathsetmacro\dy{\y - \yvFourteen}
    \pgfmathsetmacro\len{sqrt(\dx*\dx + \dy*\dy)}
    \pgfmathsetmacro\nx{\dy / \len} 
    \pgfmathsetmacro\ny{-\dx / \len}

    \node[vertex] (mv1) at ({\x + 2*\gap*\nx}, {\y + 2*\gap*\ny}) {};
    \node[vertex] (mv14) at ({\xvFourteen + 2*\gap*\nx}, {\yvFourteen + 2*\gap*\ny}) {};
    \node[vertex] (mv15) at ({(\xvFourteen + (1/7) * (\x - \xvFourteen)) + 2*\gap*\nx}, 
                              {(\yvFourteen + (1/7) * (\y - \yvFourteen)) + 2*\gap*\ny}) {};
    \node[vertex] (mv16) at ({(\xvFourteen + (3/7) * (\x - \xvFourteen)) + 2*\gap*\nx}, 
                              {(\yvFourteen + (3/7) * (\y - \yvFourteen)) + 2*\gap*\ny}) {};
    \node[vertex] (mv17) at ({(\xvFourteen + (4/7) * (\x - \xvFourteen)) + 2*\gap*\nx}, 
                              {(\yvFourteen + (4/7) * (\y - \yvFourteen)) + 2*\gap*\ny}) {};
    \node[vertex] (mv18) at ({(\xvFourteen + (6/7) * (\x - \xvFourteen)) + 2*\gap*\nx}, 
                              {(\yvFourteen + (6/7) * (\y - \yvFourteen)) + 2*\gap*\ny}) {};

    \draw (v14) -- (v15);
    \draw[wavey] (v15) -- (v16); 
    \draw (v16) -- (v17);
    \draw[wavey] (v17) -- (v18); 
    \draw (v18) -- (v1);

    \draw[wavey] (mv1) to[bend left=30] (mv14);
    \draw (mv14) -- (mv15);
    \draw[wavey] (mv15) -- (mv16); 
    \draw (mv16) -- (mv17);
    \draw[wavey] (mv17) -- (mv18); 
    \draw (mv18) -- (mv1);

    \pgfmathsetmacro\xvTwo{\radius * cos(20)}
    \pgfmathsetmacro\yvTwo{\radius * sin(20)}
    \node[vertex] (v2) at (\xvTwo, \yvTwo) {};
    
    \pgfmathsetmacro\xvSeven{\radius * cos(120)}
    \pgfmathsetmacro\yvSeven{\radius * sin(120)}
    \node[vertex] (v7) at (\xvSeven, \yvSeven) {};
    
    \node[vertex] (v3) at ({\xvTwo + (1/7) * (\xvSeven - \xvTwo)}, {\yvTwo + (1/7) * (\yvSeven - \yvTwo)}) {};
    \node[vertex] (v4) at ({\xvTwo + (3/7) * (\xvSeven - \xvTwo)}, {\yvTwo + (3/7) * (\yvSeven - \yvTwo)}) {};
    \node[vertex] (v5) at ({\xvTwo + (4/7) * (\xvSeven - \xvTwo)}, {\yvTwo + (4/7) * (\yvSeven - \yvTwo)}) {};
    \node[vertex] (v6) at ({\xvTwo + (6/7) * (\xvSeven - \xvTwo)}, {\yvTwo + (6/7) * (\yvSeven - \yvTwo)}) {};

    \pgfmathsetmacro\dx{\xvSeven - \xvTwo}
    \pgfmathsetmacro\dy{\yvSeven - \yvTwo}
    \pgfmathsetmacro\len{sqrt(\dx*\dx + \dy*\dy)}
    \pgfmathsetmacro\nx{\dy / \len} 
    \pgfmathsetmacro\ny{-\dx / \len}

    \node[vertex] (mv2) at ({\xvTwo + 2*\gap*\nx}, {\yvTwo + 2*\gap*\ny}) {};
    \node[vertex] (mv7) at ({\xvSeven + 2*\gap*\nx}, {\yvSeven + 2*\gap*\ny}) {};
    \node[vertex] (mv3) at ({(\xvTwo + (1/7) * (\xvSeven - \xvTwo)) + 2*\gap*\nx}, 
                                {(\yvTwo + (1/7) * (\yvSeven - \yvTwo)) + 2*\gap*\ny}) {};
    \node[vertex] (mv4) at ({(\xvTwo + (3/7) * (\xvSeven - \xvTwo)) + 2*\gap*\nx}, 
                                {(\yvTwo + (3/7) * (\yvSeven - \yvTwo)) + 2*\gap*\ny}) {};
    \node[vertex] (mv5) at ({(\xvTwo + (4/7) * (\xvSeven - \xvTwo)) + 2*\gap*\nx}, 
                                {(\yvTwo + (4/7) * (\yvSeven - \yvTwo)) + 2*\gap*\ny}) {};
    \node[vertex] (mv6) at ({(\xvTwo + (6/7) * (\xvSeven - \xvTwo)) + 2*\gap*\nx}, 
                                {(\yvTwo + (6/7) * (\yvSeven - \yvTwo)) + 2*\gap*\ny}) {};

    \draw[wavey] (mv2) to[bend right=30] (mv7);
    \draw (mv2) -- (mv3); 
    \draw[wavey] (mv3) -- (mv4); 
    \draw (mv4) -- (mv5);
    \draw[wavey] (mv5) -- (mv6); 
    \draw (mv6) -- (mv7);

    \pgfmathsetmacro\xvEight{\radius * cos(140)}
    \pgfmathsetmacro\yvEight{\radius * sin(140)}
    \node[vertex] (v8) at (\xvEight, \yvEight) {};

    \pgfmathsetmacro\xvThirteen{\radius * cos(240)}
    \pgfmathsetmacro\yvThirteen{\radius * sin(240)}
    \node[vertex] (v13) at (\xvThirteen, \yvThirteen) {};

    \node[vertex] (v9) at ({\xvEight + (1/7) * (\xvThirteen - \xvEight)}, {\yvEight + (1/7) * (\yvThirteen - \yvEight)}) {};
    \node[vertex] (v10) at ({\xvEight + (3/7) * (\xvThirteen - \xvEight)}, {\yvEight + (3/7) * (\yvThirteen - \yvEight)}) {};
    \node[vertex] (v11) at ({\xvEight + (4/7) * (\xvThirteen - \xvEight)}, {\yvEight + (4/7) * (\yvThirteen - \yvEight)}) {};
    \node[vertex] (v12) at ({\xvEight + (6/7) * (\xvThirteen - \xvEight)}, {\yvEight + (6/7) * (\yvThirteen - \yvEight)}) {};

    \pgfmathsetmacro\dx{\xvThirteen - \xvEight}
    \pgfmathsetmacro\dy{\yvThirteen - \yvEight}
    \pgfmathsetmacro\len{sqrt(\dx*\dx + \dy*\dy)}
    \pgfmathsetmacro\nx{\dy / \len} 
    \pgfmathsetmacro\ny{-\dx / \len}

    \node[vertex] (mv8) at ({\xvEight + 2*\gap*\nx}, {\yvEight + 2*\gap*\ny}) {};
    \node[vertex] (mv13) at ({\xvThirteen + 2*\gap*\nx}, {\yvThirteen + 2*\gap*\ny}) {};
    \node[vertex] (mv9) at ({(\xvEight + (1/7) * (\xvThirteen - \xvEight)) + 2*\gap*\nx}, 
                             {(\yvEight + (1/7) * (\yvThirteen - \yvEight)) + 2*\gap*\ny}) {};
    \node[vertex] (mv10) at ({(\xvEight + (3/7) * (\xvThirteen - \xvEight)) + 2*\gap*\nx}, 
                              {(\yvEight + (3/7) * (\yvThirteen - \yvEight)) + 2*\gap*\ny}) {};
    \node[vertex] (mv11) at ({(\xvEight + (4/7) * (\xvThirteen - \xvEight)) + 2*\gap*\nx}, 
                              {(\yvEight + (4/7) * (\yvThirteen - \yvEight)) + 2*\gap*\ny}) {};
    \node[vertex] (mv12) at ({(\xvEight + (6/7) * (\xvThirteen - \xvEight)) + 2*\gap*\nx}, 
                              {(\yvEight + (6/7) * (\yvThirteen - \yvEight)) + 2*\gap*\ny}) {};

    \draw[wavey] (mv8) to[bend right=30] (mv13);
    \draw (mv8) -- (mv9);
    \draw[wavey] (mv9) -- (mv10); 
    \draw (mv10) -- (mv11);
    \draw[wavey] (mv11) -- (mv12); 
    \draw (mv12) -- (mv13);

    \draw[very thick] (v1) -- (v2);
    \draw (v2) -- (v3); 
    \draw[wavey] (v3) -- (v4); 
    \draw (v4) -- (v5);
    \draw[wavey] (v5) -- (v6); 
    \draw (v6) -- (v7);
    \draw[very thick] (v7) -- (v8); 
    \draw (v8) -- (v9);
    \draw[wavey] (v9) -- (v10); 
    \draw (v10) -- (v11);
    \draw[wavey] (v11) -- (v12); 
    \draw (v12) -- (v13);
    \draw[very thick] (v13) -- (v14); 
    \draw[very thick, wavey] (v1) to[bend right=30] (v14);
    \draw[very thick, wavey] (v2) to[bend left=30] (v7);
    \draw[very thick, wavey] (v8) to[bend left=30] (v13);

    \draw (mv3) -- (v4);
    \draw (mv4) -- (v3);
    \draw (mv5) -- (v6);
    \draw (mv6) -- (v5);

    \draw (mv9) -- (v10);
    \draw (mv10) -- (v9);
    \draw (mv11) -- (v12);
    \draw (mv12) -- (v11);

    \draw (mv15) -- (v16);
    \draw (mv16) -- (v15);
    \draw (mv17) -- (v18);
    \draw (mv18) -- (v17);
\end{tikzpicture}
        \caption{The $1$-cycle for a vertex $u \in V(G)$.}
    \end{subfigure}%
    \hspace{0.05\linewidth}
    \begin{subfigure}{0.45\linewidth}
		\tikzstyle{wavey}=[decorate,decoration={coil, aspect=-.5, post length=1mm, segment length=1mm, pre length=2mm},
shorten <= -.8pt,shorten >= -0.8pt]

\begin{tikzpicture}[rotate=20,scale=0.5,
    vertex/.style={
        draw,
        circle, 
        inner sep=.0pt, 
        minimum size=0.1cm},
    every edge/.append style={}]
    \def\radius{5}
    \def\gap{1}

    \pgfmathsetmacro\x{\radius * cos(0)}
    \pgfmathsetmacro\y{\radius * sin(0)}
    \node[vertex] (v1) at (\x,\y) {};
    
    \pgfmathsetmacro\xvFourteen{\radius * cos(260)}
    \pgfmathsetmacro\yvFourteen{\radius * sin(260)}
    \node[vertex] (v14) at (\xvFourteen, \yvFourteen) {};
    
    \node[vertex] (v15) at ({\xvFourteen + (1/7) * (\x - \xvFourteen)}, {\yvFourteen + (1/7) * (\y - \yvFourteen)}) {};
    \node[vertex] (v16) at ({\xvFourteen + (3/7) * (\x - \xvFourteen)}, {\yvFourteen + (3/7) * (\y - \yvFourteen)}) {};
    \node[vertex] (v17) at ({\xvFourteen + (4/7) * (\x - \xvFourteen)}, {\yvFourteen + (4/7) * (\y - \yvFourteen)}) {};
    \node[vertex] (v18) at ({\xvFourteen + (6/7) * (\x - \xvFourteen)}, {\yvFourteen + (6/7) * (\y - \yvFourteen)}) {};

    \pgfmathsetmacro\dx{\x - \xvFourteen}
    \pgfmathsetmacro\dy{\y - \yvFourteen}
    \pgfmathsetmacro\len{sqrt(\dx*\dx + \dy*\dy)}
    \pgfmathsetmacro\nx{\dy / \len} 
    \pgfmathsetmacro\ny{-\dx / \len}

    \node[vertex] (mv1) at ({\x + 2*\gap*\nx}, {\y + 2*\gap*\ny}) {};
    \node[vertex] (mv14) at ({\xvFourteen + 2*\gap*\nx}, {\yvFourteen + 2*\gap*\ny}) {};
    \node[vertex] (mv15) at ({(\xvFourteen + (1/7) * (\x - \xvFourteen)) + 2*\gap*\nx}, 
                              {(\yvFourteen + (1/7) * (\y - \yvFourteen)) + 2*\gap*\ny}) {};
    \node[vertex] (mv16) at ({(\xvFourteen + (3/7) * (\x - \xvFourteen)) + 2*\gap*\nx}, 
                              {(\yvFourteen + (3/7) * (\y - \yvFourteen)) + 2*\gap*\ny}) {};
    \node[vertex] (mv17) at ({(\xvFourteen + (4/7) * (\x - \xvFourteen)) + 2*\gap*\nx}, 
                              {(\yvFourteen + (4/7) * (\y - \yvFourteen)) + 2*\gap*\ny}) {};
    \node[vertex] (mv18) at ({(\xvFourteen + (6/7) * (\x - \xvFourteen)) + 2*\gap*\nx}, 
                              {(\yvFourteen + (6/7) * (\y - \yvFourteen)) + 2*\gap*\ny}) {};

    \draw[very thick] (v14) -- (v15);
    \draw[very thick, wavey] (v15) -- (v16); 
    \draw[very thick] (v16) -- (v17);
    \draw[very thick, wavey] (v17) -- (v18); 
    \draw[very thick] (v18) -- (v1);

    \draw[wavey] (mv1) to[bend left=30] (mv14);
    \draw (mv14) -- (mv15);
    \draw[wavey] (mv15) -- (mv16); 
    \draw (mv16) -- (mv17);
    \draw[wavey] (mv17) -- (mv18); 
    \draw (mv18) -- (mv1);

    \pgfmathsetmacro\xvTwo{\radius * cos(20)}
    \pgfmathsetmacro\yvTwo{\radius * sin(20)}
    \node[vertex] (v2) at (\xvTwo, \yvTwo) {};
    
    \pgfmathsetmacro\xvSeven{\radius * cos(120)}
    \pgfmathsetmacro\yvSeven{\radius * sin(120)}
    \node[vertex] (v7) at (\xvSeven, \yvSeven) {};
    
    \node[vertex] (v3) at ({\xvTwo + (1/7) * (\xvSeven - \xvTwo)}, {\yvTwo + (1/7) * (\yvSeven - \yvTwo)}) {};
    \node[vertex] (v4) at ({\xvTwo + (3/7) * (\xvSeven - \xvTwo)}, {\yvTwo + (3/7) * (\yvSeven - \yvTwo)}) {};
    \node[vertex] (v5) at ({\xvTwo + (4/7) * (\xvSeven - \xvTwo)}, {\yvTwo + (4/7) * (\yvSeven - \yvTwo)}) {};
    \node[vertex] (v6) at ({\xvTwo + (6/7) * (\xvSeven - \xvTwo)}, {\yvTwo + (6/7) * (\yvSeven - \yvTwo)}) {};

    \pgfmathsetmacro\dx{\xvSeven - \xvTwo}
    \pgfmathsetmacro\dy{\yvSeven - \yvTwo}
    \pgfmathsetmacro\len{sqrt(\dx*\dx + \dy*\dy)}
    \pgfmathsetmacro\nx{\dy / \len} 
    \pgfmathsetmacro\ny{-\dx / \len}

    \node[vertex] (mv2) at ({\xvTwo + 2*\gap*\nx}, {\yvTwo + 2*\gap*\ny}) {};
    \node[vertex] (mv7) at ({\xvSeven + 2*\gap*\nx}, {\yvSeven + 2*\gap*\ny}) {};
    \node[vertex] (mv3) at ({(\xvTwo + (1/7) * (\xvSeven - \xvTwo)) + 2*\gap*\nx}, 
                                {(\yvTwo + (1/7) * (\yvSeven - \yvTwo)) + 2*\gap*\ny}) {};
    \node[vertex] (mv4) at ({(\xvTwo + (3/7) * (\xvSeven - \xvTwo)) + 2*\gap*\nx}, 
                                {(\yvTwo + (3/7) * (\yvSeven - \yvTwo)) + 2*\gap*\ny}) {};
    \node[vertex] (mv5) at ({(\xvTwo + (4/7) * (\xvSeven - \xvTwo)) + 2*\gap*\nx}, 
                                {(\yvTwo + (4/7) * (\yvSeven - \yvTwo)) + 2*\gap*\ny}) {};
    \node[vertex] (mv6) at ({(\xvTwo + (6/7) * (\xvSeven - \xvTwo)) + 2*\gap*\nx}, 
                                {(\yvTwo + (6/7) * (\yvSeven - \yvTwo)) + 2*\gap*\ny}) {};

    \draw[wavey] (mv2) to[bend right=30] (mv7);
    \draw (mv2) -- (mv3); 
    \draw[wavey] (mv3) -- (mv4); 
    \draw (mv4) -- (mv5);
    \draw[wavey] (mv5) -- (mv6); 
    \draw (mv6) -- (mv7);

    \pgfmathsetmacro\xvEight{\radius * cos(140)}
    \pgfmathsetmacro\yvEight{\radius * sin(140)}
    \node[vertex] (v8) at (\xvEight, \yvEight) {};

    \pgfmathsetmacro\xvThirteen{\radius * cos(240)}
    \pgfmathsetmacro\yvThirteen{\radius * sin(240)}
    \node[vertex] (v13) at (\xvThirteen, \yvThirteen) {};

    \node[vertex] (v9) at ({\xvEight + (1/7) * (\xvThirteen - \xvEight)}, {\yvEight + (1/7) * (\yvThirteen - \yvEight)}) {};
    \node[vertex] (v10) at ({\xvEight + (3/7) * (\xvThirteen - \xvEight)}, {\yvEight + (3/7) * (\yvThirteen - \yvEight)}) {};
    \node[vertex] (v11) at ({\xvEight + (4/7) * (\xvThirteen - \xvEight)}, {\yvEight + (4/7) * (\yvThirteen - \yvEight)}) {};
    \node[vertex] (v12) at ({\xvEight + (6/7) * (\xvThirteen - \xvEight)}, {\yvEight + (6/7) * (\yvThirteen - \yvEight)}) {};

    \pgfmathsetmacro\dx{\xvThirteen - \xvEight}
    \pgfmathsetmacro\dy{\yvThirteen - \yvEight}
    \pgfmathsetmacro\len{sqrt(\dx*\dx + \dy*\dy)}
    \pgfmathsetmacro\nx{\dy / \len} 
    \pgfmathsetmacro\ny{-\dx / \len}

    \node[vertex] (mv8) at ({\xvEight + 2*\gap*\nx}, {\yvEight + 2*\gap*\ny}) {};
    \node[vertex] (mv13) at ({\xvThirteen + 2*\gap*\nx}, {\yvThirteen + 2*\gap*\ny}) {};
    \node[vertex] (mv9) at ({(\xvEight + (1/7) * (\xvThirteen - \xvEight)) + 2*\gap*\nx}, 
                             {(\yvEight + (1/7) * (\yvThirteen - \yvEight)) + 2*\gap*\ny}) {};
    \node[vertex] (mv10) at ({(\xvEight + (3/7) * (\xvThirteen - \xvEight)) + 2*\gap*\nx}, 
                              {(\yvEight + (3/7) * (\yvThirteen - \yvEight)) + 2*\gap*\ny}) {};
    \node[vertex] (mv11) at ({(\xvEight + (4/7) * (\xvThirteen - \xvEight)) + 2*\gap*\nx}, 
                              {(\yvEight + (4/7) * (\yvThirteen - \yvEight)) + 2*\gap*\ny}) {};
    \node[vertex] (mv12) at ({(\xvEight + (6/7) * (\xvThirteen - \xvEight)) + 2*\gap*\nx}, 
                              {(\yvEight + (6/7) * (\yvThirteen - \yvEight)) + 2*\gap*\ny}) {};

    \draw[wavey] (mv8) to[bend right=30] (mv13);
    \draw (mv8) -- (mv9);
    \draw[wavey] (mv9) -- (mv10); 
    \draw (mv10) -- (mv11);
    \draw[wavey] (mv11) -- (mv12); 
    \draw (mv12) -- (mv13);

    \draw (v1) -- (v2);
    \draw (v2) -- (v3); 
    \draw[wavey] (v3) -- (v4); 
    \draw (v4) -- (v5);
    \draw[wavey] (v5) -- (v6); 
    \draw (v6) -- (v7);
    \draw (v7) -- (v8); 
    \draw (v8) -- (v9);
    \draw[wavey] (v9) -- (v10); 
    \draw (v10) -- (v11);
    \draw[wavey] (v11) -- (v12); 
    \draw (v12) -- (v13);
    \draw (v13) -- (v14); 
    \draw[very thick, wavey] (v1) to[bend right=30] (v14);
    \draw[wavey] (v2) to[bend left=30] (v7);
    \draw[wavey] (v8) to[bend left=30] (v13);

    \draw (mv3) -- (v4);
    \draw (mv4) -- (v3);
    \draw (mv5) -- (v6);
    \draw (mv6) -- (v5);

    \draw (mv9) -- (v10);
    \draw (mv10) -- (v9);
    \draw (mv11) -- (v12);
    \draw (mv12) -- (v11);

    \draw (mv15) -- (v16);
    \draw (mv16) -- (v15);
    \draw (mv17) -- (v18);
    \draw (mv18) -- (v17);
\end{tikzpicture}
        \caption{The $2$-cycle for a pair $(u, x)$ with $ux \in E(G)$.}
    \end{subfigure}
    
    \vspace{0.05cm}
    
    \begin{subfigure}{0.45\linewidth}
		\tikzstyle{wavey}=[decorate,decoration={coil, aspect=-.5, post length=1mm, segment length=1mm, pre length=2mm},
shorten <= -.8pt,shorten >= -0.8pt]

\begin{tikzpicture}[rotate=20,scale=0.5,
    vertex/.style={
        draw,
        circle, 
        inner sep=.0pt, 
        minimum size=0.1cm},
    every edge/.append style={}]
    \def\radius{5}
    \def\gap{1}

    \pgfmathsetmacro\x{\radius * cos(0)}
    \pgfmathsetmacro\y{\radius * sin(0)}
    \node[vertex] (v1) at (\x,\y) {};
    
    \pgfmathsetmacro\xvFourteen{\radius * cos(260)}
    \pgfmathsetmacro\yvFourteen{\radius * sin(260)}
    \node[vertex] (v14) at (\xvFourteen, \yvFourteen) {};
    
    \node[vertex] (v15) at ({\xvFourteen + (1/7) * (\x - \xvFourteen)}, {\yvFourteen + (1/7) * (\y - \yvFourteen)}) {};
    \node[vertex] (v16) at ({\xvFourteen + (3/7) * (\x - \xvFourteen)}, {\yvFourteen + (3/7) * (\y - \yvFourteen)}) {};
    \node[vertex] (v17) at ({\xvFourteen + (4/7) * (\x - \xvFourteen)}, {\yvFourteen + (4/7) * (\y - \yvFourteen)}) {};
    \node[vertex] (v18) at ({\xvFourteen + (6/7) * (\x - \xvFourteen)}, {\yvFourteen + (6/7) * (\y - \yvFourteen)}) {};

    \pgfmathsetmacro\dx{\x - \xvFourteen}
    \pgfmathsetmacro\dy{\y - \yvFourteen}
    \pgfmathsetmacro\len{sqrt(\dx*\dx + \dy*\dy)}
    \pgfmathsetmacro\nx{\dy / \len} 
    \pgfmathsetmacro\ny{-\dx / \len}

    \node[vertex] (mv1) at ({\x + 2*\gap*\nx}, {\y + 2*\gap*\ny}) {};
    \node[vertex] (mv14) at ({\xvFourteen + 2*\gap*\nx}, {\yvFourteen + 2*\gap*\ny}) {};
    \node[vertex] (mv15) at ({(\xvFourteen + (1/7) * (\x - \xvFourteen)) + 2*\gap*\nx}, 
                              {(\yvFourteen + (1/7) * (\y - \yvFourteen)) + 2*\gap*\ny}) {};
    \node[vertex] (mv16) at ({(\xvFourteen + (3/7) * (\x - \xvFourteen)) + 2*\gap*\nx}, 
                              {(\yvFourteen + (3/7) * (\y - \yvFourteen)) + 2*\gap*\ny}) {};
    \node[vertex] (mv17) at ({(\xvFourteen + (4/7) * (\x - \xvFourteen)) + 2*\gap*\nx}, 
                              {(\yvFourteen + (4/7) * (\y - \yvFourteen)) + 2*\gap*\ny}) {};
    \node[vertex] (mv18) at ({(\xvFourteen + (6/7) * (\x - \xvFourteen)) + 2*\gap*\nx}, 
                              {(\yvFourteen + (6/7) * (\y - \yvFourteen)) + 2*\gap*\ny}) {};

    \draw[very thick] (v14) -- (v15);
    \draw[very thick, wavey] (v15) -- (v16); 
    \draw (v16) -- (v17);
    \draw[very thick, wavey] (v17) -- (v18); 
    \draw[very thick] (v18) -- (v1);

    \draw[wavey] (mv1) to[bend left=30] (mv14);
    \draw (mv14) -- (mv15);
    \draw[very thick, wavey] (mv15) -- (mv16); 
    \draw[very thick] (mv16) -- (mv17);
    \draw[very thick, wavey] (mv17) -- (mv18); 
    \draw (mv18) -- (mv1);

    \pgfmathsetmacro\xvTwo{\radius * cos(20)}
    \pgfmathsetmacro\yvTwo{\radius * sin(20)}
    \node[vertex] (v2) at (\xvTwo, \yvTwo) {};
    
    \pgfmathsetmacro\xvSeven{\radius * cos(120)}
    \pgfmathsetmacro\yvSeven{\radius * sin(120)}
    \node[vertex] (v7) at (\xvSeven, \yvSeven) {};
    
    \node[vertex] (v3) at ({\xvTwo + (1/7) * (\xvSeven - \xvTwo)}, {\yvTwo + (1/7) * (\yvSeven - \yvTwo)}) {};
    \node[vertex] (v4) at ({\xvTwo + (3/7) * (\xvSeven - \xvTwo)}, {\yvTwo + (3/7) * (\yvSeven - \yvTwo)}) {};
    \node[vertex] (v5) at ({\xvTwo + (4/7) * (\xvSeven - \xvTwo)}, {\yvTwo + (4/7) * (\yvSeven - \yvTwo)}) {};
    \node[vertex] (v6) at ({\xvTwo + (6/7) * (\xvSeven - \xvTwo)}, {\yvTwo + (6/7) * (\yvSeven - \yvTwo)}) {};

    \pgfmathsetmacro\dx{\xvSeven - \xvTwo}
    \pgfmathsetmacro\dy{\yvSeven - \yvTwo}
    \pgfmathsetmacro\len{sqrt(\dx*\dx + \dy*\dy)}
    \pgfmathsetmacro\nx{\dy / \len} 
    \pgfmathsetmacro\ny{-\dx / \len}

    \node[vertex] (mv2) at ({\xvTwo + 2*\gap*\nx}, {\yvTwo + 2*\gap*\ny}) {};
    \node[vertex] (mv7) at ({\xvSeven + 2*\gap*\nx}, {\yvSeven + 2*\gap*\ny}) {};
    \node[vertex] (mv3) at ({(\xvTwo + (1/7) * (\xvSeven - \xvTwo)) + 2*\gap*\nx}, 
                                {(\yvTwo + (1/7) * (\yvSeven - \yvTwo)) + 2*\gap*\ny}) {};
    \node[vertex] (mv4) at ({(\xvTwo + (3/7) * (\xvSeven - \xvTwo)) + 2*\gap*\nx}, 
                                {(\yvTwo + (3/7) * (\yvSeven - \yvTwo)) + 2*\gap*\ny}) {};
    \node[vertex] (mv5) at ({(\xvTwo + (4/7) * (\xvSeven - \xvTwo)) + 2*\gap*\nx}, 
                                {(\yvTwo + (4/7) * (\yvSeven - \yvTwo)) + 2*\gap*\ny}) {};
    \node[vertex] (mv6) at ({(\xvTwo + (6/7) * (\xvSeven - \xvTwo)) + 2*\gap*\nx}, 
                                {(\yvTwo + (6/7) * (\yvSeven - \yvTwo)) + 2*\gap*\ny}) {};

    \draw[wavey] (mv2) to[bend right=30] (mv7);
    \draw (mv2) -- (mv3); 
    \draw[wavey] (mv3) -- (mv4); 
    \draw (mv4) -- (mv5);
    \draw[wavey] (mv5) -- (mv6); 
    \draw (mv6) -- (mv7);

    \pgfmathsetmacro\xvEight{\radius * cos(140)}
    \pgfmathsetmacro\yvEight{\radius * sin(140)}
    \node[vertex] (v8) at (\xvEight, \yvEight) {};

    \pgfmathsetmacro\xvThirteen{\radius * cos(240)}
    \pgfmathsetmacro\yvThirteen{\radius * sin(240)}
    \node[vertex] (v13) at (\xvThirteen, \yvThirteen) {};

    \node[vertex] (v9) at ({\xvEight + (1/7) * (\xvThirteen - \xvEight)}, {\yvEight + (1/7) * (\yvThirteen - \yvEight)}) {};
    \node[vertex] (v10) at ({\xvEight + (3/7) * (\xvThirteen - \xvEight)}, {\yvEight + (3/7) * (\yvThirteen - \yvEight)}) {};
    \node[vertex] (v11) at ({\xvEight + (4/7) * (\xvThirteen - \xvEight)}, {\yvEight + (4/7) * (\yvThirteen - \yvEight)}) {};
    \node[vertex] (v12) at ({\xvEight + (6/7) * (\xvThirteen - \xvEight)}, {\yvEight + (6/7) * (\yvThirteen - \yvEight)}) {};

    \pgfmathsetmacro\dx{\xvThirteen - \xvEight}
    \pgfmathsetmacro\dy{\yvThirteen - \yvEight}
    \pgfmathsetmacro\len{sqrt(\dx*\dx + \dy*\dy)}
    \pgfmathsetmacro\nx{\dy / \len} 
    \pgfmathsetmacro\ny{-\dx / \len}

    \node[vertex] (mv8) at ({\xvEight + 2*\gap*\nx}, {\yvEight + 2*\gap*\ny}) {};
    \node[vertex] (mv13) at ({\xvThirteen + 2*\gap*\nx}, {\yvThirteen + 2*\gap*\ny}) {};
    \node[vertex] (mv9) at ({(\xvEight + (1/7) * (\xvThirteen - \xvEight)) + 2*\gap*\nx}, 
                             {(\yvEight + (1/7) * (\yvThirteen - \yvEight)) + 2*\gap*\ny}) {};
    \node[vertex] (mv10) at ({(\xvEight + (3/7) * (\xvThirteen - \xvEight)) + 2*\gap*\nx}, 
                              {(\yvEight + (3/7) * (\yvThirteen - \yvEight)) + 2*\gap*\ny}) {};
    \node[vertex] (mv11) at ({(\xvEight + (4/7) * (\xvThirteen - \xvEight)) + 2*\gap*\nx}, 
                              {(\yvEight + (4/7) * (\yvThirteen - \yvEight)) + 2*\gap*\ny}) {};
    \node[vertex] (mv12) at ({(\xvEight + (6/7) * (\xvThirteen - \xvEight)) + 2*\gap*\nx}, 
                              {(\yvEight + (6/7) * (\yvThirteen - \yvEight)) + 2*\gap*\ny}) {};

    \draw[wavey] (mv8) to[bend right=30] (mv13);
    \draw (mv8) -- (mv9);
    \draw[wavey] (mv9) -- (mv10); 
    \draw (mv10) -- (mv11);
    \draw[wavey] (mv11) -- (mv12); 
    \draw (mv12) -- (mv13);

    \draw (v1) -- (v2);
    \draw (v2) -- (v3); 
    \draw[wavey] (v3) -- (v4); 
    \draw (v4) -- (v5);
    \draw[wavey] (v5) -- (v6); 
    \draw (v6) -- (v7);
    \draw (v7) -- (v8); 
    \draw (v8) -- (v9);
    \draw[wavey] (v9) -- (v10); 
    \draw (v10) -- (v11);
    \draw[wavey] (v11) -- (v12); 
    \draw (v12) -- (v13);
    \draw (v13) -- (v14); 
    \draw[very thick, wavey] (v1) to[bend right=30] (v14);
    \draw[wavey] (v2) to[bend left=30] (v7);
    \draw[wavey] (v8) to[bend left=30] (v13);

    \draw (mv3) -- (v4);
    \draw (mv4) -- (v3);
    \draw (mv5) -- (v6);
    \draw (mv6) -- (v5);

    \draw (mv9) -- (v10);
    \draw (mv10) -- (v9);
    \draw (mv11) -- (v12);
    \draw (mv12) -- (v11);

    \draw[very thick] (mv15) -- (v16);
    \draw (mv16) -- (v15);
    \draw (mv17) -- (v18);
    \draw[very thick] (mv18) -- (v17);
\end{tikzpicture}
        \caption{The $3$-cycle for a pair $(u, x)$ with $ux \in E(G)$.}
    \end{subfigure}%
    \hspace{0.05\linewidth}
    \begin{subfigure}{0.45\linewidth}
		\tikzstyle{wavey}=[decorate,decoration={coil, aspect=-.5, post length=1mm, segment length=1mm, pre length=2mm},
shorten <= -.8pt,shorten >= -0.8pt]

\begin{tikzpicture}[rotate=20,scale=0.5,
    vertex/.style={
        draw,
        circle, 
        inner sep=.0pt, 
        minimum size=0.1cm},
    every edge/.append style={}]
    \def\radius{5}
    \def\gap{1}

    \pgfmathsetmacro\x{\radius * cos(0)}
    \pgfmathsetmacro\y{\radius * sin(0)}
    \node[vertex] (v1) at (\x,\y) {};
    
    \pgfmathsetmacro\xvFourteen{\radius * cos(260)}
    \pgfmathsetmacro\yvFourteen{\radius * sin(260)}
    \node[vertex] (v14) at (\xvFourteen, \yvFourteen) {};
    
    \node[vertex] (v15) at ({\xvFourteen + (1/7) * (\x - \xvFourteen)}, {\yvFourteen + (1/7) * (\y - \yvFourteen)}) {};
    \node[vertex] (v16) at ({\xvFourteen + (3/7) * (\x - \xvFourteen)}, {\yvFourteen + (3/7) * (\y - \yvFourteen)}) {};
    \node[vertex] (v17) at ({\xvFourteen + (4/7) * (\x - \xvFourteen)}, {\yvFourteen + (4/7) * (\y - \yvFourteen)}) {};
    \node[vertex] (v18) at ({\xvFourteen + (6/7) * (\x - \xvFourteen)}, {\yvFourteen + (6/7) * (\y - \yvFourteen)}) {};

    \pgfmathsetmacro\dx{\x - \xvFourteen}
    \pgfmathsetmacro\dy{\y - \yvFourteen}
    \pgfmathsetmacro\len{sqrt(\dx*\dx + \dy*\dy)}
    \pgfmathsetmacro\nx{\dy / \len} 
    \pgfmathsetmacro\ny{-\dx / \len}

    \node[vertex] (mv1) at ({\x + 2*\gap*\nx}, {\y + 2*\gap*\ny}) {};
    \node[vertex] (mv14) at ({\xvFourteen + 2*\gap*\nx}, {\yvFourteen + 2*\gap*\ny}) {};
    \node[vertex] (mv15) at ({(\xvFourteen + (1/7) * (\x - \xvFourteen)) + 2*\gap*\nx}, 
                              {(\yvFourteen + (1/7) * (\y - \yvFourteen)) + 2*\gap*\ny}) {};
    \node[vertex] (mv16) at ({(\xvFourteen + (3/7) * (\x - \xvFourteen)) + 2*\gap*\nx}, 
                              {(\yvFourteen + (3/7) * (\y - \yvFourteen)) + 2*\gap*\ny}) {};
    \node[vertex] (mv17) at ({(\xvFourteen + (4/7) * (\x - \xvFourteen)) + 2*\gap*\nx}, 
                              {(\yvFourteen + (4/7) * (\y - \yvFourteen)) + 2*\gap*\ny}) {};
    \node[vertex] (mv18) at ({(\xvFourteen + (6/7) * (\x - \xvFourteen)) + 2*\gap*\nx}, 
                              {(\yvFourteen + (6/7) * (\y - \yvFourteen)) + 2*\gap*\ny}) {};

    \draw (v14) -- (v15);
    \draw[very thick, wavey] (v15) -- (v16); 
    \draw (v16) -- (v17);
    \draw[very thick, wavey] (v17) -- (v18); 
    \draw (v18) -- (v1);

    \draw[wavey] (mv1) to[bend left=30] (mv14);
    \draw (mv14) -- (mv15);
    \draw[very thick, wavey] (mv15) -- (mv16); 
    \draw (mv16) -- (mv17);
    \draw[very thick, wavey] (mv17) -- (mv18); 
    \draw (mv18) -- (mv1);

    \pgfmathsetmacro\xvTwo{\radius * cos(20)}
    \pgfmathsetmacro\yvTwo{\radius * sin(20)}
    \node[vertex] (v2) at (\xvTwo, \yvTwo) {};
    
    \pgfmathsetmacro\xvSeven{\radius * cos(120)}
    \pgfmathsetmacro\yvSeven{\radius * sin(120)}
    \node[vertex] (v7) at (\xvSeven, \yvSeven) {};
    
    \node[vertex] (v3) at ({\xvTwo + (1/7) * (\xvSeven - \xvTwo)}, {\yvTwo + (1/7) * (\yvSeven - \yvTwo)}) {};
    \node[vertex] (v4) at ({\xvTwo + (3/7) * (\xvSeven - \xvTwo)}, {\yvTwo + (3/7) * (\yvSeven - \yvTwo)}) {};
    \node[vertex] (v5) at ({\xvTwo + (4/7) * (\xvSeven - \xvTwo)}, {\yvTwo + (4/7) * (\yvSeven - \yvTwo)}) {};
    \node[vertex] (v6) at ({\xvTwo + (6/7) * (\xvSeven - \xvTwo)}, {\yvTwo + (6/7) * (\yvSeven - \yvTwo)}) {};

    \pgfmathsetmacro\dx{\xvSeven - \xvTwo}
    \pgfmathsetmacro\dy{\yvSeven - \yvTwo}
    \pgfmathsetmacro\len{sqrt(\dx*\dx + \dy*\dy)}
    \pgfmathsetmacro\nx{\dy / \len} 
    \pgfmathsetmacro\ny{-\dx / \len}

    \node[vertex] (mv2) at ({\xvTwo + 2*\gap*\nx}, {\yvTwo + 2*\gap*\ny}) {};
    \node[vertex] (mv7) at ({\xvSeven + 2*\gap*\nx}, {\yvSeven + 2*\gap*\ny}) {};
    \node[vertex] (mv3) at ({(\xvTwo + (1/7) * (\xvSeven - \xvTwo)) + 2*\gap*\nx}, 
                                {(\yvTwo + (1/7) * (\yvSeven - \yvTwo)) + 2*\gap*\ny}) {};
    \node[vertex] (mv4) at ({(\xvTwo + (3/7) * (\xvSeven - \xvTwo)) + 2*\gap*\nx}, 
                                {(\yvTwo + (3/7) * (\yvSeven - \yvTwo)) + 2*\gap*\ny}) {};
    \node[vertex] (mv5) at ({(\xvTwo + (4/7) * (\xvSeven - \xvTwo)) + 2*\gap*\nx}, 
                                {(\yvTwo + (4/7) * (\yvSeven - \yvTwo)) + 2*\gap*\ny}) {};
    \node[vertex] (mv6) at ({(\xvTwo + (6/7) * (\xvSeven - \xvTwo)) + 2*\gap*\nx}, 
                                {(\yvTwo + (6/7) * (\yvSeven - \yvTwo)) + 2*\gap*\ny}) {};

    \draw[wavey] (mv2) to[bend right=30] (mv7);
    \draw (mv2) -- (mv3); 
    \draw[wavey] (mv3) -- (mv4); 
    \draw (mv4) -- (mv5);
    \draw[wavey] (mv5) -- (mv6); 
    \draw (mv6) -- (mv7);

    \pgfmathsetmacro\xvEight{\radius * cos(140)}
    \pgfmathsetmacro\yvEight{\radius * sin(140)}
    \node[vertex] (v8) at (\xvEight, \yvEight) {};

    \pgfmathsetmacro\xvThirteen{\radius * cos(240)}
    \pgfmathsetmacro\yvThirteen{\radius * sin(240)}
    \node[vertex] (v13) at (\xvThirteen, \yvThirteen) {};

    \node[vertex] (v9) at ({\xvEight + (1/7) * (\xvThirteen - \xvEight)}, {\yvEight + (1/7) * (\yvThirteen - \yvEight)}) {};
    \node[vertex] (v10) at ({\xvEight + (3/7) * (\xvThirteen - \xvEight)}, {\yvEight + (3/7) * (\yvThirteen - \yvEight)}) {};
    \node[vertex] (v11) at ({\xvEight + (4/7) * (\xvThirteen - \xvEight)}, {\yvEight + (4/7) * (\yvThirteen - \yvEight)}) {};
    \node[vertex] (v12) at ({\xvEight + (6/7) * (\xvThirteen - \xvEight)}, {\yvEight + (6/7) * (\yvThirteen - \yvEight)}) {};

    \pgfmathsetmacro\dx{\xvThirteen - \xvEight}
    \pgfmathsetmacro\dy{\yvThirteen - \yvEight}
    \pgfmathsetmacro\len{sqrt(\dx*\dx + \dy*\dy)}
    \pgfmathsetmacro\nx{\dy / \len} 
    \pgfmathsetmacro\ny{-\dx / \len}

    \node[vertex] (mv8) at ({\xvEight + 2*\gap*\nx}, {\yvEight + 2*\gap*\ny}) {};
    \node[vertex] (mv13) at ({\xvThirteen + 2*\gap*\nx}, {\yvThirteen + 2*\gap*\ny}) {};
    \node[vertex] (mv9) at ({(\xvEight + (1/7) * (\xvThirteen - \xvEight)) + 2*\gap*\nx}, 
                             {(\yvEight + (1/7) * (\yvThirteen - \yvEight)) + 2*\gap*\ny}) {};
    \node[vertex] (mv10) at ({(\xvEight + (3/7) * (\xvThirteen - \xvEight)) + 2*\gap*\nx}, 
                              {(\yvEight + (3/7) * (\yvThirteen - \yvEight)) + 2*\gap*\ny}) {};
    \node[vertex] (mv11) at ({(\xvEight + (4/7) * (\xvThirteen - \xvEight)) + 2*\gap*\nx}, 
                              {(\yvEight + (4/7) * (\yvThirteen - \yvEight)) + 2*\gap*\ny}) {};
    \node[vertex] (mv12) at ({(\xvEight + (6/7) * (\xvThirteen - \xvEight)) + 2*\gap*\nx}, 
                              {(\yvEight + (6/7) * (\yvThirteen - \yvEight)) + 2*\gap*\ny}) {};

    \draw[wavey] (mv8) to[bend right=30] (mv13);
    \draw (mv8) -- (mv9);
    \draw[wavey] (mv9) -- (mv10); 
    \draw (mv10) -- (mv11);
    \draw[wavey] (mv11) -- (mv12); 
    \draw (mv12) -- (mv13);

    \draw (v1) -- (v2);
    \draw (v2) -- (v3); 
    \draw[wavey] (v3) -- (v4); 
    \draw (v4) -- (v5);
    \draw[wavey] (v5) -- (v6); 
    \draw (v6) -- (v7);
    \draw (v7) -- (v8); 
    \draw (v8) -- (v9);
    \draw[wavey] (v9) -- (v10); 
    \draw (v10) -- (v11);
    \draw[wavey] (v11) -- (v12); 
    \draw (v12) -- (v13);
    \draw (v13) -- (v14); 
    \draw[wavey] (v1) to[bend right=30] (v14);
    \draw[wavey] (v2) to[bend left=30] (v7);
    \draw[wavey] (v8) to[bend left=30] (v13);

    \draw (mv3) -- (v4);
    \draw (mv4) -- (v3);
    \draw (mv5) -- (v6);
    \draw (mv6) -- (v5);

    \draw (mv9) -- (v10);
    \draw (mv10) -- (v9);
    \draw (mv11) -- (v12);
    \draw (mv12) -- (v11);

    \draw[very thick] (mv15) -- (v16);
    \draw[very thick] (mv16) -- (v15);
    \draw[very thick] (mv17) -- (v18);
    \draw[very thick] (mv18) -- (v17);
\end{tikzpicture}
        \caption{The two $4$-cycles for an edge $ux \in E(G)$.}
    \end{subfigure}

    \caption{An illustration of the four cycle types. The vertex labels are the same 
	as in~\cref{fig:vc-covercycle-transformation} and are omitted for clarity.}
    \label{fig:vc-covercycle-cycle-types}
\end{figure}
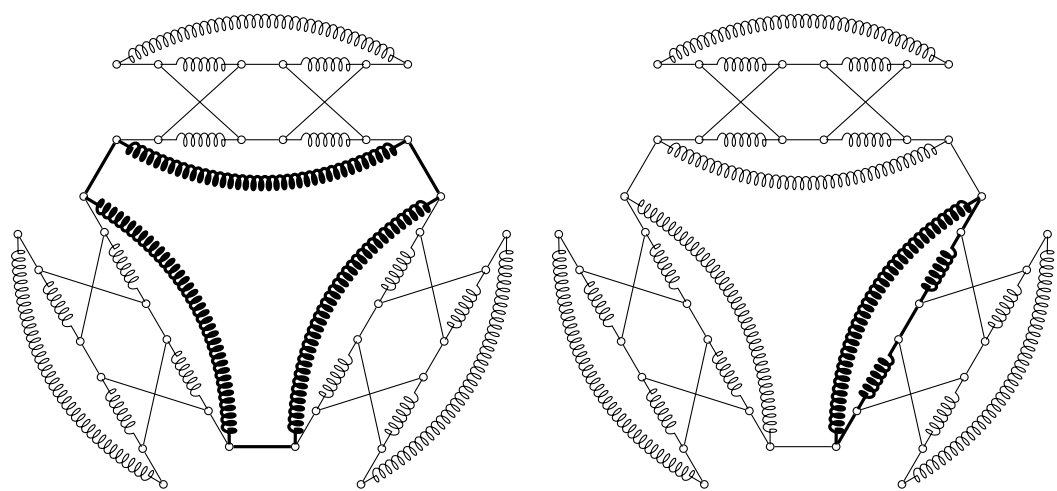
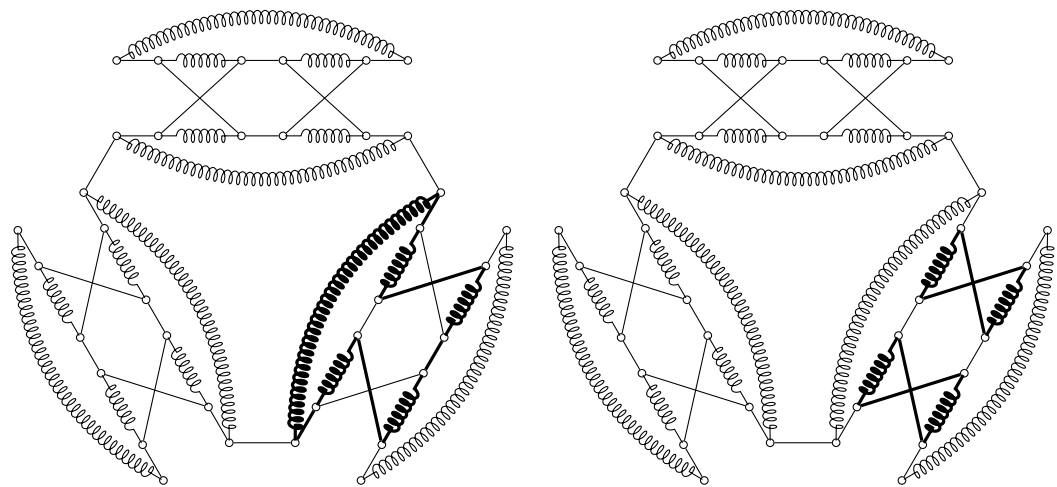
We are now ready to give the easy direction of the statement of the lemma. 

\begin{claim}\label{direasy}
Let $Z$ be a vertex cover of $G$. Then $H$ admits a cycle subpartition of weight $\alpha(|V(G)|+2|Z|)$.
\end{claim}
\begin{claimproof}
We describe a cycle subpartition $\mathcal{C}$ of $H$. First, for every $u \in V(G)-Z$, we let $\mathcal{C}$ contain a 1-cycle for $u$. Next consider some $uv \in E(G)$. As $Z$ is a vertex cover and by symmetry, we may suppose that $u \in Z$. If $v \in Z$, we let $\mathcal{C}$ contain a 2-cycle for $(u,v)$ and a 2-cycle for $(v,u)$. If $v \in V(G)-Z$, we let $\mathcal{C}$ contain a 3-cycle for $(u,v)$. This finishes the description of $\mathcal{C}$. By construction, we have $V(\mathcal{C})=V(H)$ and $|\mathcal{C}|=|V(G)-Z|+3|Z|=|V(G)|+2|Z|$. This yields $w(\mathcal{C})=\alpha|\mathcal{C}|+\beta(|V(H)-V(\mathcal{C})|)=\alpha |\mathcal{C}|=\alpha(|V(G)|+2|Z|)$. 
\end{claimproof}

We now give the other direction, which requires a much more involved proof. The difficulty is, roughly speaking, to obtain, starting with an arbitrary cycle subpartition, a cycle subpartition whose shape is similar to the shape of the cycle subpartitions constructed in \cref{direasy}. 

Let $\mathcal{C}$ be a cycle subpartition of $H$ whose weight is at most $\alpha(|V(G)|+2K)$ for some positive integer $K$. We need to show that $G$ admits a vertex cover of size at most $K$. We may suppose that $\mathcal{C}$ is of minimum weight among all cycle subpartitions of $H$.
In the following, if we say that a cycle $C \in \mathcal{C}$ contains an edge of
in $E(H)$ (a set of edges $E(H')$ of a subgraph $H'$ of $H$), we mean that
$E(C)$ contains this edge (this set of edges, $E(H')$).  The following first
restriction on the structure of $\mathcal{C}$ is simple, but very useful.

\begin{claim}\label{thorough}
Let $C \in \mathcal{C}$  and $H'$ a chain gadget in $H$ connecting two vertices $x,y$ such that $C$ contains an $xy$-path fully contained $H'$. Then $C$ contains a Hamiltonian path of $H'$.  
\end{claim}  
\begin{claimproof}
Suppose otherwise and let $C'$ be the cycle which is obtained from $C$ by replacing the $xy$-path fully contained in $H'$ by a Hamiltonian $xy$-path of $H'$. By property $(iv)$ of chain gadgets, such a path exists. Further, let $\mathcal{C'}$ be obtained from $\mathcal{C}$ by deleting and all cycles of $\mathcal{C}$ fully contained in $H'$ and adding $C'$. It is easy to see that $\mathcal{C'}$ is a cycle subpartition of $H$ and that $|\mathcal{C'}|\leq |\mathcal{C}|$ and $V(\mathcal{C})\subseteq V(\mathcal{C}')$ hold. Further, by assumption, one of $|\mathcal{C'}|<|\mathcal{C}|$ and $V(\mathcal{C})\subsetneq V(\mathcal{C}')$ holds. This yields $w(\mathcal{C}')-w(\mathcal{C})=\alpha(|\mathcal{C'}|-|\mathcal{C}|)+\beta(|V(H)-V(\mathcal{C}')|-|V(H)-V(\mathcal{C}')|)=\alpha(|\mathcal{C'}|-|\mathcal{C}|)+\beta(|V(\mathcal{C})|-|V(\mathcal{C}')|)<0$, a contradiction to the choice of $\mathcal{C}$.
\end{claimproof}

In order to approach the crucial part of our proof, we need some definitions. First, we say that a chain gadget connecting two vertices $x,y \in V(H)$ is {\it neglected} if there does not exist a cycle in $\mathcal{C}$ that contains $\overline{xy}$. Next, we say that some $u \in V(G)$ is {\it erroneous} if there is a neglected chain gadget connecting two vertices in $X'_u$ and we define $V_0$ to be the set of erroneous vertices. Finally, we say that a cycle in $\mathcal{C}$ is {\it affected} if it is not an $i$-cycle for some $i \in [4]$. 

We now prove through two claims that the existence of affected cycles implies the existence of erroneous vertices and hence contributes to $\mathcal{C}$ being of large weight.

\begin{claim}\label{rdftgzhu}
Let $u \in V(G), v \in N_G(u)$ and $C \in \mathcal{C}$ an affected cycle that contains at least one of $u_1^vu_6^v$ and $u_2^vu_3^v$. Then $u$ is erroneous.
\end{claim}   
\begin{claimproof}
Suppose otherwise and suppose by symmetry that $C$ contains $u_1^vu_6^v$. Then, as $\mathcal{C}$ is a cycle cover and by assumption, we obtain that $C$ contains $\overline{u_1^vu_2^v}$ and $\overline{u_6^vu_5^v}$. As $C$ is a cycle, we obtain that $C$ contains exactly one of $u_5^vu_4^v$ and $u_5^vv_6^{u}$. We distinguish these two cases.

\setcounter{Case}{0}
\begin{Case}
$C$ contains $u_5^vu_4^v$.
\end{Case}
By assumption, we obtain that $C$ contains $\overline{u_4^vu_3^v}$. As $C$ is not a 1-cycle for $u$, we obtain that $C$ does not contain $u_2^vu_3^v$. Hence, as $C$ is a cycle, we have that $C$ contains $u_3^vv_4^{u}$. By assumption, we obtain that $C$ contains $\overline{v_4^uv_3^u}$. Next, as $u$ is not erroneous, there must be a cycle $C'$ that contains $\overline{v_5^{u}v_6^{u}}$. As $H[\{v_5^{u},v_6^{u}\}]$ is acyclic, we have that $C'$ contains at least two edges linking $\{v_5^{u},v_6^{u}\}$ and $V(H)-\{v_5^{u},v_6^{u}\}$. As all vertices in $N_H(\{v_5^{u},v_6^{u}\})-v_1^{u}$ have two neighbors in $C$ not contained in $\{v_5^{u},v_6^{u}\}$, all edges linking $\{v_5^{u},v_6^{u}\}$ and $V(H)-\{v_5^{u},v_6^{u}\}$ in $C'$ must be incident to $v_1^{u}$. As there is only one edge in $H$ linking $\{v_5^{u},v_6^{u}\}$ and $v_1^{u}$, we obtain a contradiction to $C'$ being a cycle.

\begin{Case}
$C$ contains $u_5^vv_6^{u}$.
\end{Case}   
By assumption, we obtain that $C$ contains $\overline{v_6^{u}v_5^{u}}$. As $C$ is a cycle that contains $u_1^vu_6^v$ and $\overline{u_6^vu_5^v}$, we obtain that $C$ does not contain $v_5^{u}u_6^v$. It follows that $C$ contains $v_5^{u}v_4^{u}$ and hence, by assumption $\overline{v_4^{u}v_3^{u}}$. As $C$ is a cycle, we obtain that $C$ contains exactly one of $v_3^{u}v_2^{u}$ and $v_3^{u}u_4^v$.

First suppose that $C$ contains $v_3^{u}v_2^{u}$. As $u$ is not erroneous, there must be a cycle $C'$ that contains $\overline{u_3^{v}u_4^{v}}$. As $H[\{u_3^{v},u_4^{v}\}]$ is acyclic, we have that $C'$ contains at least two edges linking $\{u_3^{v},u_4^{v}\}$ and $V(H)-\{u_3^{v},u_4^{v}\}$. As all vertices in $N_H(\{u_3^{v},u_4^{v}\})-u_2^{v}$ have two neighbors in $C$ not contained in $\{u_3^{v},u_4^{v}\}$, all edges linking $\{u_3^{v},u_4^{v}\}$ and $V(H)-\{u_3^{v},u_4^{v}\}$ in $C'$ must be incident to $u_2^{v}$. As there is only one edge in $H$ linking $\{u_3^{v},u_4^{v}\}$ and $u_2^{v}$, we obtain a contradiction to $C'$ being a cycle.

Now suppose that $C$ contains $v_3^{u}u_4^v$. By assumption, we obtain that $C$ contains $\overline{u_4^vu_3^v}$. As $C$ is a cycle that contains $v_4^{u}v_5^u$ and $\overline{v_4^uv_5^u}$, we obtain that $C$ does not contain $u_3^{v}v_4^u$. As $C$ is a cycle, we obtain that $C$ contains $u_2^vu_3^v$. Hence $C$ is a 3-cycle for $(u,v)$, a contradiction to $C$ being affected.
\end{claimproof}

\begin{claim}\label{afferr}
Let $u \in V(G)$ and $C \in \mathcal{C}$ an affected cycle that contains at least one vertex of $X_u$. Then $u$ is erroneous.
\end{claim}   
\begin{claimproof}
Suppose otherwise and let $N_G(u)=\{x,y,z\}$ such that $u_2^xu_1^y\in E(H)$. By symmetry, we my suppose that $C$ contains $u_1^x$. By assumption and as $\mathcal{C}$ is a cycle subpartition, we have that $C$ contains $\overline{u_1^xu_2^x}$.
If $C$ contains one of  $u_1^xu_6^x$ and $u_2^xu_3^x$, we obtain a contradiction by \cref{rdftgzhu}.
As $C$ is a cycle, we obtain that $C$ contains $u_1^xu_2^z$ and $u_2^xu_1^y$. By assumption, we obtain that $C$ contains $\overline{u_1^yu_2^y}$ and $\overline{u_1^zu_2^z}$. If $C$ contains $u_2^yu_3^y$, we obtain a contradiction by \cref{rdftgzhu}. Otherwise, we have that $C$ contains $u_2^yu_1^z$, so $C$ is a 1-cycle for $u$, a contradiction to $C$ being affected. 
\end{claimproof}

The following result precisely describes the interaction of $\mathcal{C}$ and gadgets corresponding to vertices which are not erroneous.

\begin{claim}\label{ecyc}
Let $u \in V(G)-V_0$. Then either $\mathcal{C}$ contains a 1-cycle for $u$ or $\mathcal{C}$ contains a 2-cycle or a 3-cycle for $(u,v)$ for every $v \in N_G(u)$. 
\end{claim}
\begin{claimproof}
Let $N_G(u)=\{x,y,z\}$ such that $u_2^xu_1^y \in E(H)$. As $u \in V(G)-V_0$, we obtain that $\overline{u_1^xu_2^x}$ is contained in a cycle $C$ of $\mathcal{C}$. By \cref{afferr}, we obtain that $C$ is not affected. If $u_2^xu_1^y \in E(C)$, this implies that $C$ is a 1-cycle for $u$. Otherwise, this implies that $C$ is a 2-cycle or a 3-cycle for $(u,x)$. Similar arguments show that $\mathcal{C}$ either contains a 1-cycle for $u$ or 2-cycles or 3-cycles for $(u,y)$ and $(u,z)$. As $\mathcal{C}$ is a cycle subpartition, the statement follows.
\end{claimproof}

In the following, let $E_0 \subseteq E(G)$ be the set of edges $uv$ with $u,v \in V(G)-V_0$ such that $\mathcal{C}$ contains a 1-cycle for $u$ and a 1-cycle for $v$ and $\mathcal{C}$ contains at most one 4-cycle for $uv$.

We show that edges in $E_0$ have a similar effect as erroneous vertices.

\begin{claim}\label{tfgzhi}
Let $e=uv \in E_0$. Then there is a neglected chain gadget connecting two vertices in $X_e$.
\end{claim}
\begin{claimproof}
Suppose otherwise, so, by symmetry and assumption, we may suppose that $\mathcal{C}$ contains a cycle $C$ that contains $\overline{u_3^vu_4^v}$ and is not a 4-cycle. As $\mathcal{C}$ is a cycle subpartition of $H$ and contains a 1-cycle for $u$, we obtain that $C$ does not contain $u_2^vu_3^v$. Hence $C$ contains $u_3^vv_4^{u}$ and, by  assumption, $\overline{v_4^{u}v_3^{u}}$. Next, as $\mathcal{C}$ is a cycle subpartition of $H$ and contains a 1-cycle for $v$, we obtain that $C$ does not contain $v_2^uv_3^u$. Further, as $C$ is not a 4-cycle, we obtain that $C$ does not contain $v_3^{u}u_4^v$. This contradicts $C$ being a cycle.
\end{claimproof}

Finally, in the following claim, we can show that there are actually no erroneous vertices and edges in $E_0$ and hence the structure which was previously established for vertices which are not erroneous holds for all vertices in $V(G)$. The following claim in particular implies that $\mathcal{C}$ does not contain affected cycles. 
\begin{claim}\label{leer}
$V_0 = E_0=\emptyset$.
\end{claim}
\begin{claimproof}
We call a cycle in $\mathcal{C}$ {\it indirectly affected} if it is not affected, but it contains a vertex of $X_u$ for some $u \in V_0$. We describe in the following a different cycle subpartition $\mathcal{C}'$ of $H$, obtaining a contradiction to the minimality of $\mathcal{C}$. We first delete all affected and indirectly affected cycles from $\mathcal{C}$. Observe that every cycle that is not affected contains vertices of $X_u$ for at most one $u \in V(G)$. Hence no cycle of $\mathcal{C}$ that intersects $X_u$ for some $u \in V(G)-V_0$ is deleted. Now consider some $uv \in E(G)$ with $u,v \in V_0$. We add a 1-cycle for $(u,v)$ and a 1-cycle for $(v,u)$. Next consider some $u \in V_0$ and some $v \in N_G(u)-V_0$. By \cref{ecyc}, we have that $\mathcal{C}$ contains a 1-cycle for $v$, a 2-cycle for $(v,u)$, or a 3-cycle for $(v,u)$. If $\mathcal{C}$ contains a 1-cycle for $v$, we add a 3-cycle for $(u,v)$. If $\mathcal{C}$ contains a 2-cycle for $(v,u)$, we add a 2-cycle for $(u,v)$. Further, if $\mathcal{C}$ contains a 3-cycle for $(v,u)$, we replace this cycle by a 2-cycle for $(v,u)$ and we add a 2-cycle for $(u,v)$. Finally, for every $e \in E_0$, we add 2 disjoint 4-cycles for $e$, deleting previously existing ones. We denote this newly constructed collection of cycles by $\mathcal{C}'$. It is easy to see that $\mathcal{C}'$ is a cycle subpartition of $H$.

In the following, we will show that $w(\mathcal{C}')<w(\mathcal{C})$ unless $V_0$ and $E_0$ are empty. To this end, 
let $V_1\subseteq V(H)$ be the set of vertices which are contained in the interior of a neglected chain gadget and 
let $V_2=V(H)-V_1$. By property $(i)$ of chain gadgets, \cref{thorough} and construction, we have that $V(C)$ is 
fully contained in one of $V_1$ and $V_2$ for every $C \in \mathcal{C}$. For $i \in \{1,2\}$, let $\mathcal{C}_i$ be 
the cycle subpartition of $H[V_i]$ inherited from $\mathcal{C}$.

 By the definition of $V_0$, for every $u \in V_0$, 
there is a neglected chain gadget connecting two vertices in $X_u'$.
Next, by \cref{tfgzhi}, for every $e \in E_0$, there is a neglected chain gadget connecting two vertices in 
$X_e$. It follows that the total number of neglected chain gadgets is at least $\frac{1}{2}|V_0|+|E_0|$. It hence 
follows by property $(v)$ of chain gadgets that $w(\mathcal{C}_1)\geq \min\{\alpha,6 
\beta\}k(\frac{1}{2}|V_0|+|E_0|)$. Next observe that by construction and \cref{ecyc}, we have that 
$V(\mathcal{C}')=V(H)$, so in particular $|V(H)-V(\mathcal{C}')|\leq |V(H)-V(\mathcal{C})|$. Finally observe that, 
when constructing $\mathcal{C}'$, we added at most two cycles for every $u \in V_0$ and $v \in N_G(u)$ and we added 
at most two cycles for every $e \in E_0$. This yields $|\mathcal{C}'|-|\mathcal{C}|\leq 6|V_0|+2|E_0|$. This yields
\begin{align*}
 w(\mathcal{C}')-w(\mathcal{C})&=(w(\mathcal{C}')-w(\mathcal{C}_1))-w(\mathcal{C}_2)\\
 &\leq \alpha  (|\mathcal{C}'|-|\mathcal{C}|)-\min\{\alpha,6 \beta\}k(\frac{1}{2}|V_0|+|E_0|)\\
 &\leq \alpha (6|V_0|+2|E_0|)-\min\{\alpha,6 \beta\}k(\frac{1}{2}|V_0|+|E_0|)\\
 &= (6 \alpha -\frac{1}{2}\min\{\alpha,6 \beta\}k)|V_0|+ (2\alpha -\min\{\alpha,6 \beta\}k)|E_0|.
\end{align*}

By the choice of $k$, we obtain that $6 \alpha-\frac{1}{2}\min\{\alpha,6 \beta\}k<6 \alpha-\frac{1}{2}\min\{\alpha,6 \beta\}\max\{12,\frac{2 \alpha}{\beta}\}\leq 0$ and $2 \alpha-\min\{\alpha,6 \beta\}k<0$. Hence, if one of $V_0$ and $E_0$ is nonempty, we obtain a contradiction to the choice of $\mathcal{C}$. Hence the statement follows.
\end{claimproof}

We are now ready to construct a vertex cover of $G$. First we associate cycles in $\mathcal{C}$ to vertices in $V(G)$. For every vertex $u \in V(G)$ such that $\mathcal{C}$ contains a 1-cycle of $u$, we associate this 1-cycle to $u$. For every vertex $u \in V(G)$ and every $v \in N_G(u)$ such that $\mathcal{C}$ contains a 2-cycle or 3-cycle of $(u,v)$ , we associate this cycle to $u$. Finally consider an edge $e=uv \in E(G)$ such that $X_e$ contains two disjoint 4-cycles. Then we associate these two cycles to an arbitrary vertex in $\{u,v\}$. Let $Z \subseteq V(G)$ be the set of vertices which are associated to at least 3 cycles. By Claims \ref{ecyc} and \ref{leer} and construction, every vertex is associated to at least one cycle. Further, we have by construction that every cycle is associated to at most one vertex.
This implies $w(\mathcal{C})\geq \alpha |\mathcal{C}|\geq \alpha(|V(G)|+2|Z|)$. As $w(\mathcal{C})\leq \alpha(|V(G)|+2K)$, we obtain $|Z|\leq K$. It remains to prove that $Z$ is a vertex cover of $G$. To this end, let $uv \in E(G)$. If $\mathcal{C}$ contains a 2-cycle or a 3-cycle for one of $(u,v)$ or $(v,u)$, then, by construction, we have that $Z$ contains one of $u$ and $v$. Otherwise, by Claims \ref{ecyc} and \ref{leer}, we obtain that $\mathcal{C}$ contains 1-cycles for $u$ and $v$. By  \cref{leer}, we have that $uv$ is not contained in $E_0$ and hence $\mathcal{C}$ contains 2 4-cycles for $uv$. It follows by construction that $Z$ contains one of $u$ and $v$.   
\end{proof}

For the main proof of \cref{cyclehard}, we use the following simple proposition. 
\begin{proposition}\label{restdrzftustr}
Let $G$ be a cubic graph and $Z$ a vertex cover of $G$. Then $|Z|\geq |V(G)|/2$. 
\end{proposition}
\begin{proof}
    Suppose that $|Z| < |V(G)|/2$. Since $\deg_G(v) = 3$ for all $v \in V(G)$, it follows that
    the number of edges incident to $Z$ is at most $3|Z| < 3|V(G)|/2 = |E(G)|$, a contradiction.
\end{proof}

We are now ready to give the main proof of \cref{cyclehard}.

\begin{proof}[Proof of~\cref{cyclehard}]
Let $\alpha,\beta,\gamma,\kappa$ with $\alpha,\beta,\gamma>0$ be fixed. Let $\eps$ be the constant from \cref{vchard} and let $\eps'=\eps\frac{\alpha}{(3 \alpha+\gamma+1)2(54 \max\{12,\lceil2 \alpha/\beta\rceil\}+18)}$.
Now suppose that there is a polynomial-time algorithm $A$ that, given a connected cubic graph $H$ and a positive constant $K$, returns `yes' if $H$ admits a cycle subpartition of weight at most $K$ and returns `no'  if $H$ does not admit a cycle subpartition of weight at most $(1+\eps')K$.

Let a graph $G$ and an integer $K_1$ be given. We will design an algorithm that returns `yes' if $G$ admits a vertex cover of size $K_1$ and `no' if $G$ does not admit a vertex cover of size at most $(1+\eps)K_1$. First, if $K_1<\frac{1}{2}|V(G)|$, we return `no' and if $K_1 \geq |V(G)|$, we return `yes.'
This is justified by \cref{restdrzftustr} and the fact that $V(G)$ is trivially a vertex cover of $G$. Next, if $|V(G)|<\max\{\kappa,\frac{|\kappa|}{\gamma}\}$, we compute a minimum vertex cover by a brute force approach and output an appropriate answer. We may hence in the following assume that $\frac{1}{2}|V(G)| \leq K_1 < |V(G)|$ and $|V(G)|\geq \max\{\kappa,\frac{|\kappa|}{\gamma}\}$. We now use \cref{mainconstr} and compute, in polynomial time, a cubic bipartite graph $H$ with $|V(H)|\leq (54 \max\{12,\lceil\frac{2 \alpha}{\beta}\rceil\}+18) |V(G)|$ such that for every positive integer $K$, we have that $H$ admits a cycle subpartition $\mathcal{C}$ with $\alpha|\mathcal{C}|+\beta(|V(H)-V(\mathcal{C})|)\leq \alpha (|V(G)| + 2K)$ if and only if $G$ admits a vertex cover of size at most $K$.

We now apply $A$ to $H$ and $K_2$ which is defined by $K_2=\alpha(|V(G)| + 2 K_1)+\gamma |V(H)|+\kappa$.
Observe that $K_2>\gamma |V(H)|+\kappa\geq \gamma |V(G)|+\kappa\geq \gamma \frac{|\kappa|}{\gamma}+\kappa=|\kappa|+\kappa \geq 0$. We output the output of  $A$.

Observe that our complete algorithm can be executed in polynomial time. We now analyze the correctness of this algorithm. First suppose that $G$ admits a vertex cover of size at most $K_1$. Then $H$ admits a cycle subpartition of weight $K_2$. By assumption, our algorithm outputs `yes.'

Now suppose that our algorithm outputs `yes.' Then, by assumption, we have that $H$ admits a cycle subpartition $\mathcal{C}$ of weight at most $(1+\eps')K_2$.

By the assumptions on $K_1$, the definition of $\eps'$, and $|V(H)|\geq |V(G)|$, we thus have
\begin{align*}
\alpha \eps K_1&\geq \alpha \eps \frac{1}{2}|V(G)|\\
&\geq \alpha \eps \frac{1}{2(54 \max\{12,\lceil 2\alpha/\beta\rceil\}+18)}|V(H)|\\
&\geq \eps' (3 \alpha+\gamma+1)|V(H)|\\
&\geq \eps'(\alpha(2 K_1+|V(G)|)+\gamma |V(H)|+\kappa) = \eps' K_2. 
\end{align*} 

This yields 
\begin{align*}
w(\mathcal{C})&\leq (1+\eps')K_2 = K_2 + \eps' K_2\\
&\leq \alpha(|V(G)| + 2K_1) + \gamma |V(H)| + \kappa +  \alpha \eps K_1\\
&\leq \alpha(|V(G) + 2(1+\eps) K_1) + \gamma |V(H)| + \kappa.
\end{align*}

It follows by the definition of the weight function that $\alpha |\mathcal{C}|+\beta(|V(H)-V(\mathcal{C})|)\leq  \alpha(|V(G)| + 2(1+\eps)K_1)$.
This yields that $G$ admits a vertex cover of size at most $(1+\eps)K_1$. We now obtain \p = \np from  \cref{vchard}.
\end{proof}
\subsubsection[\texorpdfstring{Hardness for $\delta$-Tours}{Hardness for delta-Tours}]{Hardness for \boldmath$\delta$-Tours}\label{cycle2}
We prove the following more formal restatement of \cref{cor:approx:lb:tsp}.

\begin{theorem}\label{apx1}
For every $\delta \in (0,1/2]$, there is an $\eps >0$ such that, unless $\p=\np$, there is no algorithm that, given a connected graph $G$ and a positive real $K$, returns `yes' if $G$ admits a \deltatour of length $K$ and `no' if $G$ does not admit a \deltatour of length $(1+\eps)K$.
\end{theorem}

We first give a well-known result of Euler, which we later use.
A multigraph~$G$ is called {\it Eulerian} if $\deg_G(v)$ is even for all $v \in V(G)$. 
An {\it Euler tour} of $G$ is an integral tour $v_0,\ldots, v_z$ in $G$
such that for each pair $u,v \in V(G)$, the number of indices 
$i \in [z]$ with $\{v_{i-1},v_i\}=\{u,v\}$ is exactly the number of edges in $E(G)$ linking $u$ and~$v$.
\begin{proposition}[Euler's Theorem]\label{euler}
A multigraph~$G$ admits an Euler tour if and only if it is connected and Eulerian. Moreover, if an Euler tour exists, it can be computed in polynomial time.
\end{proposition} 

\begin{lemma}\label{xtcfrgvt}
Let $G$ be a connected, cubic graph and $\delta,K$ constants with $0<\delta<\frac{1}{2}$ and $K>0$. Then $G$ admits a \deltatour of length at most $K$ if and only if $G$ admits a cycle subpartition $\mathcal{C}$ with $\alpha|\mathcal{C}|+\beta|V(G)-V(\mathcal{C})|+\gamma|V(G)|+\kappa \leq K$ where $\alpha=4 \delta,\beta=1,\gamma=2-2\delta$ and $\kappa=-4 \delta$.
\end{lemma}
\begin{proof}
Let first $\mathcal{C}$ be a cycle subpartition of $G$ that satisfies $\alpha|\mathcal{C}|+\beta|V(G)-V(\mathcal{C})|+\gamma|V(G)|+\kappa \leq K$. Let $G'$ be the graph obtained from $G$ by contracting every cycle of $\mathcal{C}$ into a single vertex. As $G$ is connected, so is $G'$. We can hence choose a spanning tree $T$ of $G'$. Let $F$ be the set of edges in $E(G)$ that correspond to the edges in $E(T)$. Observe that $|F|=|V(G')|-1=|\mathcal{C}|+|V(G)-V(\mathcal{C})|-1$.  Now let $H$ be the multigraph on $V(G)$ that contains $E(\mathcal{C})$ and that contains two copies of every edge of $F$. As $\mathcal{C}$ is a cycle subpartition of $G$, we obtain that the degree of every vertex of $H$ is even. Further, it follows by construction that $H$ is connected. Hence, by Proposition \ref{euler}, we obtain that $H$ admits an Euler tour $T_0$. We now obtain a tour $T$ in $G$ from $T_0$ in the following way: for every $uv \in E(G)$ that is not traversed by $T_0$, we choose an arbitrary occurrence of one of $u$ and $v$, say $u$, and replace this occurrence by the segment $up(u,v,1-2\delta)u$. It follows directly by construction and Lemma \ref{char2stops} that $T$ is a \deltatour of $G$. Further, as $\mathcal{C}$ is a cycle subpartition and $G$ is cubic, we have 
\begin{align*}
\len(T)&=2|F|+|E(\mathcal{C})|+(2-4 \delta)|E(G)-(F \cup E(\mathcal{C}))|\\
&=4 \delta|F|+(4 \delta-1)|E(\mathcal{C})|+(2-4 \delta)|E(G)|\\
&=4 \delta (|\mathcal{C}|+|V(G)-V(\mathcal{C})|-1)+(4 \delta-1)|V(\mathcal{C})|+(3-6 \delta)|V(G)|\\
&=\alpha|\mathcal{C}|+\beta|V(G)-V(\mathcal{C})|+\gamma|V(G)|+\kappa\\
&\leq K.
\end{align*}

Now suppose that $G$ a \deltatour $T$ with $\len(T)\leq K$. By Lemma \ref{discreteklein} and as $G$ is cubic, we may suppose that $T$ is nice and that one of $(c)$ and $(d)$ of \cref{discreteklein} holds  for every $uv \in E(G)$. Now let $F_1$ be the set of edges in $E(G)$ which are traversed an odd number of times by $T$, let $F_2$ be the set of edges of $E(G)$ that are traversed an even number of times, but at least twice,  by $T$ and let $F_3=E(G)-(F_1 \cup F_2)$. Observe that, as $T$ is a tour, every vertex of $G$ is incident to an even number of edges in $F_1$. Hence $F_1$ forms the edge set of a cycle subpartition $\mathcal{C}$ in $G$. Further, as $T$ is a nice tour stopping at every vertex of $G$ and $(d)$ of \cref{discreteklein} holds  for every $uv \in F_3$, we obtain $|F_2|\geq |\mathcal{C}|+|V(G)-V(\mathcal{C})|-1$. 

Further, as $\mathcal{C}$ is a cycle subpartition and $G$ is cubic, we have 
\begin{align*}
K&\geq \len(T)\\
&\geq 2|F_2|+|F_1|+(2-4 \delta)|F_3|\\
&= 2|F_2|+|F_1|+(2-4 \delta)|E(G)-(F_1 \cup F_2)|\\
&=4 \delta|F_2|+(4 \delta-1)|F_1|+(2-4 \delta)|E(G)|\\
&\geq 4 \delta (|\mathcal{C}|+|V(G)-V(\mathcal{C})|-1)+(4 \delta-1)|V(\mathcal{C})|+(3-6 \delta)|V(G)|\\
&=\alpha|\mathcal{C}|+\beta|V(G)-V(\mathcal{C})|+\gamma|V(G)|+\kappa.
\end{align*}

This finishes the proof.
\end{proof}

Observe that for any constant $\delta$ with $0<\delta<\frac{1}{2}$, we have that all of $4 \delta,1$, and $2-2\delta$ are positive. Hence  \cref{xtcfrgvt} and \cref{cyclehard} directly imply \cref{apx1}.
\subsubsection{Hardness for Cubic Bipartite TSP}\label{cycle3}
We now obtain the hardness result for graphic cubic bipartite TSP through the following simple observation.
\begin{lemma}\label{tsp12}
Let $G$ be a cubic bipartite connected graph and $K$ a constant. Then $G$ admits a TSP-tour of length at most $K$ if and only if $G$ admits a $\frac{1}{2}$-tour of length $K$. 
\end{lemma}
\begin{proof}
First let $T$ be a TSP-tour of $G$ with $\len(T)\leq K$. Clearly, we may suppose that $T$ is nice. It then follows directly from \cref{char2stops} that $T$ is a $\frac{1}{2}$-tour of $G$. 

Now suppose that $G$ admits a $\frac{1}{2}$-tour $T$ with $\len(T)\leq K$. By \cref{discreteklein} and as $G$ is cubic, we may suppose that one of $(c)$ and $(d)$ of \cref{discreteklein} holds for every $uv \in E(G)$. Hence $T$ is a TSP-tour of $G$.
\end{proof}

\Cref{apx1} and \cref{tsp12} immediately imply \cref{cor:approx:lb:tsp}.

\subsection[\texorpdfstring{APX-Hardness for $\delta\ge1/2$ via Subdivision}{APX-Hardness for delta > 1/2 via Subdivision}]{APX-Hardness for \boldmath$\delta\ge1/2$ via Subdivision}
In this section, we extend the hardness result obtained for small $\delta$ in the previous section to arbitrary values of $\delta$. More concretely, we prove \cref{cor:approx:lb:half:threehalves}. The key ingredient to transfer the hardness result to larger $\delta$ is the following simple subdivision result.
\begin{lemma}\label{drftgzh}

Let $G$ be a graph, $\delta>0$ a constant and $k$ a positive integer. Then, in polynomial time, we can create a graph~$G'$ such that for every $\delta$-tour~$T$ in $G$, there is a $k \delta$-tour~$T'$ with $\len(T')=k \len(T)$ in $G'$, and for every $k\delta$-tour~$T'$ in $G'$, there is a $\delta$-tour~$T$ in $G$ with $\len(T)=\frac{1}{k} \len(T')$.

\end{lemma}
\begin{proof}
We obtain $G'$ from $G$ by subdividing each edge of $G$ $k-1$ times. Clearly, $G'$ can be constructed from $G$ in polynomial time. For some $uv \in E(G)$, we denote the newly created vertices by $x_{uv}^1,\ldots,x_{uv}^{k-1}$ in the order they appear when following the path in $G'$ corresponding to $uv$ from $u$ to $v$. Further, we use $x_{uv}^0$ for $u$ and $x_{uv}^k$ for $v$. By convention, we have $x_{uv}^{i}=x_{vu}^{k-i}$ for $i \in \{0,\ldots,k\}$. We next define a natural bijection $\phi\colon P(G)\to P(G')$; namely, we set $\phi(v)=v$ for all $v \in V(G)$ and for any $p=p(u,v,\lambda)\in P(G)$ with $\lambda \in (0,1)$, we set $\phi(p)=p(x_{uv}^{\floor{k\lambda}},x_{uv}^{\floor{k\lambda}+1},k\lambda-\floor{k\lambda})$. It is easy to see that $\phi$ is a bijection.

First let $T$ be a \deltatour in $G$. By \cref{lem:tournice}, we may suppose that $T$ is nice. In order to construct a tour $T'$ in $G'$, let $p$ and $q$ be consecutive stopping points of $T$. As $T$ is nice and by symmetry, we may suppose that there is some $uv \in E(G)$ with $p=u$ and $\phi(q)=(x_{uv}^{i-1},x_{uv}^{i},\lambda)$ for some $i \in [k]$ and $\lambda \in (0,1]$. We  replace $pq$ by $u,x_{uv}^1,x_{uv}^2,\ldots,x_{uv}^{i-1},\phi(q)$. We obtain $T'$ by applying this replacement to every pair of consecutive stopping points of $T$. It is easy to see that $T'$ is a tour in $G$ with $\len(T')=k\len(T)$. In order to prove that $T'$ is a $k \delta$-tour, consider some $q' \in P(G')$. As $T$ is a \deltatour in $G$, there is some $p \in P(T)$ with $\dist_G(p,\phi^{-1}(q'))\leq \delta$. We obtain that $\dist_{G'}(\phi(p),q')\leq k \delta$. By construction, we have $\phi(p)\in P(T')$; it follows that $T'$ is a $k \delta$-tour in $G'$. 

Now let $T'$ be a $k \delta$-tour in $G'$. We obtain $T$ from $T'$ by replacing $p'$ by $\phi^{-1}(p')$ for all stopping points $p'$ of $T'$. It is easy to see that $T$ is a tour in $G$ with $\len(T)=\frac{1}{k} \len(T')$. Now consider some $q \in P(G)$. As $T'$ is a $k\delta$-tour in $G'$, there is some $p' \in P(T')$ such that $\dist_{G'}(\phi(q),p')\leq k\delta$. This yields $\dist_{G'}(q,\phi^{-1}(p')\leq \delta$. As $\phi^{-1}(p') \in P(T)$ by construction, it follows that $T$ is a \deltatour in $G$. 
\end{proof}

We are now ready to conclude \cref{cor:approx:lb:half:threehalves} from \cref{thm:approx:lb:zero:half} and \cref{drftgzh}. For convenience, \cref{cor:approx:lb:half:threehalves} is restated below.

\APXhardFromHalfToThreeHalves*\label\thisthm

\begin{proof}
Let $\delta>0$ and consider $\delta'=\frac{\delta}{\ceil{2 \delta}+1}$. Observe that $\delta'<\frac{1}{2}$, so by Theorem~\ref{thm:approx:lb:zero:half}, there is some $\eps>0$ such that there is no polynomial-time $(1+\eps)$-approximation algorithm for the problem of computing a shortest $\delta'$-tour unless $\p=\np$. 

Suppose that there is a polynomial-time $(1+\eps)$-approximation algorithm for \deltatourprob. Let $G'$ be a graph and $K$ be a constant. By Lemma~\ref{drftgzh}, in polynomial time, we can compute a graph~$G$ such that for every $\delta'$-tour~$T'$ in $G'$, there is a $\delta$-tour~$T$ with $\len(T)=(\ceil{2 \delta}+1)\len(T')$ and for every $\delta$-tour~$T$ in $G$, there is a $\delta'$-tour~$T'$ in $G'$ with $\len(T')=\frac{1}{\ceil{2 \delta}+1} \len(T)$. By assumption, there is a polynomial-time algorithm that outputs `yes' if $G$ admits a \deltatour of length at most $(\ceil{2 \delta}+1)K$ and `no' if $G$ does not admit a \deltatour of length at most $(1+\eps)(\ceil{2 \delta}+1)K$. Observe that this algorithm outputs `yes' if $G'$ admits a $\delta'$-tour of length at most $K$ and `no' if $G$ does not admit a \deltatour of length at most $(1+\eps)K$. Further, the total running time of this procedure is polynomial. It follows that $\p=\np$.
\end{proof}

\subsection[\texorpdfstring{Inapproximability for Covering Range $\delta\geq 3/2$}{Inapproximability for Covering Range delta >= 3/2}]{Inapproximability for Covering Range \boldmath$\delta\geq 3/2$}\label{sec:aprox:lb:log:n}
This section is dedicated to proving Theorem~\ref{thm:approx:lb:log:n}. Our main technical ingredient is a lemma mapping instances of the dominating set problem in so-called split graphs to instances of our tour problem. This lemma will be reused in Section~\ref{section:w1hard-wrt-opt}. 

We need two technical definitions. A graph~$G$ is called a {\it split graph} if there is a partition $(C,I)$ of $V(G)$ such that $G[C]$ is a clique and $I$ is an independent set.
We say that a split graph is {\it trivial} if it contains a dominating set of size at most 2, {\it nontrivial} otherwise. We first need the following simple observation on dominating sets in split graphs.
\begin{proposition}\label{propsplit}
Let $G$ be a connected split graph and let $(C,I)$ be a partition of $V(G)$ with $C \neq \emptyset$ such that $G[C]$ is a clique  and $I$ is an independent set.
Then for every dominating set $S$ of $G$, there is a dominating set $S'$ of $G$ with $|S'|\leq |S|$ and $S' \subseteq C$.
\end{proposition}
\begin{proof}
Let be $S$ a dominating set of $G$ minimizing $|S \cap I|$ among
all dominating sets of size at most $|S|$. The statement is trivial for $S\cap I=\emptyset$. 
Otherwise, there is some $v\in S\cap I$. Since $G$ is a connected graph and $I$ is an independent set, $v$ has a neighbor $u\in C$. 
The set $S'=u\cup S-v$ satisfies $|S'|\le|S|$ and is a dominating set since $C$ is a clique, contradicting minimality.
\end{proof}
We are now ready to give the main lemma for Theorem~\ref{thm:approx:lb:log:n}.
\begin{lemma}
\label{lem:dom_tour_reduction}
	For any fixed $\delta \geq 1.5$ and every nontrivial, connected split graph~$G$,
	in polynomial time, we can compute a graph~$G'$ such that the following conditions hold:
	\begin{enumerate}
		\item $|V(G')| \leq \delta |V(G)|$,
			\label[lemma]{prop:dom_tour_instance_size}
		\item for every dominating set $S$ in $G$,
			there is a $\delta$-tour~$T$ in $G'$ with
			$\len(T) \leq |S|$,
			\label[lemma]{prop:dom_to_tour}
		\item for every $\delta$-tour~$T$ in $G'$, 
			 there is a dominating set $S$ in $G$
			with 
			$|S| \leq
			\len(T)/s_\delta$.
			\label[lemma]{prop:tour_to_dom}
	\end{enumerate}
\end{lemma}
\begin{proof}
Let $(C,I)$ be a partition of $V(G)$ such that $C$ induces a clique and $I$ an independent set. It is easy to see that $(C,I)$ can be computed in polynomial time given $G$. We now create $G'$ from $G$ in the following way. If $\delta \geq 2$, for every $v \in I$, we add a new vertex $v'$ and we connect $v$ and $v'$ by a path of length $\floor{\delta}-1$. If $\delta<2$, we set $G'=G$ and $v'=v$ for all $v \in I$. This finishes the description of $G'$. Clearly, $G'$ can be constructed in polynomial time. By construction, we have $|V(G')|=|V(G)|+(\floor{\delta}-1)|I|\leq \delta |V(G)|$, so (1) holds.

To prove (2), let $S$ be a dominating set in $G$. As $G$ is connected and nontrivial, we have that $C \neq \emptyset$.
Therefore, by \cref{propsplit}, we  may assume that 
$S \subseteq C$. Let $\{s_1,\ldots,s_k\}$ be an arbitrary ordering of $S$ and let $T = (s_1, \dots, s_k, s_1)$. As $C$ is a clique, we have that $T$ is a tour in $G'$. Further, we clearly have $\len(T)=|S|$. In order to see that $T$ is a $\delta$-tour in $G'$, let $p=p(x,y,\lambda) \in P(G)$. By symmetry, we may suppose that $\lambda \leq \frac{1}{2}$. If $x \in C$, let $u \in S$. As $C$ is a clique, we have $\dist_{G'}(p,T)\leq \dist_{G'}(p,u)\leq \dist_{G'}(p,x)+\dist_{G'}(x,u)\leq \frac{1}{2}+1\leq \delta$. Otherwise, by construction, we have $\dist_{G'}(p,v)\leq \delta-1$ for some $v \in I$. As $S$ is a dominating set, there is some $u \in N_{G'}(v)\cap S$. This yields $\dist_{G'}(p,T)\leq \dist_{G'}(p,u)\leq \dist_{G'}(p,x)+\dist_{G'}(x,u)\leq (\delta-1)+1= \delta$. This proves (2).

In order to prove (3), let $T$ be a \deltatour in $G'$. By Lemma~\ref{lemma:discretization}, we may suppose that every stopping point of $T$ is of the form $p(x,y,\lambda)$ for some $xy \in E(G)$ and some $\lambda \in S_\delta$ and further, that either $\alpha(T)\leq 2$ or $T$ is nice. Further, let $S$ contain all the elements $u$ of $C$ such that $T$ stops at a point of the form $p(u,v,\lambda)$ for some $\lambda<1$. 

First suppose that $S=\emptyset$. Then, by construction, there is some $v \in I$ such that all stopping of $T$ are on the unique path from $v$ to $v'$ in $G'$. As $G$ is non-trivial, there is some $w \in I-v$. By construction, we have $\dist_{G'}(w',T)\geq \dist_{G'}(w',v)=\dist_{G'}(w',w)+\dist_{G'}(w,v)=\floor{\delta}-1+2>\delta$, a contradiction to $T$ being a $\delta$-tour. We may hence suppose that $S\neq \emptyset$.

We next show that $S$ is a dominating set of $G$. As $S \neq \emptyset$ and $C$ is a clique, we have that $S$ dominates $C$. Now consider some $v \in I$ and suppose for the sake of a contradiction that $S$ does not dominate $v$. As $T$ is a tour and by construction, we obtain that $T$ does not stop at any point of the unique path in $G'$ linking $v$ and $v'$. This yields that $\dist_G(v',T)=\dist_G(v,v')+\dist_G(v,T)=\floor{\delta}-1+2>\delta$, a contradiction to $T$ being a \deltatour.

As $G$ is non-trivial and $S$ is a dominating set, we obtain $|S|\geq 3$. It follows that $T$ has at least $|S|$ stopping points and hence, in particular, is nice. We obtain $\len(T)\geq s_\delta |S|$ by \cref{lentour}. Hence (3) holds.  
\end{proof}

The following well-known result can easily be obtained from a corresponding result for the set cover problem in \cite{DinurS14}.
\begin{proposition}\label{domsetapprox}
Unless $\p=\np$, there is an absolute constant $\alpha_0$ for which there is no polynomial-time algorithm that, given a graph~$G$ and a constant $K$, returns `yes' if $G$ admits a dominating set of size at most $K$ and `no' if $G$ does not admit a dominating set of size at most $\alpha_0 \log (|V(G)|)K$. 
\end{proposition}
Following a proof by Raman and Saurabh \cite[Thm.~2]{RamanS08} verbatim, we obtain the following result from \cref{domsetapprox}.

\begin{theorem}\label{rabsaur}
Unless $\p=\np$, there is an absolute constant $\alpha$ for which there is no algorithm that runs in polynomial time and, given a split graph~$G$ and a constant $K$, returns `yes' if $G$ admits a dominating set of size at most $K$ and `no' if $G$ does not admit a dominating set of size at most $\alpha \log (|V(G)|)K$. 
\end{theorem}

We are now ready to prove \cref{thm:approx:lb:log:n}, which we restate here for convenience.
\ThmApproxLbLogN*\label\thisthm

\begin{proof}
Let $\delta \geq \frac{3}{2}$, let $\alpha_\delta=s_\delta\alpha/2$ where $\alpha$ is the constant from Theorem~\ref{rabsaur}, and suppose that there is a polynomial-time algorithm $A$ that, given a graph~$G'$ and a constant $K'$, returns `yes' if $G'$ admits a \deltatour of length at most $K'$ and `no' if $G'$ does not admit a \deltatour of length at most $\alpha_\delta \log(|V(G')|)K'$. We will show that $\p=\np$ using Theorem~\ref{rabsaur}.  

Let $G$ be a split graph and $K$ a constant. Observe that every subgraph of a split graph is also a split graph. Hence, if $G$ is disconnected, we may apply our algorithm to every connected component of $G$. We may hence suppose that $G$ is connected.
Next, if $|V(G)|\leq \delta$, we can solve the problem by a brute-force approach by \cref{decide}. We first check in $O(|V(G)|^2)$ time whether $G$ admits a dominating set of size at most 2. If this is the case, we may output an appropriate answer.  We may therefore suppose that $G$ is nontrivial, connected, and satisfies $|V(G)|\geq \delta$. 

It then follows by Lemma~\ref{lem:dom_tour_reduction} that we can compute a graph~$G'$ satisfying (1), (2), and (3) in polynomial time. We then apply $A$ to $G'$ and $K$ and output the output of $A$. Observe that the total running time of the algorithm is polynomial. Further, if $G$ admits a dominating set of size at most $K$, then, by $(1)$, we have that $G'$ admits a \deltatour of length at most $K$, so the algorithm returns `yes.' Now assume that the algorithm returns `yes,' so by assumption, we have that $G'$ admits a \deltatour of length at most $\alpha_\delta \log(|V(G')|)K$. By (3), there is a dominating set $S$ of $G$ that satisfies $|S|\leq \alpha_\delta \log(|V(G')|)K/s_\delta$. As $|V(G)|\geq \delta$, we have $\log(|V(G')|)\leq 2\log(|V(G)|)$, yielding $|S|\leq \alpha \log(|V(G)|)K$. By Theorem~\ref{rabsaur}, we obtain that $\p=\np$.
\end{proof}
\section{Parameterized Complexity}
\label{section:param:complexity}
In this section, we analyze the parameterized complexity of our problem \deltatourprob. 
On the one hand, we look at the problem parameterized by the length of a shortest tour. 
Here we give an FPT algorithm for all $\delta<3/2$ in~\cref{section:parametrized-by-opt}, we provide  a \wone-hardness result and an XP algorithm for fixed $\delta\ge3/2$ in~\cref{section:w1hard-wrt-opt}, and we prove an NP-hardness result for $\delta$ being part of the input in \cref{dfguhs}. 
On the other hand, we analyze the problem with the parameter $n/\delta$ 
where $n$ is the given graph's order.  
\Cref{section:by-n-over-delta}~presents an XP-algorithm and 
\cref{section:w1-hard-n-over-delta} rules out the existence of an 
FPT algorithm under the assumption $\fpt\neq\wone$.

\subsection[\texorpdfstring{FPT Algorithm Parameterized by Tour Length for $\delta<3/2$}{FPT Algorithm Parameterized by Tour Length for delta < 3/2}]{FPT Algorithm Parameterized by Tour Length for \boldmath$\delta<3/2$}
\label{section:parametrized-by-opt}
This section derives an FPT algorithm for \deltatourprob 
	parameterized by the length of a shortest tour $\opttour$. 
In fact, we give an algorithm which allows $\delta$ to be part of the input
	and that is fixed-parameter tractable for $\delta$ and tour length $K$ combined.
Formally, we prove \cref{fpt:sol}, which we restate here for convenience. 

\FPTSol*\label\thisthm

Our algorithm is based on the following kernelization:
Given a graph~$G$, reals $\delta \in (0,3/2)$ and $K\geq0$,
	we either correctly conclude that $G$ has no \deltatour of length at most $K$,
	or output an equivalent instance of size at most $f(K,\delta)$
	for a computable function $f$.
The key insight is a bound on the vertex cover size of $f(K,\delta)$
	for a computable function,
	assuming there is a \deltatour of length at most $K$; see \cref{vertexcover}.
Hence we may compute an approximation $C$ of a minimum vertex cover,
	and reject the instance if $C$ is too large.
We partition the vertices in $V(G) \setminus C$ by their neighborhood in the vertex cover $C$.
The main remaining insight, proven in~\cref{neglect}, is that we can remove from any set $S$ 
of the partition whose size is larger than $f(K,\delta)$ for a certain computable function $f$,
	one vertex and obtain an equivalent instance.

Let $G$ be a connected graph and $T$ a tour of $G$ that stops at at least 3 points.
We then say that $T$ \emph{neglects} a vertex $v \in V(G)$
	if $T$ does not stop at a point in distance less than $1$ to~$v$.

The next result shows that if a connected graph $G$ admits a short \deltatour, then it also admits a small vertex cover. 
\begin{lemma}\label{vertexcover}
Let $G$ be a connected graph and let $T$ be a nice shortest \deltatour in $G$ for some $\delta>0$ with $\delta< \frac{3}{2}$ where $T$ stops at at least 3 points. Then $G$ admits a vertex cover of size at most $\len(T)/s_\delta+1$.
\end{lemma}
\begin{proof}
Let $X$ be the set of vertices in $V(G)$ which are not neglected by $T$. As $T$ is nice and stops at at least 3 points, we may suppose that $T$ stops at some $v_0 \in V(G)$ and $\abs{X}\geq 3$.  For every $v \in X-v_0$, we now define $\lambda_v$ to be the largest constant such that $T$ stops at a point of the form $p(u,v,\lambda_v)$ for some $u \in N_G(v)$. As $T$ is nice, we obtain that $\len(T)\geq \sum_{v \in X-v_0}\min\{2\lambda_v,1\}\geq \sum_{v \in X-v_0}\min\{2s_\delta,1\}\geq (\abs{X}-1)s_\delta$. It follows that $\abs{X}\leq \len(T)/s_\delta+1$. In order to prove that $X$ is a vertex cover, let $e=uv \in E(G)$ and let $p=(u,v,\frac{1}{2})$. As $T$ is a \deltatour in $G$ and by Proposition~\ref{nearstop}, there is some $q\in P(G)$ stopped at by $T$ such that $\dist_G(p,q)\leq \delta$. As $\delta<\frac{3}{2}$ and by symmetry, we may suppose that $q=(w,u,\lambda)$ for some $w \in N_G(u)$ and some $\lambda>0$. We obtain that $u \in X$. It follows that $X$ is a vertex cover of $G$. 
\end{proof}
The next result shows that if the graph is very large in comparison to the tour length, then the tour neglects some vertices of the graph.
\begin{proposition}\label{neglect}
Let $G$ be a connected graph, $\delta >0$, $T$ a nice shortest \deltatour in $G$ stopping at at least 3 points and $Y$ a set of at least $\len(T)/s_\delta+2$ vertices in $V(G)$. Then at least one vertex in $Y$ is neglected by $T$. 
\end{proposition}
\begin{proof}

By assumption, there is some $v_0 \in V(G)$ at which $T$ stops. Now suppose that no vertex of $Y$ is neglected by $T$. Then we have $\len(T)\geq  \abs{Y-v_0}\min\{2s_\delta,1\}\geq (\abs{Y}-1)s_\delta$. This yields $\abs{Y}\leq \len(T)/s_\delta+1$, a contradiction.
\end{proof}

Given a graph~$G$, two vertices $u,v \in V(G)$ are \emph{false twins} in $G$ if $N_G(u)=N_G(v)$ and $E(G)$ does not contain an edge linking $u$ and $v$.
A \emph{set of false of twins} in $G$ is a set of vertices which are pairwise false twins in~$G$.
\begin{lemma}\label{eliminate}
Let $G$ be a connected graph, $\delta>0$ and $K>0$. Moreover, let $Y\subseteq V(G)$ be a set of false twins of size at least $K/s_\delta+3$ and $y \in Y$. Then $G$ admits a \deltatour of length at most $K$ if and only if $G-y$ admits a \deltatour of length at most $K$. Moreover, every nice \deltatour of length at most $K$ in $G-y$ is  a \deltatour of length at most $K$ in $G$.
\end{lemma}
\begin{proof}
First suppose that $G$ admits a \deltatour~$T$ of length at most $K$. By \cref{lem:tournice}, without loss of generality, we may suppose that $T$ is nice. It follows from Proposition~\ref{neglect} that $G$ neglects some vertex in $Y$. Possibly relabeling some elements of $Y$, we may suppose that $T$ neglects $y$. We now show that $T$ is also a \deltatour in $G-y$. First observe that, as $T$ neglects $y$, we clearly have that $T$ is a tour in $G-y$. Now consider some $p \in P(G-y)$. As $T$ is a \deltatour in $G$, we obtain that there is some point $q \in P(G)$ passed by $T$ such that $\dist_G(p,q)\leq \delta$. As $T$ neglects $y$, we obtain $q \in P(G-y)$. Now let $W=pp_1\ldots p_zq$ be a shortest $\seq{p&q}$-walk in $G$. By assumption, we have $\len(W)\leq \delta$. Further, we may suppose that $p_i \in V(G)$ for $i \in [z]$. If $p_i \neq y$ for all $i \in [z]$, it holds $\dist_{G-y}(p,T)\leq \dist_{G-y}(p,q)\leq \len(W)\leq \delta$. We may hence suppose that $p_i=y$ for some $i \in [z]$. Let $y' \in Y-y$ and let $W'=pp_1\ldots p_{i-1}y'p_{i+1}\ldots p_zq$.

 As $y$ and $y'$ are false twins, we obtain that $W'$ is a walk in $G-y$. This yields $\dist_{G-y}(p,T)\leq \dist_{G-y}(p,q)\leq \len(W')=\len(W)\leq \delta$. Hence $T$ is a \deltatour in $G-y$.

Now suppose that $G-y$ admits a \deltatour of length at most $K$ and let $T$ be a nice shortest \deltatour in $G-y$.  Clearly, we have $\len(T)\leq K$ and $T$ is also a tour in $G$. Now consider some $p \in P(G)$. If $p \in P(G-y)$, as $T$ is a \deltatour in $G-y$, we have $\dist_G(p,T)\leq \dist_{G-y}(p,T)\leq \delta$. We may hence suppose that $p=p(x,y,\lambda)$ for some $x \in N_G(y)$ and some $\lambda>0$. Further, by Proposition~\ref{neglect}, there is some $y'\in Y-y$ which is neglected by $T$. Let $p'=p(x,y',\lambda)$. As $T$ neglects $y'$ and is a \deltatour in $G-y$, we have $\dist_G(p,T)=\dist_G(p',T)\leq \delta$.
Hence $T$ is a \deltatour in $G$. 
\end{proof}
We are now ready to prove the existence of the desired kernel.
\begin{lemma}\label{findkernel}
Let $G$ be a connected graph, $\delta>0$ with $\delta<\frac{3}{2}$ and $K>0$. Then the problem of deciding whether $G$ admits a \deltatour of length at most $K$ admits a kernel of size $f(K,\delta)$ that can be computed in $n^{\Oh(1)}$ for some computable function $f:\mathbb{R}_{\geq 0}^2\rightarrow \mathbb{Z}_{\geq 0}$.
\end{lemma}
\begin{proof}
We first use a brute force algorithm to check whether $G$ contains a \deltatour of length at most $K$ stopping at at most 2 points. This is possible in $n^{\Oh(1)}$ by \cref{lemma:discretization}. If this is the case, then $K_1$ serves as the desired kernel.
We next greedily compute a maximal matching in $G$. Let $X$ be the set of vertices covered by this matching and let $k=\abs{X}$. Observe that every vertex cover needs to contain at least one vertex of each edge in the matching, hence every vertex cover of $G$ is of size at least $\frac{1}{2}k$.
Thus, if  $\frac{1}{2}k>K/s_\delta+1$, by Lemma~\ref{vertexcover}, we obtain that the desired tour does not exist. Hence, a path on $\lceil K+4 \delta\rceil+1$ vertices may serve as the desired kernel. We may therefore assume that  $\frac{1}{2}k\leq K/s_\delta+1$. Let $x_1,\ldots,x_k$ be an arbitrary ordering of $X$. Now for every $I \subseteq [k]$, let $Y_I$ be the set of vertices $y \in V(G)-X$
	with neighborhood $N_G(y)=\{x_i:i \in I\}$.
Observe that for all $I \subseteq [k]$, we have that $Y_I$ is a set of false twins. Further, as $X$ is a vertex cover of $G$, we have that $(X,(Y_I:I \subseteq [k]))$ is a partition of $V(G)$. Now for every $I \subseteq [k]$, let $Y_I'$ be an arbitrary subset of $Y_I$ of size $\min\{\abs{Y_I}, \lceil K/s_\delta+2 \rceil\}$. Further, let $G'=G[X \cup \bigcup_{I \subseteq [k]}Y'_I]$. Repeatedly applying Lemma~\ref{eliminate}, we obtain that $G'$ admits a \deltatour of length at most $K$ if and only if $G$ admits a \deltatour of length at most $K$. Hence $G'$ is a kernel for the problem. Clearly, $G'$ can be computed in $n^{\Oh(1)}$. Finally, we have that $\abs{V(G')}=\abs{X}+\sum_{I \subseteq [k]}\abs{Y'_I}\leq k+2^{k} \lceil K/s_\delta+2 \rceil\leq 2 (K/s_\delta+1)+2^{2K/s_\delta+1} \lceil K/s_\delta+2 \rceil$.
\end{proof}

Lemma~\ref{findkernel} and \cref{decide} directly imply Theorem~\ref{fpt:sol}.

\subsection[\texorpdfstring{Hardness with Respect to Tour Length for $\delta\ge3/2$}{Hardness with Respect to Tour Length for delta > 3/2}]{\boldmath Complexity results for fixed $\delta\ge3/2$}\label{section:w1hard-wrt-opt}
In Section~\ref{section:parametrized-by-opt}, we have showed that for $\delta<3/2$, the problem of computing a shortest \deltatour is fixed parameter tractable with respect to the length of a shortest \deltatour. The first objective of this section is to rule out a similar result when $\delta \geq 3/2$. We heavily make use of Lemma~\ref{lem:dom_tour_reduction}. 

The following result was proved by Karthik et al.~in \cite{KarthikLM19}.

\begin{theorem}\label{graphs}
Unless $\fpt = \wone$, for any function $f\colon \mathbb{Z}_{\geq 0}\to \mathbb{Z}_{\geq 0}$, there is no algorithm that, given a graph~$G$ and an integer $k$, returns `yes' if $G$ admits a dominating set of size $k$, returns `no' if $G$ does not admit a dominating set of size $f(k)$,  and runs in $f(k)n^{O(1)}$.
\end{theorem}
Again, using literally the same reduction as for Theorem 2 in \cite{RamanS08}, we can obtain the same result when restricting to split graphs.
\begin{theorem}\label{w1splitgraphs}
Unless $\fpt = \wone$, for any function $f\colon \mathbb{Z}_{\geq 0}\to \mathbb{Z}_{\geq 0}$, there is no algorithm that, given a split graph~$G$ and an integer $k$, returns `yes' if $G$ admits a dominating set of size $k$, returns `no' if $G$ does not admit a dominating set of size $f(k)$,  and runs in $f(k)n^{O(1)}$.
\end{theorem}

We are now ready to prove the main result of this section that rules out the possibility of an asymptotic approximation in FPT time.
\begin{theorem}\label{fptdeltasmall}
Let $\delta \geq 3/2$ be fixed. Unless $\fpt = \wone$, for any function $f\colon \mathbb{R}_{\geq 0}\to \mathbb{Z}_{\geq 0}$, there is no algorithm that, given a connected graph~$G$ and a constant $K$, returns `yes' if $G$ admits a \deltatour of length $K$, returns `no' if $G$ does not admit a \deltatour of length $f(K)$ and runs in $f(K)|V(G)|^{O(1)}$.
\end{theorem}
\begin{proof}
Suppose that for some function $f\colon \mathbb{R}_{\geq 0}\to \mathbb{Z}_{\geq 0}$, there is an algorithm $A$ that, given a connected graph~$G$ and a constant $K$, returns `yes' if $G$ admits a \deltatour of length $K$, returns `no' if $G$ does not admit a \deltatour of length $f(K)$ and runs in $f(K)|V(G)|^{O(1)}$. Without loss of generality, we may suppose that $f$ is monotonously increasing. We define $f'\colon\mathbb{Z}_{\geq 0}\rightarrow \mathbb{Z}_{\geq 0}$ by $f'(k)=\max\{\delta^c,1/s_\delta\}f(k)$ for all $k\in \mathbb{Z}_{\geq 0}$ where $c$ is the constant hidden in the polynomial time expression in Lemma~\ref{lem:dom_tour_reduction}. We will show that $\fpt = \wone$ using Theorem~\ref{w1splitgraphs}.

Let $G$ be a split graph and $k$ a positive integer. We first test in $O(|V(G)|^2)$ whether $G$ admits a dominating set of size at most 2. If this is the case, there is nothing to prove. Next observe that every subgraph of a split graph is also a split graph. Hence, if $G$ is disconnected, we may apply our algorithm to every connected component of $G$. We may therefore assume that $G$ is nontrivial and connected. 

It then follows by Lemma~\ref{lem:dom_tour_reduction} that we can compute a graph~$G'$ satisfying (1), (2), and (3) in polynomial time. We now apply $A$ to $(G',k)$ and output the output of $A$.

First suppose that $G$ admits a dominating set of size $k$. Then, by $(2)$ of \cref{lem:dom_tour_reduction}, we obtain that $G'$ admits a \deltatour of length $k$ and hence $A$ outputs `yes' by assumption.

Now suppose that $A$ outputs `yes.' By assumption, we have that $G'$ admits a \deltatour of length at most $f(k)$. It follows by $(3)$ of \cref{lem:dom_tour_reduction} that $G$ admits a dominating set of size at most $f(k)/s_\delta\leq f'(k)$. 

 Finally observe that by (1), the total running time of our algorithm is $f(k)\delta^{c}|V(G)|^{c}=f'(k)|V(G)|^{O(1)}$.

It hence follows that $\fpt = \wone$ by Theorem~\ref{w1splitgraphs}.
\end{proof}
As an immediate corollary, we have the following result, which contrasts with Theorem~\ref{fpt:sol}.
\begin{corollary}\label{asdfasd}
For every $\delta \geq \frac{3}{2}$, \deltatour is $\wone$-hard parameterized by the length of a shortest \deltatour. 
\end{corollary}
Observe that \cref{asdfasd} and \cref{fpt:sol} imply \cref{thm:parammain}.
In the following, we complement \cref{fptdeltasmall} by the following simple result showing that an XP algorithm is available.
\begin{theorem}\label{xpdeltafix}
Let a graph $G$ and constants $\delta,K>0$ be given. Then, in $n^{f(K,\delta)}$, we can decide whether $G$ admits a \deltatour of length at most $K$.
\end{theorem}
\begin{proof}
Let $P_\delta(G)$ be the collection of points in $P(G)$ which are of the form $p(u,v,\lambda)$ for some $uv \in E(G)$ and $\lambda \in S_\delta$. We now check for every sequence of points in $P_\delta(G)$ of length at most $\lceil \frac{K}{s_\delta}\rceil$ whether this sequence is a \deltatour in $G$ of length at most $K$. If this is the case, we output the sequence. If we try all sequences without finding such a tour, we report that $G$ does not admit a \deltatour of length at most $K$. This finishes the description of the algorithm.

For the correctness of the algorithm, first observe that if the algorithm outputs a tour, then it is a \deltatour in $G$ of length at most $K$ by construction. Now suppose that $G$ admits a \deltatour $T$ of length at most $K$. By \cref{lemma:discretization}, we may suppose that all stopping points of $T$ are in $P_\delta(G)$. It follows by \cref{lentour} that $\alpha(T)\leq \lceil\frac{\len(T)}{s_\delta}\rceil\leq \lceil\frac{K}{s_\delta}\rceil$. It hence follows by construction that our algorithm outputs a \deltatour of $G$ of length at most $K$.

For the analysis of the algorithm, observe that, as $|S_\delta|=4$, we have $|P_\delta(G)|\leq 8|E(G)|\leq 8n^2$. It follows that the total number of sequences we consider is bounded by $(8n^2)^{\frac{K}{s_\delta}}=n^{\Oh(\frac{K}{s_\delta})}$. Further, for every sequence, it is not difficult to see that we can check in polynomial time whether it is a tour of length at most $K$. Finally, by \cref{check}, we can check in polynomial time whether it is a \deltatour. Hence our algorithm has the desired running time.
\end{proof}
\subsection[\texorpdfstring{Hardness for $\delta$ Being Part of the Input}{Hardness for delta Being Part of the Input}]{\boldmath Hardness for $\delta$ Being Part of the Input}\label{dfguhs}
This section is dedicated to showing that not even an XP algorithm can be hoped for if $\delta$ is part of the input. More concretely, we prove \cref{nphdinp}, which we restate here for convenience.
\nphdinp*

Observe that \cref{xpdeltafix} and \cref{nphdinp} imply \cref{thm:parammain2}.
The main technicality is contained in the following main lemma.

\begin{lemma}\label{detgzhu}
Let $G$ be an incomplete connected split graph, $k$ a positive integer, and $\eps>0$ a constant. Then, in polynomial time, we can construct a graph $G'$, a constant $\delta$ with $0<\delta<2$, and a constant $K$ with $K<\eps$ such that the following hold:
\begin{enumerate}[(i)]
\item if $G$ admits a dominating set of size at most $k$, then $G'$ admits a \deltatour of length at most $K$,
\item if $G'$ admits a \deltatour of length at most $K$, then $G$ admits a dominating set of size at most $2k+1$.
\end{enumerate}
\end{lemma}
\begin{proof}
Let $(C,I)$ be a partition of $V(G)$ such that $G[C]$ is a clique and $I$ is an independent set in $G$. Observe that such a partition exists as $G$ is a split graph and it can be computed in polynomial time. Further, as $G$ is incomplete, we have $I \neq \emptyset$. We now obtain $G'$ from $G$ by adding three new vertices $v_0,v_1,v_2$ and the edges $v_0v_1,v_1v_2$, and $v_0c$ for all $c \in C$. Now choose some $\eps_0>0$ such that $(4k+2)\eps_0<\min\{\eps,\frac{1}{2}\}$. We set $\delta=2-\eps_0$ and $K=(4k+2)\eps_0$. This finishes the description of $(G',\delta,K)$. Observe that $(G',\delta,K)$ can be computed in polynomial time and $0<\delta<2$ and $K<\eps$ hold. We still need to show that $(i)$ and $(ii)$ hold.

First suppose that $G$ admits a dominating set $S$ with $|S|\leq k$. By \cref{propsplit}, we may suppose that $S \subseteq C$. Let $s_1,\ldots,s_q$ be an arbitrary enumeration of $S$. We now define a tour $T$ in $G'$ by $T=v_0p(v_0,v_1,\eps_0)v_0p(v_0,s_1,2\eps_0)v_0\ldots v_0p(v_0,s_q,2\eps_0)v_0$. It is easy to see that $T$ is indeed a tour in $G'$. Further, by construction, we have $\ell(T)=(4q+2)\eps_0\leq K$. We still need to show that $T$ is a \deltatour in $G'$.

To this end, consider some $x_1x_2 \in E(G')$. First suppose that $v_0 \in \{x_1,x_2\}$, say $x_1=v_0$. Then, as $T$ stops at $v_0$ and $0 \geq 2 \eps_0-2=2-2 \delta$, we obtain that $x_1x_2$ satisfies $(i)$ of \cref{char1stop} and is hence covered by $T$. Next suppose that $\{x_1,x_2\}=\{v_1,v_2\}$, say $x_1=v_1$ and $x_2=v_2$. As $T$ stops at $p(v_0,v_1,\eps_0)$ and $2 \eps_0=4-2\delta=\dist_{G'}(x_1,v_1)+\dist_{G'}(x_2,v_1)+3-2\delta$, we obtain that $x_1x_2$ satisfies $(i)$ of \cref{char2stops} and is hence covered by $T$. Next suppose that $\{x_1,x_2\}\subseteq C$. As $T$ stops at $v_0$ and $0\geq 3-2\delta=\dist_{G'}(x_1,x_1)+\dist_{G'}(x_2,x_2)+3-2\delta$, we obtain that $x_1x_2$ satisfies $(i)$ of \cref{char2stops} and is hence covered by $T$. Hence, by symmetry, we may suppose that $x_1 \in C$ and $x_2 \in I$.  As $S$ is a dominating set in $G$, there is some $s \in S \cap N_{G}(x_2)$. As $T$ stops at $v_0$ and $p(v_0,s,2 \eps_0)$ and $2 \eps_0= 4-2 \delta=\dist_{G'}(x_1,x_1)+\dist_{G'}(s,x_2)+3-2\delta$, we obtain that $x_1x_2$ satisfies $(i)$ of \cref{char2stops} and is hence covered by $T$. It follows that $T$ is \deltatour in $G'$. Hence $(i)$ follows.

Now suppose that $G'$ admits a \deltatour $T$ of length at most $K$.
\begin{claim}\label{v_0stop}
$T$ stops at $v_0$.
\end{claim}
\begin{proof}Observe that $I \neq \emptyset$ as $G$ is incomplete. Let $z \in I$.
As $T$ is a \deltatour and by \cref{nearstop}, there are stopping points $p_1,p_2$ of $T$ such that $\dist_G(p_1,z)\leq \delta$ and $\dist_G(p_2,v_2)\leq \delta$. By construction, we obtain that $p_1=p(c,v,\lambda)$ for some $c \in C, v\in N_{G'}(c)$, and $\lambda \in [0,1]$ and $p_2=p(v_1,v,\lambda)$ for some $v \in N_{G'}(v_1)$ and $\lambda \in [0,1]$. Next, as $T$ is a tour, we obtain that $T$ contains a $p_1p_2$-walk. Further, it follows by construction that every $p_1p_2$-walk in $G'$ stops at $v_0$. Hence $T$ stops at $v_0$.
\end{proof}
By \cref{v_0stop}, as $T$ is nice and $\ell(T)\leq K<1$, we obtain for some $u_1,\ldots,u_q \in N_{G'}(v_0)$ and $\lambda_1,\ldots,\lambda_q \in (0,1)$ that $T=v_0p(v_0,u_1,\lambda_1)v_0\ldots v_0p(v_0,u_q,\lambda_q)v_0$. We now let $S$ contain all vertices $s \in C$ for which there is some $i \in [q]$ such that $u_i=s$ and $\lambda_i \geq \eps_0$. By construction, we have $(4k+2)\eps_0=K\geq \ell(T)\geq 2 \eps_0|S|$, so $|S|\leq 2k+1$. We still need to show that $S$ is a dominating set of $G$. To this end, consider some $z \in I$. As $T$ is a \deltatour in $G'$ and by \cref{nearstop}, there is a stopping point $p$ of $T$ such that $\dist_G(p,z)\leq \delta$. By the above, we have that $p=p(v_0,u,\lambda)$ for some $u \in C \cup \{v_1\}$ and $\lambda \in (0,1)$. If $u=v_1$, we have $\dist_G(p,z)=2+\lambda>\delta$, a contradiction. If $u \in C-N_{G}(z)$, we have $\dist_G(p,z)\geq \min\{\dist_G(v_0,z),\dist_G(u,z)\}=2>\delta$, a contradiction. If $u \in N_{G}(z)$ and $\lambda<\eps_0$, we have $\dist_G(p,z)= 2-\lambda>2-\eps_0=\delta$, a contradiction. We hence obtain that $u \in N_G(z)$ and $\lambda\geq \eps_0$. By the definition of $S$, we obtain in particular $S \cap N_G(u)\neq \emptyset$. It follows that $S$ dominates $I$. As $G$ is a connected split graph and $S  \subseteq C$, we obtain that $S$ is a dominating set of $G$. Hence $(ii)$ follows.
\end{proof}
We are now ready to give the main proof of \cref{nphdinp}.
\begin{proof}[Proof of \cref{nphdinp}]
Let some $\eps>0$ be fixed and suppose that there is a polynomial-time algorithm $A$ that, given a graph $G$, a constant $\delta$ with $0<\delta<2$, and a constant $K$ with $K<\eps$, correctly decides whether $G$ admits a \deltatour of length at most $K$. We will obtain a contradiction using \cref{rabsaur}.

 Let $G$ be a split graph on $n$ vertices. If $G$ is disconnected, we can solve the problem for every connected component of $G$ and output an appropriate answer. We may hence suppose that $G$ is connected. 
 Next, if $k \leq 1$ or $n\leq 2^{3/\alpha_0}$, where $\alpha_0$ is the constant from \cref{rabsaur}, we can solve the problem by a brute-force approach. We may hence suppose that $k \geq 2$ and $n\geq 2^{3/\alpha_0}$. Finally, if $G$ is a complete graph, we can clearly solve the problem in polynomial time. We may hence suppose that $G$ is incomplete. Now, using \cref{detgzhu}, we can compute in polynomial time a graph $G'$, a constant $\delta$ with $0<\delta<2$, and a constant $K$ with $K<\eps$ such that $(i)$ and $(ii)$ hold. We now run $A$ on $(G',\delta,K)$ and output the output of $A$. This finishes the description of our algorithm. Observe that it runs in polynomial time.

First suppose that $G$ admits a dominating set of size at most $k$. Then by $(i)$, we have that $G'$ admits a \deltatour of length at most $K$. By the assumption on $A$, it follows that our algorithm outputs `yes.'

Now suppose that our algorithm outputs `yes.' By assumption, we obtain that $G'$ admits a \deltatour of length at most $K$. It follows by $(ii)$ that $G$ admits a dominating set of size at most $2k+1$. As $k\geq 2$ and $n\geq 2^{3/\alpha_0}$, we have $2k+1<3k\leq \alpha_0 \log(n)k$. 
Hence \cref{rabsaur} yields $\p=\np$. 
\end{proof}

\subsection[\texorpdfstring{XP Algorithm Parameterized by $n/\delta$}{XP Algorithm Parameterized by n/delta}]{XP Algorithm Parameterized by \boldmath$n/\delta$}
\label{section:by-n-over-delta}
The problem \deltatourprob is easy for $\delta=0$, where it corresponds to the \CPP; see~\cite{FreiGHHM24} for the details. 
Further, the problem is also trivial for the other extreme, that is, when $\delta$ is very large. Indeed, if $\delta\ge n$ (where $n$ denotes the number of vertices in the graph), any single point $\delta$-covers the entire graph and hence forms a shortest \deltatour. 
We are interested in how the hardness of the problem increases when $\delta$ is decreased from this upper point of triviality. 
We therefore analyze the problem \deltatourprob parameterized by $n/\delta$. (Note that this parameter is 1 for the trivial case of $\delta=n$.)
We first present in this section an XP-algorithm and then rule out, in the following section, the possibility of an FPT algorithm 
under the assumption that $\fpt\neq\wone$.

The main result of this section is Theorem~\ref{thm:param:k:ub:largedelta} which we restate here.
\ThmParamKUbLargeDelta*\label\thisthm
We start by giving the key lemma for the proof of Theorem~\ref{thm:param:k:ub:largedelta} which is restated below.

\MainLemXP*\label\thisthm
\begin{proof}
We denote by $X$ the set of vertices $x \in V(G)$ with $\dist_G(x,T)\geq \frac{n}{2k}$.
Let $Z_1$ be an inclusion-wise minimal set of points passed by $T$ such that
	$\dist_G(x,T)=\dist_G(x,Z_1)$ for all $x \in X$. By Proposition~\ref{nearstop}, we may suppose that $T$ stops at $z$ for every $z \in Z_1$.

\begin{claim}
	$\abs{Z_1} \leq 3k$.
\end{claim}
\begin{proof}
Clearly, we have $\abs{Z_1} \leq n$, so $Z_1$ is finite. Hence, by the definition of $Z_1$, there is an ordering $z_1,\dots,z_q$ of $Z_1$ and a collection $x_1,\dots,x_q$ of elements of $X$ such that $\dist_G(x_i,z_i)=\dist_G(x_i,T)$ for $i \in [q]$ and $\dist_G(x_i,z_j)>\dist_G(x_i,T)$ for $i,j \in [q]$ with $i \neq j$. For $i \in [q]$, let $P_i$ be a shortest $\seq{x_i&z_i}$-walk in $G$. We may suppose that for $i \in [q]$, all the points distinct from $x_i$ and $z_i$ that $P_i$ stops at are vertices of $G$. Further, by the minimality of $P_i$, we have that $P_i$ stops at most once at every vertex of $G$. We denote the set of vertices of $G$ that $P_i$ stops at by $V(P_i)$.  
First suppose that there are distinct $i,j \in [q]$ with $V(P_i)\cap V(P_j)\neq \emptyset$ and let $v \in V(P_i)\cap V(P_j)$. This yields
\begin{align*}
\dist_G(x_i,z_i)+\dist_G(x_j,z_j)&=\dist_G(x_i,v)+\dist_G(v,z_i)+\dist_G(x_j,v)+\dist_G(v,z_j)\\
&=(\dist_G(x_i,v)+\dist_G(v,z_j))+(\dist_G(x_j,v)+\dist_G(v,z_i))\\
&\geq \dist_G(x_i,z_j)+\dist_G(x_j,z_i)\\
&> \dist_G(x_i,z_i)+\dist_G(x_j,z_j),
\end{align*}
a contradiction.

We hence obtain that $V(P_i)\cap V(P_j)=\emptyset$ for all $i,j \in [q]$ with $i \neq j$. Further, observe that for every $i \in [q]$, as the length of $P_i$ is at least $\frac{n}{2k}$, we have that $\abs{V(P_i)}\geq \frac{n}{2k}-1$. This yields $n \geq \sum_{i=1}^q \abs{V(P_i)}\geq q (\frac{n}{2k} -1)$. As $n \geq 12k$, we obtain
$q \leq \frac{2k}{1-\frac{2k}{n}} \leq 3k$.
\end{proof}

\begin{claim}\label{trgvftzuh-application-1}
There is a set $Z_2$ of points stopped at by $T$ with $\abs{Z_2} \leq 9k$ such that $\dist_G(p,Z_2)\leq \frac{n}{3k}$ holds for all $p \in P(G)$ which are stopped at by $T$.
\end{claim}
\begin{proof}
Let $\floor{T}=v_1\dots v_tv_1$ be the truncation of $T$.
Further, let $\beta$ be the largest integer such that
$1+\beta \floor{\frac{n}{2k}}\leq t$ and let $Z_2 = \{v_{1 + i
\floor{\frac{n}{2k}}} \mid i \in \{0, \ldots,\beta\}\}$, where only a single copy of a vertex is maintained in case it occurs several times.
Now consider a point $p$ stopped at by $T$. Then, clearly, there is some
$i \in [t]$ with $\dist_G(p,v_i)\leq 1$. Further, either $t-i+1\leq \frac{n}{4k}$ or there is some
$j \in \{1 + i \floor{\frac{n}{2k}} \mid i \in \{0,\ldots, \beta\}\}$ such that
 $\abs{i-j}\leq \frac{n}{4k}$. We obtain $\dist_G(p,Z_2)\leq
\dist_G(p,v_j)\leq \dist_G(p,v_i)+\dist_G(v_i,v_j)\leq 1 +\abs{i-j}\leq
1+\frac{n}{4k} \leq \frac{n}{3k}$ as $n \geq 12k$. If $\delta \leq \frac{1}{2}$, we obtain $k=\lceil\frac{n}{\delta}\rceil\geq n$, a contradiction. We hence have $\delta\geq \frac{1}{2}$, so by the optimality of $T$ and \cref{trgvftzuh}, it follows that $\len(\floor{T})\leq \len(T)\leq 2n-2$.
As $n \geq 12k$, and $k \geq 1$,
we obtain that
$\abs{Z_2} \leq \beta+1 \leq
	\frac{\len(\floor{T})}{\floor{\frac{n}{2k}}}+1 \leq
	\frac{2n}{\floor{\frac{n}{2k}}}+1\leq \frac{2n}{\frac{n}{4k}}+1 \leq 8k+ 1 \leq 9k$.
\end{proof}
Let $Z=Z_1 \cup Z_2$ and observe that $\abs{Z} \leq 12k$.
Now let $p=(u,v,\lambda) \in P(G)$ for some $uv \in E(G)$ and some $\lambda\in [0,1]$. First suppose that $\{u,v\}-X\neq \emptyset$, say $\dist_G(u,T)\leq \frac{n}{2k}$. Then there is a point $p'$ stopped at by $T$ with $\dist_G(u,p')\leq \frac{n}{2k}$. By the definition of $Z_2$, we obtain $\dist_G(p,Z)\leq \dist_G(p,u)+\dist_G(u,p')+\dist_G(p',Z_2)\leq 1+\frac{n}{2k}+\frac{n}{3k} \leq n/k$, as $n \geq 12k$. Now suppose that $\{u,v\}\subseteq X$. As $n \geq 12k$ and $T$ is a feasible \deltatour of $G$, we then have
\begin{align*}
\dist_G(p,Z)&\leq \min\{\dist_G(p,u)+\dist_G(u,Z_1), \dist_G(p,v)+\dist_G(v,Z_1)\}\\
&= \min\{\dist_G(p,u)+\dist_G(u,T), \dist_G(p,v)+\dist_G(v,T)\}\\
&=\dist_G(p,T)\\
&\leq \delta.
\end{align*}
\end{proof}

Given a graph~$G$ and a positive constant $\delta$, we denote by $S_{G,\delta}$ be the set of points $(u,v,\lambda)\in P(G)$ such that $uv \in E(G)$ and $\lambda \in S_\delta$.
We can now prove a structural result which is the key to our algorithm.
\begin{lemma}\label{drftgzhu}
Let $G$ be a connected graph of order $n$, $\delta$ a positive real, $k=\lceil\frac{n}{\delta}\rceil$ and suppose that $n \geq 12 k$. Then there is a sequence $(z_1,\dots,z_q)$ of points in $S_{G,\delta}$ with $q \leq 12k$ such that for any collection of walks $Q_1,\dots,Q_q$ in $G$ such that $Q_i$ is a shortest $z_iz_{i+1}$-walk in $G$ for $i \in [q-1]$ and $Q_q$ is a shortest $z_qz_1$-walk in $G$, we have that the concatenation $Q_1\dots Q_q$ is a shortest \deltatour. 
\end{lemma}
\begin{proof}
As $n \geq 12k$, we have $\delta>1$. Hence, by Lemma~\ref{lemma:discretization}, there is a shortest \deltatour~$T$ of $G$ such that all the points at which $T$ stops are contained in $S_{G,\delta}$. By Lemma~\ref{lem:main:lem:xp}, there is a collection $Z=\{z_1,\dots,z_q\}$ of $q \leq 12k$ points stopped at by $T$ such that $\dist_G(p,Z)\leq \delta$ for all $p \in P(G)$. By symmetry, we may suppose that the points of $Z$ are stopped at in the order $(z_1\dots z_q)$ by $T$. Now let $Q_1,\dots,Q_q$ be a collection of paths in $G$ such that $Q_i$ is a shortest $\seq{z_i&z_{i+1}}$-walk in $G$ for $i \in [q-1]$ and $Q_q$ is a shortest $z_qz_1$-walk in $G$ and let $T^*$ be the concatenation $Q_1\dots Q_q$. It follows by construction that $T^*$ is a tour in $G$. Further, as $\dist_G(p,Z)\leq  \delta$ for all $p \in P(G)$ and all points of $Z$ are stopped at by $T^*$, we obtain that $T^*$ is a \deltatour in $G$. Next, for $i \in [q-1]$, let $Q_i'$ be the subwalk of $T$ from $z_i$ to $z_{i+1}$ corresponding to the chosen occurrences of $z_i$ and $z_{i+1}$ and let $Q_q'$ be the subwalk of $T$ from $z_q$ to $z_{1}$ corresponding to the chosen occurrences of $z_q$ and $z_{1}$. Observe that $T$ is the concatenation $Q'_1\ldots Q'_q$. Further, by the choice of $Q_i$, we have $\len(Q_i)\leq \len(Q_i')$ for $i \in [q]$. This yields $\len(T^*)=\sum_{i=1}^{q} \len(Q_i)\leq \sum_{i=1}^{q} \len(Q'_i)=\len(T)$. Hence $T^*$ is a shortest \deltatour in $G$.
\end{proof}
We are now ready to prove our main result.
\begin{proof}(Proof of Theorem~\ref{thm:param:k:ub:largedelta})
Let $G$ be a connected graph of order $n$, $\delta>0$ a constant and $k=\lceil \frac{n}{\delta}\rceil$.

If $n<12k$, then, by \cref{decide}, we can find a shortest \deltatour in $G$ in $f(k)$.

We may hence suppose that $n \geq 12k$. We first compute a \deltatour~$T_0$ in $G$ of length at most $2n-2$ which can be done in $n^{\Oh(1)}$ by Proposition~\ref{trgvftzuh}. Throughout the algorithm, we maintain a \deltatour~$T$ in $G$ that is shortest among all those found so far. It hence suffices to prove that we find a shortest \deltatour in $G$ at some point during the algorithm.
Let $(z_1,\dots,z_q)$ be a sequence of points in $S_{G,\delta}$. We now compute a collection of walks $Q_1,\dots,Q_q$ in $G$ such that $Q_i$ is a shortest $\seq{z_i&z_{i+1}}$-walk in $G$ for $i \in [q-1]$ and $Q_q$ is a shortest $z_qz_1$-walk in $G$. We then let $Q$ be the tour in $G$ obtained by concatenating $Q_1,\dots,Q_q$. If $Q$ is a \deltatour in $G$ and $Q$ is shorter than $T$, we replace $T$ by $Q$. We now do this for every sequence $(z_1,\dots,z_q)$ of points in $S_{G,\delta}$ for $q \leq 12k$. By Lemma~\ref{drftgzhu}, we obtain that $T$ is assigned a shortest \deltatour in $G$ at some point during this process.

We still need to analyze the running time of this algorithm. Note that
the shortest paths between all pairs of vertices in $V(G)$ can be computed in time
$\Oh(n^3)$ by the Floyd–Warshall algorithm~\cite{CLRS}. After, given two points in $P(G)$, a shortest walk linking these points can be computed in constant time. Indeed, if both points are on the same edge, it is trivial to compute. Otherwise, such a walk can be obtained by attaching these points to a shortest path from a vertex incident to the edge the first point is on to a vertex incident to the edge the second point is on. There are only four choices for these vertices.

It follows that, given a sequence $(z_1,\dots,z_q)$ of points in $S_{G,\delta}$, each of the walks $Q_i$ for $i \in [q]$ can be computed in $\Oh(1)$-time, 
given the all-pair shortest paths computed above. Further, it follows directly from Corollary~\ref{check} that we can check in $\Oh(n^c)$-time for some fixed constant  whether $Q$ is a \deltatour. Observe that, as $G$ does not contain multiple edges, we have $\abs{S_{G,\delta}}= \Oh(n^2)$. Given an integer $q$, there are hence at most $\Oh(n^{2q})$ possibilities to choose $\{z_1,\dots,z_q\}$ and for each of them we need to consider $q!=f(k)$ orderings. Hence the total running time is
$\Oh(q!n^{2q}n^{c}) = \Oh(f(k)n^{24k+c})$.
\end{proof}

\subsection[\texorpdfstring{\wone-Hardness When Parameterized by $n/\delta$}{W[1]-Hardness When Parameterized by n/delta}]{\wone-Hardness When Parameterized by \boldmath$n/\delta$}
\label{section:w1-hard-n-over-delta}

In this subsection, we show that the running time of the algorithm implied
by \cref{thm:param:k:ub:largedelta} is conditionally optimal up to a $\log{k}$-factor in the
exponent.
More precisely, we show that the problem is \wone-hard for the parameterization
by $k \coloneqq\ceil{\frac{n}{\delta}}$, and, moreover, unless \ETH fails, the
\xp-time algorithm from \cref{thm:param:k:ub:largedelta} cannot be
significantly improved in the regime $\delta = \Omega(n)$. We prove \cref{thm:param:k:lb:largedelta}, which we restate here for convenience.

\ThmParamKLbLargeDelta*\label\thisthm

We base our hardness on that of the Binary Constraint Satisfaction Problem (\binaryCSP) in
cubic graphs, defined as follows.
A \binaryCSP instance $I$ is
a triple $(V, \Sigma, \Ce)$, where
$V = \{v_1, \dots, v_{|V|}\}$ is a set of \textit{variables}
taking values from a \textit{domain} $\Sigma = [n]$, and
$\mathcal{C}$ is a set of $k$ \textit{constraints}, each of which
is a triple $(v_i, v_j, C_{i,j})$ with $1 \leq i < j \leq |V|$, and
$C_{i,j} \subseteq \Sigma \times \Sigma$ is a relation constraining variables
$v_i$ and $v_j$. The instance $I$ is \textit{satisfiable}
if there is an assignment to the variables $\mathcal{A} \colon  V \to \Sigma$ such
that $(\assig{i}, \assig{j}) \in C_{i,j}$ for every
constraint relation $C_{i,j}$. \binaryCSP  is the problem of deciding whether a given instance is satisfiable. 

The instance $I$ is associated with a constraint graph~$G$, with
$V(G) \coloneqq V$ and where two vertices are adjacent if
there is a constraint between them.

\binaryCSP is \wone-hard when parameterized by the number of constraints $k$.
In~\cite{Marx07}, Marx showed that an $f(k)\cdot n^{\oh(k/\log{k})}$-time
algorithm for \binaryCSP on cubic constraint graphs violates \ETH.
We make use of a slightly stronger formulation of that result, which was shown
by Karthik et al.~\cite{KarthikMPS23}.

\begin{theorem}[{\cite[Theorem 1.5]{KarthikMPS23}}]\label{thm:csp_hardness}
	Unless \ETH fails, there exist absolute positive constants
	$\alpha$ and $k_0\geq 6$ such that for every fixed $k \geq k_0$,
	where $k$ is divisible by $3$,
	there is no algorithm solving \binaryCSP on cubic constraint graphs
	with $k$ constraints over a domain of size $n$ 
	in time $\Oh(n^{\alpha k / \log{k}})$.%
	\footnote{Note that the exact formulation in~\cite{KarthikMPS23}
	uses $k$ to denote the number of variables (vertices), which is even.
	As the input graph is cubic, the statement we use is equivalent.}
\end{theorem}

The main technical difficulties for the proof of \cref{thm:param:k:lb:largedelta} are contained in the following result.
\begin{lemma}\label{lem:csp_to_dtour}
	For every fixed $k \geq 6$ which is divisible by $3$,
	and constant integer $r \geq 0$,
	given an instance $I = (V, \Sigma, \Ce)$ of \binaryCSP in a cubic graph
	$G$ with $k \coloneqq |E(G)|$ over a domain of size $n$,
	we can compute
	a graph~$G'$ in $f(k)n^{\Oh(1)}$-time of order
	$n' = (9k + r) \delta$, for some integer $\delta$, with $n' = \Oh(n^5)$
	such that $I$ is satisfiable if and only if there is a \deltatour of $G'$
	of length at most $k(\delta+1)$.
\end{lemma}

\begin{proof}
	Let $I = (V, \Sigma, \Ce)$ be a \binaryCSP instance whose constraint graph is a cubic
	graph~$G$ with $V = \{v_1, \dots, v_{|V|}\}$ and $\Sigma = [n]$ for some positive integer $n$. For some $e \in E(G)$, we use $C_e$ for the unique constraint in $\mathcal{C}$ corresponding to $e$. Further, let $k=|\Ce|$ and observe that $|V|=\frac{2}{3}k$ as $G$ is cubic.

	First suppose that $n < k$. Then $I$ can be decided in $f(k)$-time by a brute-force approach. In order to finish this case, we still need to show that a yes-instance and a no-instance of the desired form can efficiently be constructed. Let $\delta \coloneqq 1$
	and $n' = 9k + r$.
	If $I$ is satisfiable, let $n_0' \coloneqq k  + 2$; otherwise,
	we let $n_0' \coloneqq k + 3$.
	We then return a graph~$G'$ of order $n'$ constructed from a path
	$v_0, \ldots, v_{n_0'}$ of length $n_0'$ with
	additional vertices $u_1, \ldots, u_{n'-n_0'}$ and edges
	$v_1 u_1, v_1 u_2, \ldots, v_1 u_{n'-n_0'}$. The shortest
	\deltatour~$T$ of $G'$ is of the form
	$T = (v_1, \ldots, v_{n_0'-1}, v_{n_0'-2}, \ldots, v_1)$ of
	length $\len(T) = 2(n_0'-2)$.
	Then we have that $\len(T) \leq 2k = k(\delta+1)$ if and only if $I$ is
	satisfiable.
\medskip

	Therefore, we may assume that $n \geq k$. Let $\delta \coloneqq 42kn^3$.
	We construct $G'$ consisting of $k$ subgraphs we call \textit{constraint gadgets}
	and a collection of further gadgets connecting them, which we detail in what
	follows.

	We first construct a constraint gadget $\Gamma_{e}$ for every $e\in E(G)$. See \cref{fig:constraint_gadget} for an illustration.  Let $e = uv\in E(G)$. By symmetry, we may suppose that there are $i,j \in [|V|]$ with $i<j$ such that $u=v_i$ and $v=v_j$.

	Let $\Gamma_{e}$ include the following four paths of length $\delta-4n$,
	which we call \textit{connectors}:
	\begin{itemize}
		\item a path $\conn{R}{e}{u}{0}$ between vertices $\conn{s}{e}{u}{0}$ and $\conn{t}{e}{u}{0}$,
		\item a path $\conn{R}{e}{u}{1}$ between vertices $\conn{s}{e}{u}{1}$ and $\conn{t}{e}{u}{1}$,
		\item a path $\conn{R}{e}{v}{0}$ between vertices $\conn{s}{e}{v}{0}$ and $\conn{t}{e}{v}{0}$, and
		\item a path $\conn{R}{e}{v}{1}$ between vertices $\conn{s}{e}{v}{1}$ and $\conn{t}{e}{v}{1}$.
	\end{itemize}

	We then add two paths
	$R_{e}^{\tail_1}$ and $R_{e}^{\tail_2}$, each of length $\delta-2n$
	between vertices $s_{e}^{\tail_1}$ and $t_{e}^{\tail_1}$
	and between vertices $s_{e}^{\tail_2}$ and $t_{e}^{\tail_2}$. We refer to the
	paths $R_{e}^{\tail_1}$ and $R_{e}^{\tail_2}$ as \textit{tails}.

	Let $\Phi_e \coloneqq
	\{\phi \colon  \{v_i, v_j\} \to \Sigma \mid (\phi(v_i), \phi(v_j)) \in C_{i, j}\}$,
	and let $\Gamma_e$ include a set of vertices
	$A_e \coloneqq \{a_e^\phi \mid \phi \in \Phi_e\}$, which we refer to
	as \textit{candidate} vertices. 
	Note that $|A_e|=|\Phi_e|\le n^2$. 
	For each candidate $\cand{e}{\phi} \in A_{e}$, add the following
	paths between $a_e^\phi$ and the connectors in~$\Gamma_e$: 
	\begin{itemize} 
		\item a path from $\cand{e}{\phi}$ to $\conn{s}{e}{u}{0}$ of length $2n+2\phi(u)$,
		\item a path from $\cand{e}{\phi}$ to $\conn{s}{e}{u}{1}$ of length $4n-2\phi(u)$,
		\item a path from $\cand{e}{\phi}$ to $\conn{s}{e}{v}{0}$ of length $2n+2\phi(v)$, and
		\item a path from $\cand{e}{\phi}$ to $\conn{s}{e}{v}{1}$ of length $4n-2\phi(v)$.
	\end{itemize}
	
	Finally, for all $\phi \in \Phi_e$, we create a path from
	$\cand{e}{\phi}$ to $s_{e}^{\tail_1}$ and
	a path from
	$\cand{e}{\phi}$ to $s_{e}^{\tail_2}$ of length $2n$ each. We choose these paths so that they are pairwise vertex-disjoint, except possibly for their endvertices, when it is explicitly mentioned.
	This finishes the description of $\Gamma_e$, which is the constraint gadget for $e$.

	Next, consider some $v \in V(G)$ and let
	$e_1, e_2$, and $e_3$ be the edges incident to $v$ in $G$.
	We now connect the constraint gadgets $\Gamma_{e_1}, \Gamma_{e_2}$, and $\Gamma_{e_3}$ by adding paths linking some of the endvertices of the connectors in these constraint gadgets.
	More precisely, for all distinct $i,j \in [3]$, we add the two paths
	$\midpath{e_i, v, e_j}{0}$ from~$\conn{t}{e_i}{v}{0}$~to~$\conn{t}{e_j}{v}{1}$,
	and 
	$\midpath{e_j, v, e_i}{0}$ from~$\conn{t}{e_j}{v}{0}$~to~$\conn{t}{e_i}{v}{1}$,
	of length $2n$, which we refer to as \textit{middle} paths. The graph consisting of the six paths corresponding to a vertex of $V(G)$ is called the {\it middle gadget} of the vertex. Again, all the newly added middle paths are vertex-disjoint from each other and all constraint gadgets, except for their endvertices when explicitly specified.
	\Cref{fig:gadget_connectivity} shows how the
	three gadgets sharing a variable are connected.
	Note that every vertex $v \in V(G)$ corresponds to six such middle paths
	as $G$ is cubic; thus, we add exactly $4k$ new paths.

	We next add a collection of gadgets connecting the constraint gadgets. While the form of these gadgets is independent of the actual constraints, we will later show that any tour meeting the desired length bound is forced to be mainly contained in these new gadgets. 
	Let $e_1,\ldots,e_k$ be an arbitrary
	ordering of $E(G)$. 
	We add $k$ paths
	$R_{i, i+1}$ with endpoints $y_i$ and $z_{i+1}$
	of length $\delta-4n+1$ each for all $i \in [k]$. In the following, we set $e_{k+1}=e_1$ and $z_{k+1}=z_1$.
	Moreover, we define $Y\coloneqq \{y_1,\ldots,y_k\}$ and $	Z\coloneqq \{z_1,\ldots,z_k\}$.
	Now, for $i \in [k]$ all $\phi \in \Phi_{e_i}$, we connect the vertices $y_i$ and $z_{i+1}$ to the constraint gadgets $\Gamma_{e_i}$ and $\Gamma_{e_{i+1}}$, respectively. Namely, for all $i \in [k]$ and $\phi \in \Phi_{e_i}$, we add a path of length $2n$
	from
	$a_{e_i}^\phi$ to $y_i$ and a path of length $2n$
	from
	$a_{e_i}^\phi$ to $z_i$. For any $i \in [k]$, we call the path $R_{i,i+1}$ together with the paths connecting $R_{i,i+1}$ to the corresponding constraint gadgets a {\it tour gadget}. Again, we choose our tour gadgets to be vertex-disjoint from each other and the remainder of the graph, unless explicitly specified otherwise.

	Let $G'_0$ be the graph constructed so far. We want to construct $G'$ from $G'_0$ by adding a set of vertices with the point of meeting the desired requirement on the total number of vertices. To this end, we first need to show that this limit is not yet exceeded.
	\begin{claim}
	$|V(G_0')|\leq n'$.
	\end{claim}
	\begin{proof}
	First observe that for every $e \in E(G)$, the total number of vertices in $V(\Gamma_e)$ is $4(\delta-4n+1)+2(\delta-2n+1)+|A_e|16n\leq 6\delta+16n^3$. It follows that the total number of vertices of $G_0'$ contained in a constraint gadget is at most $k(6\delta+16n^3)$. Next observe that every middle path contains exactly $2n+1$ vertices. As there are in total $4k$ middle paths, we obtain that the total number of vertices of $G_0'$ contained in a middle path is bounded by $8kn+4k$. Finally observe that by construction, for every $i \in [k]$ the number of vertices of the tour gadget containing $R_{i,i+1}$ is bounded by $(\delta-4n+2)+(|A_{e_i}|+|A_{e_{i+1}}|)2n\leq \delta+4n^3$. As $G_0'$ contains exactly $k$ tour gadgets, the total number of vertices of $G_0'$ contained in a tour gadget is at most $k(\delta+4n^3)$. We obtain that $|V(G_0')|\leq k(7 \delta+20n^3+8n+4)\leq 9k\delta\leq n'$.
	\end{proof}
	
	We now choose an arbitrary $e \in E(G)$ and add a collection of $(9k + r) \delta -|V(G_0')|$ vertices together with an edge linking it to the unique neighbor of $t_{e}^{\tail_1}$ in $G_0'$.
	Let $G'$ be the resulting graph.
	Observe that $|V(G')|=(9k+r)\delta$.
	Further, as $k \leq n$, we have $\delta=\Oh(n^4)$ and hence $n'=\Oh(n^5)$. Further, we clearly have that $G'$ can be constructed in time $\Oh(n^5)$.

	To prove the correctness of the reduction, we make use of 
	some key observations regarding~$G'$.

The following result will be used several times throughout the proof. It gives a characterization of when two candidates from different constraint gadgets are of limited distance to each other.
	\begin{claim}\label{pathlength}
Let $i,j \in [k]$ with $i\neq j, \phi_i\in \Phi_{e_i}$ and $\phi_j \in \Phi_{e_j}$. Then there is a path of length at most $\delta+1$ linking $a_{e_i}^{\phi_i}$ and $a_{e_j}^{\phi_j}$ in $G'$ if and only if $|i-j| \equiv 1 \pmod k$. Moreover, if such a path exists, then it is the unique path linking $a_{e_i}^{\phi_i}$ and $a_{e_j}^{\phi_j}$ fully contained in a tour gadget, so in particular of length exactly $\delta+1$.
\end{claim}
\begin{proof}
First suppose that such a path $Q$ exists. We may suppose that $Q$ is a shortest path linking $A_{e_{i_1}}$ and $A_{e_{i_2}}$ for distinct $i_1,i_2 \in [k]$.

 Suppose for the sake of a contradiction that $V(Q)$ contains $s_e^{u}$ for a vertex $u$ that is incident to $e$. Then, by construction, as $Q$ is a path and by the minimality of $Q$, we obtain that $Q$ contains $t_e^{u}$ and $R_e^u$. Now let $e'$ be the unique edge incident to $u$ such that one of $t_{e'}^{u}$ and $\conn{t}{e'}{u}{1}$ is the last vertex in the middle gadget of $u$ when following $Q$ from $a_{e_i}^{\phi_i}$ to $a_{e_j}^{\phi_j}$. As $Q$ is a path and by construction it follows that $Q$ contains one of $\conn{R}{e'}{u}{0}$ and $\conn{R}{e'}{u}{1}$. By construction, we have that $V(R_e^u)\cap V(\conn{R}{e'}{u}{0})=V(R_e^u)\cap V(\conn{R}{e'}{u}{1})=\emptyset$. It follows that $\delta+1\geq |V(Q)|\geq |V(R_e^u)| + \min\{|V(\conn{R}{e'}{u}{0})|,|V(\conn{R}{e'}{u}{1})|\}\geq 2(\delta-4n)>\delta$, a contradiction.

 We obtain that $Q$ does not contain $s_e^{u}$. Similar arguments show that $Q$ does not contain any of $\conn{s}{e}{u}{1},\conn{t}{e}{u}{0}$, and $\conn{t}{e}{u}{1}$. By the minimality of $Q$, as $Q$ is a path and by construction, it follows that $Q$ is fully contained in a single tour gadget. It follows directly from the existence of this gadget and construction that $|i-j|\equiv 1 \pmod k$ and that $Q$ is the unique path linking $a_{e_i}^{\phi_i}$ and $a_{e_j}^{\phi_j}$ fully contained in the tour gadget, so in particular of length exactly $\delta+1$.
 
 On the other hand, if $|i-j|\equiv 1 \pmod k$, it is easy to see that the there exists a tour gadget that contains a path linking $a_{e_i}^{\phi_i}$ and $a_{e_j}^{\phi_j}$ of length exactly $\delta+1$.
\end{proof}

	For some $p_0 \in P(G')$, let $\ball{p_0}$ be defined as $\{p \in V(G') \mid \dist_{G'}(p, p_0) \leq \delta\}$. Further, for some $S \subseteq P(G')$, we set $\ball{S}=\bigcup_{u \in S}\ball{u}$. A set $S \subseteq P(G')$ is a \textit{$\delta$-dominating set} of $G'$ if $\ball{S}=P(G')$.
	By definition, we immediately get the following simple property.

	\begin{claim}
	\label{prop:one_per_ball}
	Let $S$ be a $\delta$-dominating set of $G'$.
	For every $i \in [k]$ and $b \in [2]$,
	there is some $s \in S \cap \ball{t_{e_i}^{\tail_b}}$
	\end{claim}

	For $\delta$-dominating sets of size at most $k$, we get a stronger property.
	\begin{claim}
	\label{prop:one_candidate}
	Let $S\subseteq P(G')$ be a set of size at most $k$ such that $S \cup Y\cup Z$ is a $\delta$-dominating set of $G'$.  Then $S$ consists of exactly
	one candidate from each constraint gadget.
	\end{claim}
	\begin{claimproof}
		Consider an arbitrary $e \in E(G)$, and
		observe that
		$\ball{t_{e}^{\tail_1}}\subseteq P(\Gamma_e)$ by construction. It follows by \cref{prop:one_per_ball} and the fact that $Y\cap P(\Gamma_e)=Z\cap P(\Gamma_e)=\emptyset$ that $S$ contains at least one point from $P(\Gamma_e)$. As $P(\Gamma_e)\cap P(\Gamma_{e'})=\emptyset$ for all distinct $e,e' \in E(G)$ and $|E(G)|=k\geq |S|$, we obtain that $P(\Gamma_e)$ contains exactly one point of $S$ for every $e \in E(G)$. Now consider some $e \in E(G)$. As $\ball{t_{e_i}^{\tail_1}}\cup \ball{t_{e_i}^{\tail_2}}\subseteq P(\Gamma_e)$, the unique point in $S \cap P(\Gamma_e)$ 
		needs to be contained in $\ball{t_{e_i}^{\tail_1}}\cap \ball{t_{e_i}^{\tail_2}}$. As
		$\ball{t_{e_i}^{\tail_1}}\cap \ball{t_{e_i}^{\tail_2}} = A_{e}$, the statement follows.\end{claimproof}

	The following property shows the key equivalence between the existence of
	$\delta$-dominating sets of size at most $k$ in $G'$ and the satisfiability of the \binaryCSP instance in $G$.

	\begin{claim}
	\label{prop:correctness}
	The following statements are equivalent:
	\begin{enumerate}[(i)]
	\item $I$ is satisfiable,
	\item $G'$ admits a $\delta$-dominating set of size at most $k$,
	\item $P(G')$ contains a set $S$ of at most $k$ points such that $S \cup Y\cup Z$ is a $\delta$-dominating set in $G'$.
	\end{enumerate}
	\end{claim}
	\begin{claimproof}
	We first show that $(i)$ implies $(ii)$.
		Let $\mathcal{A} \colon  V(G) \to \Sigma$ be a satisfying assignment.
		Then we claim that the set of vertices
		$S \coloneqq \{\cand{e}{\phi} \mid
			e = uv \in E(G),~
			\phi(u) = \assig{u}, ~\phi(v) = \assig{v}\}$
		is a $\delta$-dominating set of $G'$ of size $k$. First observe that $S$ is well-defined because $\mathcal{A}$ is satisfying and by construction. Further observe that $|S|=k$.

		To see that $S$ is a $\delta$-dominating set of $G'$, first consider some $e\in E(G)$ and let $\phi$ be the unique element of $\Phi_e$ such that $a_e^{\phi}\in S$. Observe that $t_{e}^{\tail_1}$ and $t_{e}^{\tail_2}$
		are at a distance $2n + (\delta - 2n) =
		\delta$ from $a_e^{\phi}$.
		Next, all other candidates are at a distance at most $4n < \delta$
		from $a_e^{\phi}$. Finally, the vertices 
		$\conn{s}{e}{u}{0}$, $\conn{s}{e}{u}{1}$,
		$\conn{s}{e}{v}{0}$, and $\conn{s}{e}{v}{1}$
		lie at a distance at most $4n + \delta - 4n = \delta$ from $a_e^{\phi}$.
		The remaining points in $P(\Gamma_e)$ lie on one of the
		shortest paths between $a_e^{\phi}$ and the above vertices. It follows that all points of $G'$ that are contained in a constraint gadget are contained in $\ball{S}$.
		
		Next consider a point $p\in P(G')$ contained in a tour gadget containing the path $R_{i,i+1}$ for some $i \in [k]$. If $p$ is not contained in $P(R_{i,i+1})$, then, we have $\min\{\dist_{G'}(p,y_i),\dist_{G'}(p,z_{i+1})\}\leq 2n\leq \frac{1}{2}\delta$ by construction. If $p\in V(R_{i,i+1})$, then we have $\min\{\dist_{G'}(p,y_i),\dist_{G'}(p,z_{i+1})\}\leq \frac{1}{2}(\dist_{G'}(p,y_i)+\dist_{G'}(p,z_{i+1}))\leq \frac{1}{2}\delta$. In either case, we have $\min\{\dist_{G'}(p,y_i),\dist_{G'}(p,z_{i+1})\}\leq \frac{1}{2}\delta$. Moreover, as $A_{e_i}\cap S \neq \emptyset$ and $A_{e_{i+1}}\cap S \neq \emptyset$ hold and by construction, we have $\dist_{G'}(y_i,S)\leq 2n$ and  $\dist_{G'}(z_{i+1},S)\leq 2n$. It follows that $\dist_{G'}(p,S)\leq 2n+\frac{1}{2}\delta\leq\delta$. It follows that all points of $G'$ that are contained in a tour gadget are contained in $\ball{S}$.
		
		Finally consider a point $p$ that is contained in the middle path $M_{e_i,v,e_j}$
		for some $e_i, e_j \in E(G)$ and $v \in V(G)$ with $e_i \cap e_j = \{v\}$.
		By construction, we have $\dist_{G'}(p,\conn{t}{e_i}{v}{0})+\dist_{G'}(p,\conn{t}{e_j}{v}{1})=2n$. Next, by construction, we have $\dist_{G'}(\conn{s}{e_i}{v}{0},\conn{t}{e_i}{v}{0})=\dist_{G'}(\conn{t}{e_j}{v}{1},\conn{s}{e_j}{v}{1})=\delta-4n$. Next, for $\mu \in \{i,j\}$, let $u_\mu \in V(G)$ such that $e_\mu=u_\mu v$ and $\phi_\mu\colon \{u_\mu,v\}\to [n]$ be defined by $\phi_\mu(u_\mu)=\mathcal{A}(u_\mu)$ and $\phi_\mu(v)=\mathcal{A}(v)$. Observe that by construction, we have $a_{e_\mu}^{\phi_\mu}\in S$ for $\mu\in \{i,j\}$. Moreover, by construction, we have $\dist_{G'}(\conn{s}{e_i}{v}{0},a_{e_i}^{\phi_i})=2n+2\mathcal{A}(v)$ and $\dist_{G'}(\conn{s}{e_j}{v}{1},a_{e_j}^{\phi_j})=4n-2\mathcal{A}(v)$.
		This yields
		\begin{align*}
		\dist_{G'}(p,S)&\leq \min\{\dist_{G'}(p,a_{e_i}^{\phi_i}),\dist_{G'}(p,a_{e_j}^{\phi_j})\}\\
		&\leq \frac{1}{2}(\dist_{G'}(p,a_{e_i}^{\phi_i})+\dist_{G'}(p,a_{e_j}^{\phi_j}))\\
		&\leq \frac{1}{2}(\dist_{G'}(p,\conn{t}{e_i}{v}{0})+\dist_{G'}(\conn{s}{e_i}{v}{0},\conn{t}{e_i}{v}{0})+\dist_{G'}(\conn{s}{e_i}{v}{0},a_{e_i}^{\phi_i})+{}\\
		&\quad\dist_{G'}(x,\conn{t}{e_j}{v}{1})+\dist_{G'}(\conn{t}{e_j}{v}{1},\conn{s}{e_j}{v}{1})+\quad\dist_{G'}(\conn{s}{e_j}{v}{1},a_{e_j}^{\phi_j}))\\
		&\leq \frac{1}{2}(2n+2(\delta-4n)+(2n+2\mathcal{A}(v))+(4n-2\mathcal{A}(v)))\\
		&=\delta.
		\end{align*}
		It follows that $p \in \ball{S}$. We obtain that all vertices contained in a middle gadget of $G'$ are contained $\ball{S}$.
		
		We obtain that $S$ is  a $\delta$-dominating set of $G'$. Hence $(ii)$ follows.
		
		Next, observe that $(ii)$ clearly implies $(iii)$.
		
		We finally prove that $(iii)$ implies $(i)$.
		Let $S\subseteq P(G')$ with $|S|\leq k$ such that $S'$ is a $\delta$-dominating set in $G'$, where $S'=S \cup Y\cup Z$.
		By \cref{prop:one_candidate}, we get that $S$ consists of exactly $k$ candidates, one per constraint gadget. Hence for every $e=uv \in E(G)$, there is exactly one mapping $\phi_e\colon \{u,v\}\rightarrow [n]$ such that $a_e^{\phi_e}\in S$.

		Now consider some $u \in V(G)$ and let $e$ and $e'$ be two edges incident to $u$ in $G$. We claim that $\phi_{e}(u)=\phi_{e'}(u)$. Suppose otherwise. By symmetry, we may suppose that $\phi_{e}(u)>\phi_{e'}(u)$. Now let $P$ be the unique path in $G'$ that contains $\{\cand{e}{\phi_{e}}, \conn{s}{e}{u}{0},\conn{t}{e}{u}{0},\conn{t}{e'}{u}{1},\conn{s}{e'}{u}{1},\cand{e'}{\phi_{e'}}\}$, and all of whose other vertices are of degree 2 in $G'$.
		 We obtain that the length of $P$ is $2\delta+2(\phi_{e}(u)-\phi_{e'}(u))\geq 2\delta+2$.
		 
		  Hence there is a vertex $x\in V(P)$ with $\min\{\dist_{P}(\cand{e}{\phi_{e}},x),\dist_{P}(\cand{e'}{\phi_{e'}},x)\}\geq \delta+1$.
		 As $\dist_{G'}(\conn{t}{e}{u}{0},{\cand{e}{\phi_{e}}})\leq (\delta-4n)+(2n+2 \phi_{e}(u))\leq \delta$ and  $\dist_{G'}(\conn{t}{e'}{u}{1},\cand{e'}{\phi_{e'}})\leq (\delta-4n)+(4n-2 \phi_{e'}(u))\leq \delta$, we obtain that $x$ is contained in the interior of the middle path $M_{e,u,e'}$. As $S'$ is a $\delta$-dominating set and by the structure of $S'$, there is some $s \in S'$ such that $\dist_{G'}(x,s)\leq \delta$. Let $Q$ be a path of length at most $\delta$ in $G'$ linking $x$ and $s$. By construction and symmetry, we may suppose that $\conn{t}{e}{u}{0}\in V(Q)$. Next, as $Q$ is a path, it will either fully contain $R_e^u$ or fully contain a path linking $\conn{t}{e}{u}{0}$ and $T$ fully contained in the middle gadget corresponding to $u$, where $T=\{\conn{t}{e}{u}{1},\conn{t}{e'}{u}{0},\conn{t}{e'}{u}{1},\conn{t}{e''}{u}{0},\conn{t}{e''}{u}{1}\}$, for $e''$ being the unique edge incident to $u$ distinct from $e$ and $e'$.
		 
		  If $Q$ fully contains $R_e^u$, then, as $Q$ is a path, it fully contains the path $Q_1$ from $\conn{s}{e}{u}{0}$ to $a$ for some $a \in A_e$. We may suppose that $a$ is the last vertex of $Q$ contained in $A_e$ when following $Q$ from $a$ to $s$. If $a=\cand{e}{\phi_{e}}$, we obtain a contradiction to the assumption that $\dist_{G'}(a_e^{\phi_e},x)\geq \delta+1$.
		  Otherwise, as $a_e^{\phi_e}$ is the only vertex in $S \cap A_e$, we obtain that $Q$ fully contains a path $Q_2$ linking $A_e$ and $A_{\bar{e}}$ for some $\bar{e}\in E(G)\setminus \{e\}$ or linking $A_e$ and $Y\cup Z$.
		  By construction, we have that the length of $R_e^u$ is $\delta-4n$ and the length of $Q_1$ is at least $2n$. Moreover, by Claim \ref{pathlength}, construction and $2n\leq \delta+1$, we obtain that the length of $Q_2$ is at least $2n$. As $R_e^u,Q_1$, and $Q_2$ are vertex-disjoint except possibly for their endvertices and $Q$ contains at least one edge that is contained in none of $R_e^u$, $Q_1$, and $Q_2$, we obtain that the length of $Q$ is at least $\delta+1$, a contradiction.

		  Now suppose that $Q$ fully contains a path $Q_1$ linking $\conn{t}{e}{u}{0}$ and some $t \in T$. We may suppose that $t$ is the last vertex contained in $T$ when following $Q$ from $x$ to $s$. Then, as $Q$ is a path, by the structure of $S$ and construction, we obtain that $Q$ contains one of $\conn{R}{\bar{e}}{u}{0}$ and $\conn{R}{\bar{e}}{u}{1}$ for some $\bar{e}\in E(G)$. Hence, in particular, $Q$ contains one of $\conn{s}{\bar{e}}{u}{0}$ and $\conn{s}{\bar{e}}{u}{1}$. As $Q$ is a path ending in $S$ and by construction, we have that $Q$ contains a path $Q_2$ linking $\{\conn{s}{\bar{e}}{u}{0},\conn{s}{\bar{e}}{u}{1}\}$ and $A_{\bar{e}}$.
		  By construction, we have that the length of $Q_1$ is $2n$, the length of $\conn{R}{\bar{e}}{u}{0}$ and $\conn{R}{\bar{e}}{u}{1}$ is at $\delta-4n$, and the length of $Q_2$ is at least $2n$. As $Q_1,\conn{R}{\bar{e}}{u}{0}, \conn{R}{\bar{e}}{u}{1}$, and $Q_2$ are vertex-disjoint except possibly for their endvertices and $Q$ contains at least one edge that is contained in none of $Q_1,\conn{R}{\bar{e}}{u}{0}, \conn{R}{\bar{e}}{u}{1}$, and $Q_2$, we obtain that the length of $Q$ is at least $\delta+1$, a contradiction.  We obtain that $\phi_{e}(u)=\phi_{e'}(u)$.

		We thus obtain an assignment $\mathcal{A}$ by setting
		$\mathcal{A}(u) \coloneqq z_u$ where $z_u$ is the
		unique value such that 
		$\{z_u\} = \{\phi_e(u)~\mid~ u \in e \in E(G)\}$. It follows directly from the existence of $\cand{e}{\phi_{e}}$ for all $e \in E(G)$ that $\mathcal{A}$ is a satisfying assignment for $(V,\Sigma,\mathcal{C})$.
	\end{claimproof}
	We are now ready to prove the second main claim, which is necessary for the proof of the lemma.
\begin{claim}
	\label{prop:dom_equiv_tour}
		$G'$ has a $\delta$-dominating set of size at most $k$
		if and only if
		$G'$ has a \deltatour of length at most $k(\delta+1)$.
	\end{claim}
	\begin{claimproof}
		First let $S$ be a $\delta$-dominating set in $G'$ of size at most $k$.
		By \cref{prop:one_candidate}, for every $e=v_iv_j \in E(G)$ with $i<j$, there is exactly one mapping $\phi_e\colon\{v_i,v_j\}\rightarrow [n]$ such that $a_e^{\phi_e}\in S$ and $(\phi(v_i),\phi(v_j))\in C_e$. Observe that by Claim \ref{pathlength}, for $i \in [k]$, there is a path $Q_i$ of length $\delta+1$ linking $a_{e_i}^{\phi_{e_i}}$ and $a_{e_{i+1}}^{\phi_{e_{i+1}}}$.  Now let $T$ be the tour obtained by concatenating $T_1,\ldots,T_k$ in that order. By construction, we have $\len(T)=k(\delta+1)$. Moreover, as $T$ stops at all vertices of $S$ and $S$ is a $\delta$-dominating set in $G'$, we obtain that $T$ is a $\delta$-tour in $G$.

		Now suppose that $G'$ admits a $\delta$-tour $T$ of length at most $k(\delta+1)$.
		By definition, we obtain that $P(T)$ is a $\delta$-dominating set of $G'$. It follows from \cref{prop:one_per_ball} and \cref{nearstop} that for $i \in [k]$,
		there is a stopping point of $T$ in $\ball{t_{e_i}^{\tail_1}}$.
		As $\ball{t_{e_i}^{\tail_1}}\cap \ball{t_{e_j}^{\tail_1}}=\emptyset$ for all $i,j \in[k]$ with $i \neq j$, $T$ is a tour, and by construction, we obtain that for
		$e \in E(G)$, we have that $T$ stops at $a_{e}^{\phi_{e}}$ for some $\phi_e \in \Phi_e$.
		Let $A_T=\{a_{e}^{\phi_{e}}:e \in E(G)\}$ and let $a_1,\ldots,a_k$ be the enumeration of $A_T$ in the order the vertices are stopped at when following $T$, choosing an arbitrary occurence of an element of $A_T$ if $T$ stops at this point multiple times.
		Further, let $\tau:[k]\rightarrow [k]$ be defined so that $a_i \in A_{e_{\tau(i)}}$ for $i \in [k]$. Observe that $\tau$ is a bijection. Without loss of generalty, we may suppose that $\tau(1)=1$. For $i \in [k-1]$, let $W_i$ be the subwalk of $T$ from $a_i$ to $a_{i+1}$ and let $W_k$ be the subwalk of $T$ from $a_k$ to $a_1$. By Claim \ref{pathlength} and the choice of $T$, we obtain $k(\delta+1)\geq \len(T)=\sum_{i \in [k]}\len(W_i)\geq k(\delta+1)$. Hence equality holds throughout yielding that $|\tau(i)-\tau(i+1)|=1$ or $\{\tau(i),\tau(i+1)\}=\{1,k\}$ for $i \in [k-1]$ and $|\tau(k)-\tau(1)|=1$ or $\{\tau(k),\tau(1)\}=\{1,k\}$ by Claim \ref{pathlength}. As $\tau$ is a bijection and by symmetry, we may suppose that $\tau(i)=i$ for $i\in [k]$. It now follows from Claim \ref{pathlength} that for $i \in [k]$, we have that $W_i$ is the unique path from $a_{e_i}^{\phi_{e_i}}$ to $a_{e_{i+1}}^{\phi_{e_{i+1}}}$ that is fully contained in a tour gadget.
		
We now conclude that $A_T\cup Y\cup Z$ is a $\delta$-dominating set in $G'$. To this end, let $p$ be an arbitrary point in $P(G')$. As $T$ is a $\delta$-tour, there is a walk $Q$ from $p$ to $P(T)$ of length at most $\delta$. By the above characterization of $T$, we have that $T$ is an integral tour and all stopping points of $T$ which are not contained in $A_T\cup Y\cup Z$ are vertices of degree 2 in $G'$ both of whose incident edges are traversed exactly once by $T$. It follows that the endvertex of $Q$ in $P(T)$ is contained in $A_T\cup Y\cup Z$. Hence $A_T\cup Y\cup Z$ is a $\delta$-dominating set in $G'$. It follows by Claim \ref{prop:correctness} that there is a $\delta$-dominating set in $G'$ of size at most $k$.

	\end{claimproof}

Now the lemma follows directly from \cref{prop:correctness} and \cref{prop:dom_equiv_tour}.
\end{proof}

With \cref{lem:csp_to_dtour}, the main theorem of this section follows easily.

\ThmParamKLbLargeDelta*\label\thisthm
\begin{proof}
	Assume that for some fixed $k \geq k_0'$, there is an
		algorithm as in the above statement
		for $\alpha' \coloneqq \alpha/90$ and $k_0' \coloneqq 27 k_0$, where
		$\alpha$ and $k_0$ are the absolute constants implied by
		\cref{thm:csp_hardness}.
		Given an instance of cubic \binaryCSP with
		$k_{\CSPsubscript} \coloneqq 3\floor{\frac{k}{27}}$
	constraints over a domain of size $n$,
	apply \cref{lem:csp_to_dtour} with $k_{\CSPsubscript}$ and $r \coloneqq k \bmod 27$
	to obtain $G'$ where $n' = (9k_{\CSPsubscript} + r) \delta$ and $n' = \Oh(n^5)$.

	Observe that $n'/\delta = \ceil{\frac{n'}{\delta}}=
		27 \floor{\frac{k}{27}} + k \bmod 27 = k$.

	Thus, a hypothetical algorithm deciding $(G',k(\delta+1))$
		yields an algorithm deciding such \binaryCSP instances in time
		$\Oh((n')^{\alpha' \cdot k/\log{k}}) =
		\Oh(n^{\alpha k/(18\log{k})}) = \Oh(n^{\alpha k_{\CSPsubscript}/\log{k_{\CSPsubscript}}})$,
		contradicting \cref{thm:csp_hardness}.
	Moreover, as the starting problem is \wone-hard parameterized
		by the number of constraints $k$,
		we obtain \wone-hardness in our setting for parameterization by $n'/\delta$.
\end{proof}
As a side product of the proof of \cref{lem:csp_to_dtour}, we obtain the following result which might be of independent interest.
\begin{corollary}
\label{cor:d_domset_hardness}
There are constants $\alpha>0$ and $k_0$ such that,
	unless \ETH fails,
	for every $k \geq k_0$,
	there is no algorithm that, given an $n$-vertex graph,
	computes a minimum-cardinality $\delta$-dominating set
	in $\Oh(n^{\alpha k/\log{k}})$ time where $\delta = n/k$.
Moreover, the problem is \wone-hard parameterized by~$k$. 
\end{corollary}
\begin{proof}
	As in \cref{thm:param:k:lb:largedelta}, the proof follows immediately
	from \cref{lem:csp_to_dtour} and \cref{prop:dom_equiv_tour}.
\end{proof}

\tikzstyle{wiggly} = [decorate, decoration={snake, amplitude=.4mm, segment length=2mm, post length=0mm}]
\tikzstyle{solid}=[fill=black, draw=black, shape=circle]
\tikzstyle{solid_small}=[fill=black, draw=black, shape=circle, inner sep=0pt, minimum size=4pt]

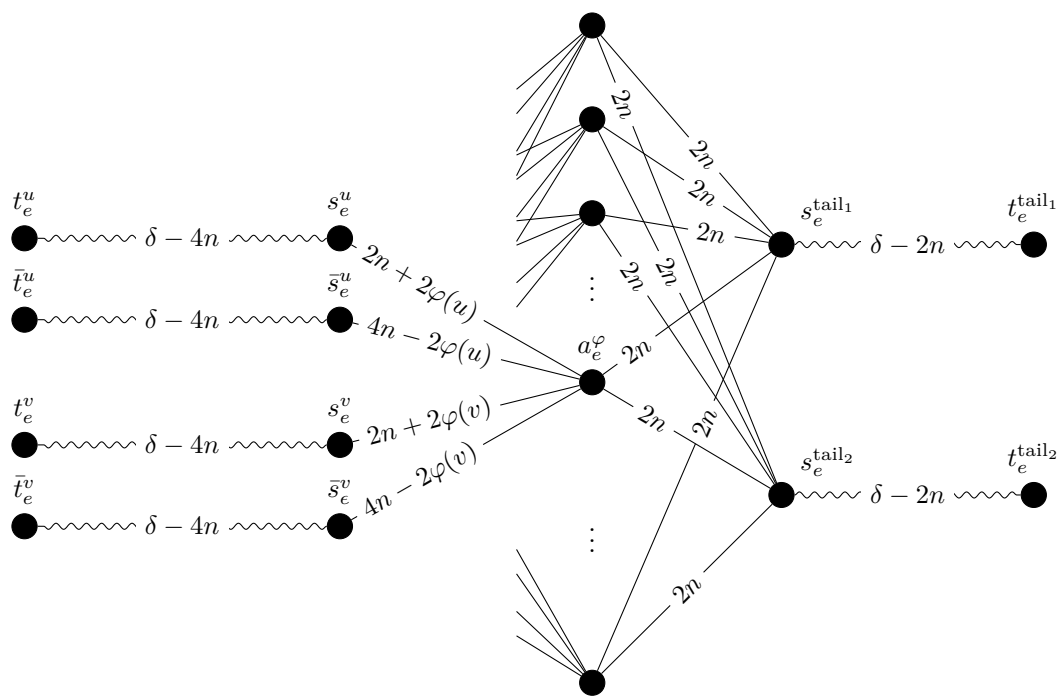
\begin{figure}
	\centering
\begin{tikzpicture}[scale=.83]
	\node [style=solid, label=above right:$s_{e}^{\tail_1}$] (tail1_start) at (3.5, 6.5) {};
	\node [style=solid, label=above:$t_{e}^{\tail_1}$] (tail1_end) at (7.5, 6.5) {};
	\node [style=solid, label=above right:$s_{e}^{\tail_2}$] (tail2_start) at (3.5, 2.5) {};
	\node [style=solid, label=above:$t_{e}^{\tail_2}$] (tail2_end) at (7.5, 2.5) {};
	\node [style=solid] (4) at (0.5, 10.0) {};
	\node [style=solid] (5) at (0.5, 8.5) {};
	\node [style=solid] (6) at (0.5, 7.0) {};
	\node [style=solid] (7) at (0.5, -0.5) {};
	\node [style=solid, label=above:$\cand{e}{\phi}$] (c_xy) at (0.5, 4.3) {};
	\node [style=solid, label=above:$\conn{s}{e}{u}{0}$] (x_start) at (-3.5, 6.6) {};
	\node [style=solid, label=above:$\conn{s}{e}{u}{1}$] (xbar_start) at (-3.5, 5.3) {};
	\node [style=solid, label=above:$\conn{s}{e}{v}{0}$] (y_start) at (-3.5, 3.3) {};
	\node [style=solid, label=above:$\conn{s}{e}{v}{1}$] (ybar_start) at (-3.5, 2.0) {};
	\node [style=solid, label=above:$\conn{t}{e}{u}{0}$] (x_end) at (-8.5, 6.6) {};
	\node [style=solid, label=above:$\conn{t}{e}{u}{1}$] (xbar_end) at (-8.5, 5.3) {};
	\node [style=solid, label=above:$\conn{t}{e}{v}{0}$] (y_end) at (-8.5, 3.3) {};
	\node [style=solid, label=above:$\conn{t}{e}{v}{1}$] (ybar_end) at (-8.5, 2.0) {};

	\draw [style=wiggly] (tail1_start) -- (tail1_end) node[midway,fill=white]{$\delta-2n$};
	\draw [style=wiggly] (tail2_start) -- (tail2_end) node[midway,fill=white]{$\delta-2n$};

	\draw (tail1_start) -- (4) node[pos=.4,fill=white,sloped] {$2n$};
	\draw (tail1_start) -- (5) node[pos=.4,fill=white,sloped] {$2n$};
	\draw (6) -- (tail1_start) node[pos=.65,fill=white,sloped] {$2n$};
	\draw (tail1_start) -- (c_xy) node[pos=.8,fill=white,sloped] {$2n$};
	\draw (tail1_start) -- (7) node[pos=.4,fill=white,sloped] {$2n$};
	\draw (tail2_start) -- (4) node[pos=.85,fill=white,sloped] {$2n$};
	\draw (tail2_start) -- (5) node[pos=.6,fill=white,sloped] {$2n$};
	\draw (tail2_start) -- (6) node[pos=.8,fill=white,sloped] {$2n$};
	\draw (tail2_start) -- (c_xy) node[pos=.7,fill=white,sloped] {$2n$};
	\draw (tail2_start) -- (7) node[pos=.5,fill=white,sloped] {$2n$};

	\draw (c_xy) -- (x_start) node[pos=.7,fill=white,rotate=-30] {$2n + 2\phi(u)$};
	\draw (c_xy) -- (xbar_start) node[pos=.66,fill=white,rotate=-15] {$4n - 2\phi(u)$};
	\draw (c_xy) -- (y_start) node[pos=.66,fill=white,rotate=15] {$2n + 2\phi(v)$};
	\draw (c_xy) -- (ybar_start) node[pos=.7,fill=white,,rotate=30] {$4n - 2\phi(v)$};

	\foreach \name in {4, 5, 6, 7} {
		\draw (\name) -- ($(\name)!0.3!(x_start)$);
		\draw (\name) -- ($(\name)!0.3!(xbar_start)$);
		\draw (\name) -- ($(\name)!0.3!(y_start)$);
		\draw (\name) -- ($(\name)!0.3!(ybar_start)$);
	}

	\draw [style=wiggly](x_start) -- (x_end) node[midway,fill=white]{$\delta-4n$};
	\draw [style=wiggly] (xbar_start) -- (xbar_end) node[midway,fill=white]{$\delta-4n$};
	\draw [style=wiggly] (y_start) -- (y_end) node[midway,fill=white]{$\delta-4n$};
	\draw [style=wiggly] (ybar_start) -- (ybar_end) node[midway,fill=white]{$\delta-4n$};

	\node at ($(c_xy)!0.5!(7)$) {$\vdots$};
	\node at ($(c_xy)!0.6!(6)$) {$\vdots$};
\end{tikzpicture}
\caption{Constraint gadget $\Gamma_{e}$ for $e = uv \in E(G)$ as defined in the proof of \cref{thm:csp_hardness}.} 
\label{fig:constraint_gadget}
\end{figure}

\tikzset{pics/constraint_gadget_module/.style n args={8}{
	code = {
		\node (0) at (-1, -1) {}; %
		\node (2) at (1, -1) {}; %
		\node (3) at (1, 1) {}; %
		\node (4) at (-1, 1) {}; %
		\draw [in=90, out=-90] (0.center) -- (4.center);
		\draw (0.center) -- (2.center);
		\draw (2.center) -- (3.center);
		\draw (3.center) -- (4.center);

		\ifthenelse{\equal{#4}{}}{
			\node (22) at ($($(0)!0.1!(4)$)!0.1!($(3)!0.1!(2)$)$) {$\Gamma_{#1,#2}$};
		}{
			\node (22) at ($($(0)!1.25!(4)$)!0.1!($(3)!0.3!(2)$)$) {$\Gamma_{#4}$};
		}

		\node (5) at (0, -1) {}; %
		\node (6) at (0, 1) {}; %
		\foreach \x in {1,...,9}{
			\pgfmathsetmacro{\position}{\x/10}
			\node[circle,draw,fill=black,inner sep=1pt] (cand\x) at ($ (5.center)!\position!(6.center) $) {};
		}

		\ifthenelse{\equal{#3}{1}}{
			\def\rot{0}
		}{
			\def\rot{270}
		}

		\ifthenelse{\not\equal{#1}{}}{
			\begin{scope}[rotate=\rot]
				\node (12) at (1, 0.5){}; %
				\node[style=solid_small] (16) at (2, 0.5) {};  %

				\node (17) at (1, -0.5) {}; %
				\node[style=solid_small] (18) at (2, -0.5) {};  %

				\draw[style=wiggly] (12.center) -- (16.center) node[pos=1,#5]{$t_{#4}^{#1}$};
				\draw[style=wiggly] (17.center) -- (18.center) node[pos=1,#6]{$\bar{t}_{#4}^{#1}$};
			\end{scope}
		}{}
		\ifthenelse{\not\equal{#2}{}}{
			\begin{scope}[rotate=\rot]
				\node (9) at (-1, 0.5) {}; %
				\node[style=solid_small] (13) at (-2, 0.5) {};   %

				\node (19) at (-1, -0.5) {}; %
				\node[style=solid_small] (20) at (-2, -0.5) {}; %

				\draw[style=wiggly] (9.center) -- (13.center) node[pos=1,#7]{$t_{#4}^{#2}$};
				\draw[style=wiggly] (19.center) -- (20.center) node[pos=1,#8]{$\bar{t}_{#4}^{#2}$};
			\end{scope}
		}{}

		\ifthenelse{\equal{#3}{1}}{
		}{
			\node[circle,draw,fill=black,inner sep=1pt] (s) at (1.5, 0) {};
			\node[circle,draw,fill=black,inner sep=1pt] (e) at (-1.5, 0) {};
		}
	}
}}

\begin{figure}
	\centering
	\begin{tikzpicture}
		\pic[rotate=270] (XW) at (3,3.75) {constraint_gadget_module={v}{u}{1}{e_1}{left}{left}{above}{above}};
		\pic[rotate=180] (XY) at (6,0) {constraint_gadget_module={v}{w}{1}{e_2}{below}{below}{above}{above}};
		\pic (XZ) at (0,0) {constraint_gadget_module={v}{x}{1}{e_3}{below}{below}{above}{above}};

		\draw[line width=1pt] (XY16.center) -- (XZ18.center) node {};
		\draw[line width=1pt] (XY18.center) -- (XZ16.center) node {};

		\draw[line width=1pt] (XW16.center) -- (XZ18.center) node {};
		\draw[line width=1pt] (XW18.center) -- (XZ16.center) node {};

		\draw[line width=1pt] (XW16.center) -- (XY18.center) node {};
		\draw[line width=1pt] (XW18.center) -- (XY16.center) node {};
	\end{tikzpicture}
	\caption{An example of the connection between gadgets
		$\Gamma_{e_1}$, $\Gamma_{e_2}$, and $\Gamma_{e_3}$ sharing a variable $v$.
		Oscillating lines are paths of length $\delta-4n$, and bold lines are paths of length $2n$.
		The vertices within gadgets indicate candidate vertices.}
	\label{fig:gadget_connectivity}
\end{figure}
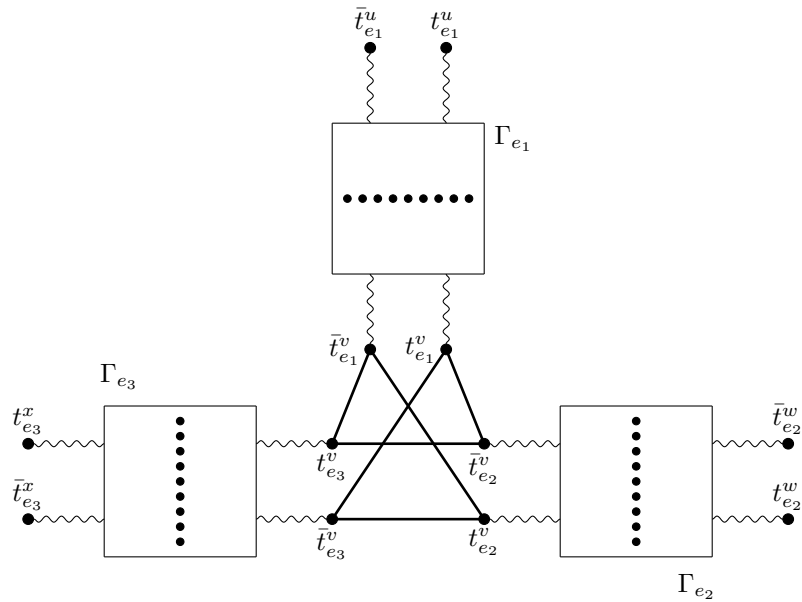

\begin{figure}
	\centering
	\begin{tikzpicture}[scale=.8]
		\pic[rotate=90] (G_0_) at (0.5,-0.25) {constraint_gadget_module={}{}{0}{e_1}{}{}{}{}};
		\pic[rotate=35] (G_1_) at (2.5,3.5) {constraint_gadget_module={}{}{0}{e_2}{}{}{}{}};
		\pic (G_2_) at (7,5) {constraint_gadget_module={v}{}{0}{e_3}{right}{left}{}{}};
		\pic[rotate=-45]  (G_3_) at (11.5,3.5) {constraint_gadget_module={}{}{0}{e_4}{}{}{}{}};
		\pic[rotate=270] (G_4_) at (13.5,-0.25) {constraint_gadget_module={}{}{0}{e_5}{}{}{}{}};
		\pic[rotate=225] (G_5_) at (11,-4) {constraint_gadget_module={v}{}{0}{e_6}{below left}{above right}{}{}};
		\pic[rotate=135] (G_6_) at (3,-4) {constraint_gadget_module={v}{}{0}{e_7}{above left}{below right}{}{}};

		\draw[style=wiggly,line width=1pt] (G_0_s.center) -- (G_1_e.center) node {};
		\draw[style=wiggly,line width=1pt] (G_1_s.center) -- (G_2_e.center) node {};
		\draw[style=wiggly,line width=1pt] (G_2_s.center) -- (G_3_e.center) node {};
		\draw[style=wiggly,line width=1pt] (G_3_s.center) -- (G_4_e.center) node {};
		\draw[style=wiggly,line width=1pt] (G_4_s.center) -- (G_5_e.center) node {};
		\draw[style=wiggly,line width=1pt] (G_5_s.center) -- (G_6_e.center) node[midway,fill=white] {$~~~~\dots\dots~~~~$};
		\draw[style=wiggly,line width=1pt] (G_6_s.center) -- (G_0_e.center) node {};

		\foreach \x in {0,...,6}{
			\foreach \y in {1,...,9}{
				\draw[line width=1pt] (G_\x_s.center) -- (G_\x_cand\y.center) node {};
				\draw[line width=1pt] (G_\x_e.center) -- (G_\x_cand\y.center) node {};
			}
		}

		\draw[line width=1pt] (G_2_16.center) -- (G_5_18.center) node {};
		\draw[line width=1pt] (G_2_18.center) -- (G_5_16.center) node {};

		\draw[line width=1pt] (G_6_16.center) -- (G_5_18.center) node {};
		\draw[line width=1pt] (G_6_18.center) -- (G_5_16.center) node {};

		\draw[line width=1pt] (G_6_16.center) -- (G_2_18.center) node {};
		\draw[line width=1pt] (G_6_18.center) -- (G_2_16.center) node {};
	\end{tikzpicture}
	\caption{The graph~$G'$. The figure shows an example of how three
		gadgets sharing a variable connect.
		The remaining connectors and middle paths are omitted for
		clarity. Bold oscillating lines are paths of length $\delta-4n+1$, and
		bold lines are paths of length $2n$.}
	\label{fig:graph_gprime_large_delta}
\end{figure}
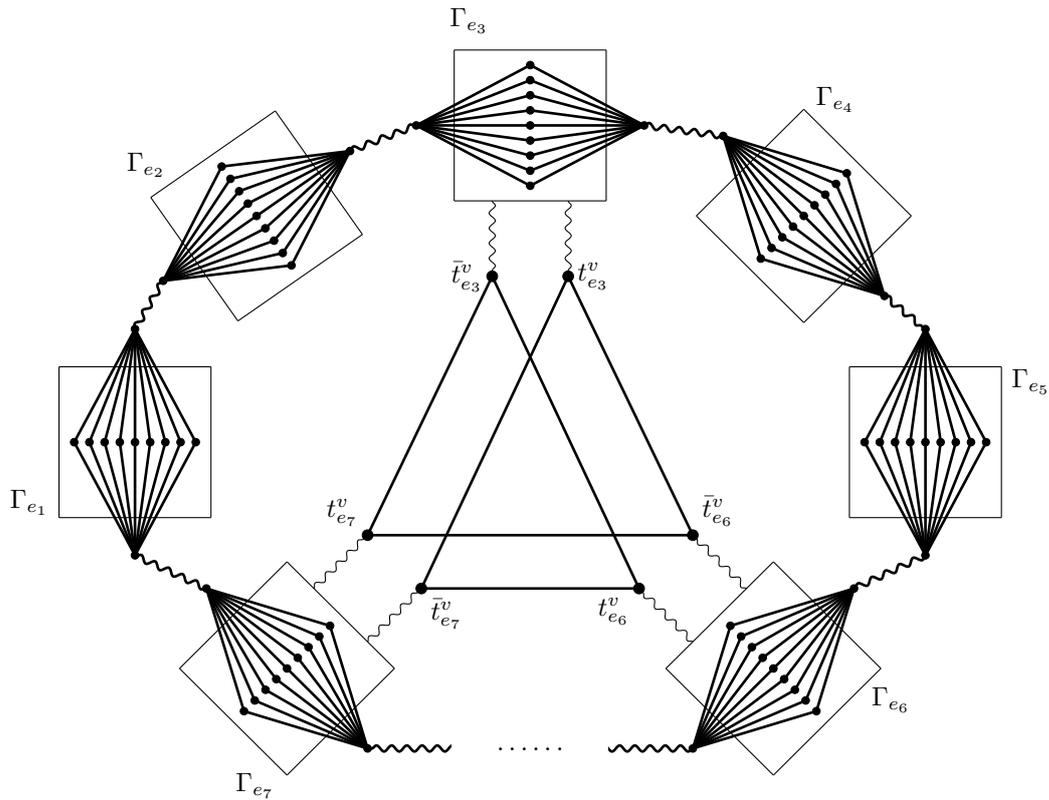
\clearpage

\bibliography{lit}

\begin{thebibliography}{10}

\bibitem{AlimontiK97}
Paola Alimonti and Viggo Kann.
\newblock Hardness of approximating problems on cubic graphs.
\newblock In {\em {CIAC}}, volume 1203 of {\em Lecture Notes in Computer Science}, pages 288--298. Springer, 1997.

\bibitem{CLRS}
Thomas~H. Cormen, Charles~E. Leiserson, Ronald~L. Rivest, and Clifford Stein.
\newblock {\em Introduction to {{Algorithms}}, Fourth Edition}.
\newblock {MIT Press}, 2022.

\bibitem{Dearing1974}
Perino~M. Dearing and Richard~Lane Francis.
\newblock A minimax location problem on a network.
\newblock {\em Transportation Science}, 8(4):333--343, 1974.

\bibitem{DinurS14}
Irit Dinur and David Steurer.
\newblock Analytical approach to parallel repetition.
\newblock In {\em {STOC}}, pages 624--633. {ACM}, 2014.

\bibitem{EngebretsenK01}
Lars Engebretsen and Marek Karpinski.
\newblock Approximation hardness of {TSP} with bounded metrics.
\newblock In Fernando Orejas, Paul~G. Spirakis, and Jan van Leeuwen, editors, {\em Automata, Languages and Programming, 28th International Colloquium, {ICALP} 2001, Crete, Greece, July 8-12, 2001, Proceedings}, volume 2076 of {\em Lecture Notes in Computer Science}, pages 201--212. Springer, 2001.
\newblock \href {https://doi.org/10.1007/3-540-48224-5\_17} {\path{doi:10.1007/3-540-48224-5\_17}}.

\bibitem{EngebretsenK06}
Lars Engebretsen and Marek Karpinski.
\newblock {TSP} with bounded metrics.
\newblock {\em Journal of Computer and System Sciences}, 72(4):509--546, 2006.
\newblock URL: \url{https://www.sciencedirect.com/science/article/pii/S0022000005001285}, \href {https://doi.org/10.1016/j.jcss.2005.12.001} {\path{doi:10.1016/j.jcss.2005.12.001}}.

\bibitem{FreiGHHM24}
Fabian Frei, Ahmed Ghazy, Tim~A. Hartmann, Florian H{\"{o}}rsch, and D{\'{a}}niel Marx.
\newblock From chinese postman to salesman and beyond: Shortest tour {\(\delta\)}-covering all points on all edges.
\newblock {\em CoRR}, abs/2410.10613, 2024.
\newblock URL: \url{https://doi.org/10.48550/arXiv.2410.10613}, \href {https://arxiv.org/abs/2410.10613} {\path{arXiv:2410.10613}}, \href {https://doi.org/10.48550/ARXIV.2410.10613} {\path{doi:10.48550/ARXIV.2410.10613}}.

\bibitem{HartmannJanssen2024}
Tim~A. Hartmann and Tom Jan{\ss}en.
\newblock Approximating $\delta$-covering.
\newblock In Marcin Bie{\'{n}}kowski and Matthias Englert, editors, {\em Approximation and Online Algorithms. {WAOA} 2024}, Lecture Notes in Computer Science, pages 61--75. Springer, 2025.

\bibitem{HartmannLW22}
Tim~A. Hartmann, Stefan Lendl, and Gerhard~J. Woeginger.
\newblock Continuous facility location on graphs.
\newblock {\em Math. Program.}, 192(1):207--227, 2022.
\newblock URL: \url{https://doi.org/10.1007/s10107-021-01646-x}, \href {https://doi.org/10.1007/S10107-021-01646-X} {\path{doi:10.1007/S10107-021-01646-X}}.

\bibitem{Karpinski15FCT}
Marek Karpinski.
\newblock Towards better inapproximability bounds for {TSP:} {A} challenge of global dependencies.
\newblock In Adrian Kosowski and Igor Walukiewicz, editors, {\em Proceedings of the 20th International Symposium on Fundamentals of Computation Theory ({FCT} 2015)}, volume 9210 of {\em Lecture Notes in Computer Science}, pages 3--11. Springer, 2015.
\newblock \href {https://doi.org/10.1007/978-3-319-22177-9\_1} {\path{doi:10.1007/978-3-319-22177-9\_1}}.

\bibitem{KarpinskiS13website}
Marek Karpinski and Richard Schmied.
\newblock Improved inapproximability results for the shortest superstring and the bounded metric {TSP}.
\newblock \url{https://theory.cs.uni-bonn.de/ftp/reports/cs-reports/2013/85339-CS.pdf}.
\newblock Accessed: 2024-04-01.

\bibitem{KarpinskiS12}
Marek Karpinski and Richard Schmied.
\newblock On approximation lower bounds for {TSP} with bounded metrics.
\newblock {\em CoRR}, abs/1201.5821, 2012.
\newblock URL: \url{http://arxiv.org/abs/1201.5821}, \href {https://arxiv.org/abs/1201.5821} {\path{arXiv:1201.5821}}.

\bibitem{KarpinskiS13}
Marek Karpinski and Richard Schmied.
\newblock Improved inapproximability results for the shortest superstring and related problems.
\newblock In Anthony Wirth, editor, {\em Nineteenth Computing: The Australasian Theory Symposium, {CATS} 2013, Adelaide, Australia, February 2013}, volume 141 of {\em {CRPIT}}, pages 27--36. Australian Computer Society, 2013.
\newblock URL: \url{http://crpit.scem.westernsydney.edu.au/abstracts/CRPITV141Karpinski.html}.

\bibitem{KarpinskiS15}
Marek Karpinski and Richard Schmied.
\newblock Approximation hardness of graphic {TSP} on cubic graphs.
\newblock {\em RAIRO -- Operations Research}, 49(4):651--668, 2015.
\newblock URL: \url{http://www.numdam.org/articles/10.1051/ro/2014062/}, \href {https://doi.org/10.1051/ro/2014062} {\path{doi:10.1051/ro/2014062}}.

\bibitem{KarthikLM19}
{Karthik {C. S.}}, Bundit Laekhanukit, and Pasin Manurangsi.
\newblock On the parameterized complexity of approximating dominating set.
\newblock {\em J. {ACM}}, 66(5):33:1--33:38, 2019.
\newblock \href {https://doi.org/10.1145/3325116} {\path{doi:10.1145/3325116}}.

\bibitem{KarthikMPS23}
{Karthik {C. S.}}, D{\'{a}}niel Marx, Marcin Pilipczuk, and U{\'{e}}verton~S. Souza.
\newblock Conditional lower bounds for sparse parameterized 2-{CSP}: {A} streamlined proof.
\newblock {\em CoRR}, abs/2311.05913, 2023.

\bibitem{Marx07}
D{\'{a}}niel Marx.
\newblock Can you beat treewidth?
\newblock In {\em {FOCS}}, pages 169--179. {IEEE} Computer Society, 2007.

\bibitem{RamanS08}
Venkatesh Raman and Saket Saurabh.
\newblock Short cycles make \emph{W}-hard problems hard: {FPT} algorithms for \emph{W}-hard problems in graphs with no short cycles.
\newblock {\em Algorithmica}, 52(2):203--225, 2008.

\bibitem{SahniG76}
Sartaj Sahni and Teofilo~F. Gonzalez.
\newblock P-complete approximation problems.
\newblock {\em J. {ACM}}, 23(3):555--565, 1976.
\newblock \href {https://doi.org/10.1145/321958.321975} {\path{doi:10.1145/321958.321975}}.

\bibitem{schrijver-book}
Alexander Schrijver.
\newblock {\em Combinatorial Optimization -- Polyhedra and Efficiency}.
\newblock Springer, 2003.

\bibitem{Shier1977}
Douglas~R. Shier.
\newblock A min-max theorem for $p$-center problems on a tree.
\newblock {\em Transportation Science}, 11(3):243--252, 1977.
\newblock URL: \url{http://www.jstor.org/stable/25767877}.

\end{thebibliography}

\tableofcontents
\label{toc}

\end{document}